%% file: main.tex
\documentclass[11pt]{article}
\pdfoutput=1
\newif\iflipics\lipicsfalse
\usepackage[letterpaper, margin=1in]{geometry}
\usepackage{mathptmx}
\DeclareMathAlphabet{\mathcal}{OMS}{cmsy}{m}{n}

\usepackage{amsmath}
\usepackage{amsfonts}
\usepackage{amssymb}
\usepackage{amsthm}
\usepackage{xspace}
\usepackage[dvipsnames]{xcolor}
\usepackage{tabularx}
\usepackage{ragged2e}
\usepackage{nicefrac}
\usepackage{bbm}
\usepackage{caption}
\usepackage{graphicx}
\usepackage{multirow}
\usepackage{booktabs}
\usepackage{enumitem}
\usepackage{soul}
\usepackage[most]{tcolorbox}
\usepackage{colortbl}
\usepackage{hhline}

\usepackage{thmtools}
\usepackage{thm-restate}

\input{macros}

\usepackage[pdfusetitle,colorlinks=true,linkcolor=blue,citecolor=BrickRed,bookmarksopen]{hyperref}

\definecolor{darkgreen}{rgb}{0,0.4,0}
\usepackage{algorithm}
\usepackage[commentColor=darkgreen]{algpseudocodex}
\newcommand{\LineComment}{\LComment}
\algnewcommand{\IfThenElse}[3]{
  \State \algorithmicif\ #1\ \algorithmicthen\ #2\ \algorithmicelse\ #3}
\algdef{SE}[SUBALG]{Indent}{EndIndent}{}{\algorithmicend\ }%
\makeatletter
\renewcommand{\ALG@name}{Listing}
\makeatother

\definecolor{todoboxcolor}{RGB}{255,233,254}

\usepackage{cleveref}
\crefname{listing}{listing}{listings}
\crefname{equation}{eq.}{eqs.}

\newtheorem{theorem}{Theorem}[section]
\newtheorem{lemma}[theorem]{Lemma}

\theoremstyle{definition}
\newtheorem{definition}[theorem]{Definition}

\newtheorem{result}{Result}
\theoremstyle{remark}
\newtheorem*{remark}{Remark}
\makeatletter
\newcommand\squeezepar{\@startsection{paragraph}{4}%
{\z@}{1.5ex \@plus1ex \@minus.2ex}{-1em}{\normalfont\normalsize\bfseries}}
\newcommand{\mypara}[1]{\squeezepar{{#1}.}}
\renewcommand{\paragraph}[1]{\squeezepar{{#1}.}}
\makeatother

\tcolorboxenvironment{result}{enhanced, colback=black!10!white, frame hidden, 
  breakable, boxsep=1pt, left=3pt, right=3pt, before skip=4pt, after skip=6pt}
\title{%
Finding missing items requires strong forms of randomness%
\thanks{This work was supported in part by the National Science Foundation under award 2006589.}}

\author{Amit Chakrabarti\thanks{Department of Computer Science, Dartmouth College, Hanover NH 03755, USA.} 
\and
Manuel Stoeckl$^\fnsymbol{footnote}$}

\date{}

\begin{document}

\maketitle
\thispagestyle{empty}

\input{abstract}

\newpage
\addtocounter{page}{-1}

\input{main-content}

\section*{Acknowledgements}

We thank Jonathan Conroy for helpful feedback on a earlier draft of this paper.
We also thank Omri Ben-Eliezer for the sampling versus sketching
interpretation of a portion of our results (noted after \Cref{res:rt-ub}).

\bibliographystyle{alpha}
\bibliography{refs}

\input{appendix}

\end{document}

%% file: macros.tex
\newcommand{\EE}{\mathbb{E}}

\newcommand{\NN}{\mathbb{N}}
\newcommand{\RR}{\mathbb{R}}

\newcommand{\indic}{\mathbbm{1}}
\newcommand{\cA}{\mathcal{A}}
\newcommand{\cB}{\mathcal{B}}

\newcommand{\cD}{\mathcal{D}}

\newcommand{\cI}{\mathcal{I}}

\newcommand{\cO}{\mathcal{O}}

\newcommand{\cR}{\mathcal{R}}

\newcommand{\tO}{\widetilde{O}}
\newcommand{\tP}{\widetilde{P}}
\newcommand{\tQ}{\widetilde{Q}}

\newcommand{\tOmega}{\widetilde{\Omega}}

\newcommand{\hw}{\hat{w}}
\newcommand{\hb}{\hat{b}}
\newcommand{\hP}{\hat{P}}
\newcommand{\hf}{\hat{f}}

\newcommand{\ceil}[1]{{\left\lceil{#1}\right\rceil}}
\newcommand{\floor}[1]{{\left\lfloor{#1}\right\rfloor}}

\DeclareSymbolFont{newfont}{OML}{cmm}{m}{it}
\DeclareMathSymbol{\emptystream}{3}{newfont}{15}
\let\epsilon\varepsilon
\renewcommand{\emptyset}{\varnothing}

\newcommand{\avoid}{\textsc{avoid}\xspace}
\newcommand{\mif}{\textsc{mif}\xspace}
\newcommand{\miflong}{\textsc{MissingItemFinding}\xspace}
\newcommand{\fco}{\textsc{fco}\xspace}
\newcommand{\fao}{\textsc{fao}\xspace}

\newcommand{\poly}{\mathrm{poly}\xspace}
\newcommand{\polylog}{\mathrm{polylog}\xspace}
\newcommand{\Algs}{\textsc{Algs}\xspace}

\newcommand{\unsafesc}{\textsc{unsafe}\xspace}

\newcommand{\emptysc}{\textsc{empty}\xspace}
\newcommand{\timeoutsc}{\textsc{timeout}\xspace}
\newcommand{\bigsc}{\textsc{big}\xspace}
\newcommand{\abort}{\textsc{abort}\xspace}
\newcommand{\sort}{\textsc{sort}\xspace}
\newcommand{\abortsc}{\textsc{abort}\xspace}
\newcommand{\errorsc}{\textsc{error}\xspace}

\newcommand{\rl}{\mathbf{\ell}}             
\newcommand{\stream}{\tau}                  
\newcommand{\concat}{\mathbin{\circ}}       
\newcommand{\trans}{T}                      
\newcommand{\SSD}[2]{\textsc{seqs}(#1,#2)}  
\newcommand{\Outs}{\textsc{Outs}\xspace}    
\newcommand{\Mstk}{\Delta}                  

\newcommand{\gobble}[1]{}

%% file: abstract.tex
\begin{abstract}

Adversarially robust streaming algorithms are required to process a stream of
elements and produce correct outputs, even when each stream element can be
chosen as a function of earlier algorithm outputs. As with classic streaming
algorithms, which must only be correct for the worst-case fixed stream,
adversarially robust algorithms with access to randomness can use
significantly less space than deterministic algorithms. We prove that for the
Missing Item Finding problem in streaming, the space complexity also
significantly depends on how adversarially robust algorithms are permitted to
use randomness. (In contrast, the space complexity of classic streaming
algorithms does not depend as strongly on the way randomness is used.)

For Missing Item Finding on streams of length $\ell$ with elements in
$\{1,\ldots,n\}$, and $\le 1/\text{poly}(\ell)$ error, we show that when $\ell
= O(2^{\sqrt{\log n}})$, ``random seed'' adversarially robust algorithms,
which only use randomness at initialization, require $\ell^{\Omega(1)}$ bits
of space, while ``random tape'' adversarially robust algorithms, which may
make random decisions at any time, may use $O(\text{polylog}(\ell))$ space.
When $\ell$ is between $n^{\Omega(1)}$ and $O(\sqrt{n})$, ``random tape''
adversarially robust algorithms need $\ell^{\Omega(1)}$ space, while ``random
oracle'' adversarially robust algorithms, which can read from a long random
string for free, may use $O(\text{polylog}(\ell))$ space. The space lower
bound for the ``random seed'' case follows, by a reduction given in prior
work, from a lower bound for pseudo-deterministic streaming algorithms given
in this paper.

\end{abstract}

%% file: main-content.tex
\section{Introduction}\label{sec:intro}

Randomized streaming algorithms can achieve exponentially better space bounds
than corresponding deterministic ones: this is a basic, well-known, easily
proved fact that applies to a host of problems of practical interest. A
prominent class of randomized streaming algorithms uses randomness in a very
specific way, namely to sketch the input stream by applying a random linear
transformation---given by a sketch matrix $S$---to the input frequency vector.
The primary goal of a streaming algorithm is to achieve sublinear space, so it
is infeasible to store $S$ explicitly. In some well-known cases, the most
natural presentation of the algorithm is to explicitly describe the
distribution of $S$, a classic case in point being frequency moment
estimation~\cite{Indyk06}. This leads to an algorithm that is very
space-efficient {\em provided one doesn't charge the algorithm any space cost
for storing $S$}. Algorithms that work this way can be thought of as accessing
a ``random oracle'': despite their impracticality, 
they have theoretical
value, because the standard ways of proving space {\em lower} bounds for
randomized streaming algorithms in fact work in this model. For the specific
frequency-moment algorithms mentioned earlier, \cite{Indyk06} goes on to
design variants of his algorithms that use only a small (sublinear) number of
random bits and apply a pseudorandom generator to suitably mimic the behavior
of his random-oracle algorithms. Thus, at least in this case, a random {\em
oracle} isn't necessary to achieve sublinear complexity. This raises a natural
question: from a space complexity viewpoint, does it ever help to use a random
oracle, as opposed to ``ordinary'' random bits that must be stored (and thus
paid for) if they are to be reused?

For most classic streaming problems, the answer is ``No,'' but for
unsatisfactory reasons: Newman's Theorem~\cite{Newman91} allows one to replace
a long oracle-provided random string by a much shorter one (that is cheap to
store), though the resulting algorithm is non-constructive. This brings us to
the recent and ongoing line of work on {\em adversarially robust} streaming
algorithms where we shall find that the answer to our question is a very
interesting ``Yes.'' For the basic and natural \miflong problem, defined
below, we shall show that three different approaches to randomization result
in distinct space-complexity behaviors. To explain this better, let us review
adversarial robustness briefly.

Some recent works have studied streaming algorithms in a setting where the
input to the algorithm can be adaptively (and adversarially) chosen based on
its past outputs. Existing (``classic'') randomized streaming algorithms may
fail in this {\em adversarial setting} when the input-generating adversary
learns enough about the past random choices of the algorithm to identify
future inputs on which the algorithm will likely fail. There are,
heuristically, two ways for algorithm designers to protect against this:
(a)~prevent the adversary from learning the past random choices of the
algorithm (in the extreme, by making a pseudo-deterministic algorithm), or
(b)~prevent the adversary from exploiting knowledge of past random decisions,
by having the algorithm's future behavior depend on randomness that it has not
yet revealed. Concretely, algorithms in this setting use techniques such as
independent re-sampling~\cite{BenEliezerY20}, sketch switching using
independent sub-instances of an underlying classic
algorithm~\cite{BenEliezerJWY20}, rounding outputs to limit the number of
computation paths~\cite{BenEliezerJWY20}, and differential privacy to safely
aggregate classic algorithm sub-instances~\cite{HassidimKMMS20}. Mostly, these
algorithms use at most as many random bits as their space bounds allow.
However, some recently published adversarially robust streaming algorithms for
vertex-coloring a graph (given by an edge stream)~\cite{ChakrabartiGS22,
AssadiCGS23}, and one for the \miflong problem~\cite{Stoeckl23}, assume access
to a large amount of oracle randomness: they prevent the adversary from
exploiting the random bits it learns by making each output depend on an
unrevealed part of the oracle random string. It is still open whether these
last two problems have efficient solutions that do not use this oracle
randomness hammer. This suggests the following question: 

\begin{center}\begin{minipage}{.75\linewidth}\begin{center}
  \emph{Are there problems for which space-efficient adversarially robust
  streaming algorithms provably require access to oracle randomness?}
\end{center}\end{minipage}\end{center}

In this paper, we prove that for certain parameter regimes, \miflong
(henceforth, $\mif$) is such a problem. In the problem $\mif(n,\rl)$, the
input is a stream $\langle e_1,\ldots,e_\rl \rangle$ of $\rl$ integers, not
necessarily distinct, with each $e_i \in \{1,\ldots,n\}$, where $1 \le \rl \le
n$. The goal is as follows: having received the $i$th integer, output a number
$v$ in $\{1,\ldots,n\} \setminus \{e_1,\ldots,e_i\}$.  We will be mostly interested in the setting $\rl = o(n)$, so
the ``trivial'' upper bound on the space complexity of $\mif(n,\rl)$ is
$O(\rl \log n)$, achieved by the deterministic algorithm that simply stores
the input stream as is.

\subsection{Groundwork for Our Results}\label{subsec:models}

To state our results about $\mif$, we need to introduce some key terminology.
Notice that \mif is a {\em tracking problem}: an output is required after
reading each input.%
\footnote{We do not consider algorithms with a ``one-shot'' guarantee, to only
be correct at the end of the stream, because a) the adversarial setting
requires tracking output b) for \mif and most other problems the difference
in space complexity is generally small.} 
Thus, we view streaming algorithms as generalizations of finite state
(Moore-type) machines. An algorithm $\cA$ has a finite set of states $\Sigma$
(leading to a space cost of $\log_2|\Sigma|$), a finite input set $\cI$, and a
finite output set $\cO$.  It has a transition function $\trans \colon \Sigma
\times \cI \times \cR \to \Sigma$ indicating the state to switch to after
receiving an input, plus an output function $\gamma \colon \Sigma \times \cR
\to \cO$ indicating the output produced upon reaching a state. How the final
parameter (in $\cR$) of $\trans$ and $\gamma$ is used depends on the type of
randomness. We consider four cases, leading to four different models of
streaming computation.

\iflipics\begin{itemize}\else
\begin{itemize}[left=2ex]\fi
  \item \emph{Deterministic.~} The initial state of the algorithm is a fixed
  element of $\Sigma$, and $\trans$ and $\gamma$ are deterministic (they do not
  depend on the parameter in $\cR$). 
  \item \emph{Random seed.~} The initial state is drawn from a distribution
  $\cD$ over $\Sigma$, and $\trans$ and $\gamma$ are deterministic. This models
  the situation that all random bits used count towards the algorithm's space
  cost.
  \item \emph{Random tape.~} The initial state is drawn from a distribution
  $\cD$ over $\Sigma$.\footnote{Requiring that this model use a fixed initial
  state could make some algorithms use one additional ``\textsc{init}''
  state.} The space $\cR$ is a sample space; when the algorithm receives an
  input $e \in \cI$ and is at state $\sigma \in \Sigma$, it chooses a random
  $\rho \in \cR$ independent of all previous choices and moves to state
  $\trans(e, \sigma, \rho)$. However, $\gamma$ is deterministic.%
  \footnote{Alternatively, we could associate a \emph{distribution} of outputs
  to each state, or a function mapping (input, state) pairs to outputs. As
  these formulations are slightly more complicated to prove things with, and
  only affect the space usage of \miflong algorithms by an additive $O(\log n
  + \log \frac{1}{\delta})$ amount, we stick with the one state = one output
  convention.} %
  This models the situation that the algorithm can make random decisions at
  any time, but it cannot remember past random decisions without recording
  them (which would add to its space cost). 
  \item \emph{Random oracle.~} The initial state is fixed; $\cR$ is a
  sample space. A specific $R \in \cR$ is drawn at the start of the algorithm
  and stays the same over its lifetime. When the algorithm is at state
  $\sigma$ and receives input $e$, its next state is $\trans(e, \sigma, R)$.
  The output given at state $\sigma$ is $\gamma(\sigma, R)$.  This models the
  situation that random bits are essentially ``free'' to the algorithm; it can
  read from a long random string which doesn't count toward its space cost and
  which remains consistent over its lifetime. A random oracle algorithm can be
  interpreted as choosing a random deterministic algorithm, indexed by $R$,
  from some family.  
\end{itemize}

These models form a rough hierarchy; they have been presented in (almost)
increasing order of power. Every $z$-bit ($2^z$-state) deterministic algorithm
can be implemented in any of the random models using $z$ bits of space; the
same holds for any $z$-bit random seed algorithm.  Every $z$-bit random tape
algorithm has a corresponding $(z + \log\rl)$-bit random oracle
algorithm---the added space cost is because for a random oracle algorithm to
emulate a random tape algorithm, it must have a way to get ``fresh''
randomness on each turn.%
\footnote{An alternative, which lets one express $z$-bit random tape
algorithms using a $z$-bit random oracle variant, is to assume the random
oracle algorithm has access to a clock or knows the position in the stream for free;
both are reasonable assumptions in practice.}

Streaming algorithms are also classified by the kind of correctness guarantee
they provide. Recall that we focus on ``tracking''
algorithms~\cite{BenEliezerJWY20}; they present an output after reading each
input item and this {\em entire sequence} of outputs must be correct.  Here
are three possible meanings of the statement ``algorithm $\cA$ is
$\delta$-error'' (we assume that $\cA$ handles streams of length $\rl$ with
elements in $\cI$ and has outputs in $\cO$):

\iflipics\begin{itemize}\else
\begin{itemize}[left=2ex]\fi
  \item \emph{Static setting.~} For all inputs $\stream \in \cI^\rl$,
  running $\cA$ on $\stream$ produces incorrect output with probability $\le
  \delta$.
  \item \emph{Adversarial setting.~} For all (computationally unbounded)
  adaptive adversaries $\alpha$ (i.e., for all functions
  $\alpha \colon \cO^\star \to \cI$),\footnote{By the minimax
  theorem, it suffices to consider deterministic adversaries.} running $\cA$
  against $\alpha$ will produce incorrect output with probability $\le
  \delta$.
  \item \emph{Pseudo-deterministic setting.~} There exists a canonical output
  function $f \colon \cI^\star \rightarrow \cO$ producing all correct outputs
  so that, for each $\stream \in \cI^\rl$, $\cA(\stream)$ fails to output
  $f(\stream)$ with probability $\le \delta$.
\end{itemize}

Algorithms for the static setting are called ``classic'' streaming algorithms;
ones for the adversarial setting are called ``adversarially robust'' streaming
algorithms. All pseudo-deterministic algorithms are adversarially robust, and
all adversarially robust algorithms are also classic.

As a consequence of Newman's theorem~\cite{Newman91}, any random oracle or
random tape algorithm in the static setting with error $\delta$ can be
emulated using a random seed algorithm with only $\epsilon$ increase in error
and an additional $O(\log \rl + \log \log |\cI| + \log \frac{1}{\epsilon
\delta})$ bits of space. However, the resulting algorithm is non-constructive.

\subsection{Our Results}\label{subsec:results}

As context for our results, we remind the reader that it's trivial to solve
$\mif(n,\rl)$ in $O(\rl\log n)$ space deterministically (somewhat better
deterministic bounds were obtained in \cite{Stoeckl23}). Moving to randomized
algorithms, \cite{Stoeckl23} gave a space bound of $O(\log^2 n)$ for $\rl
\le n/2$ in the static setting, and a bound of $\tO(\rl^2/n + 1)$\,\footnote{The
notations $\tO(\cdot)$ and $\tOmega(\cdot)$ hide factors polylogarithmic in
$n$ and $\rl$.} in the adversarial setting, using a random \emph{oracle}. The
immediate takeaway is that, given access to a deep pool of randomness (i.e.,
an oracle), \mif becomes easy in the static setting for essentially the full
range of stream lengths $\rl$ and remains easy even against an adversary for
lengths $\rl \le \sqrt{n}$.

\newcommand{\centeredcell}[1]{\begin{tabular}{l} #1 \end{tabular}}
\begin{table}[!htbp]
  \begin{minipage}{\linewidth}
  \renewcommand{\arraystretch}{1.5}
  \centering
  \begin{tabular}{ p{0.2\textwidth} p{0.17\textwidth} p{0.32\textwidth} p{0.16\textwidth} }
  \toprule
    Setting & Type  & Bound & Reference \\ 
  \midrule
  \rowcolor{black!7!white}  
  Static & Random seed & $O((\log n)^2)$ if $\rl \le n/2$ & \cite{Stoeckl23}\footnote{This is obtained by accounting for the randomness cost of
\cite{Stoeckl23}'s random {\em oracle} algorithm for the static setting.} \\
  {Adversarial} & {Random oracle} & $O((\frac{\rl^2}{n} + \log n) \log n)$ & \cite{Stoeckl23} \\
          & & $\Omega(\frac{\rl^2}{n})$ & \cite{Stoeckl23} \\
  \rowcolor{black!7!white}  
  {Adversarial} & {Random tape} & $O(\rl^{\log_n \rl} (\log \rl)^2 + \log \rl \cdot\log n ) ~~\dagger$ & \Cref{thm:rt-ub-intro} \\
  \rowcolor{black!7!white}
            &  & $\Omega(\rl^{\frac{15}{32} \log_n \rl}) ~~\dagger$ & \Cref{thm:rt-lb-intro} \\
  {Adversarial} & {Random seed} & $O((\frac{\rl^2}{n} + \sqrt{\rl} + \log n) \log n)$ & \cite{Stoeckl23}\footnote{The random seed algorithm for the adversarial setting is given in the arXiv version of \cite{Stoeckl23}.} \\
            &  & $\Omega(\frac{\rl^2}{n} + \sqrt{\frac{\rl}{(\log n)^3}} + \rl^{1/5} )$ & \Cref{cor:rs-lb-intro} \\
  \rowcolor{black!7!white}  
  Pseudo-deterministic & Random oracle & $\Omega(\frac{\rl}{(\log (2n/\rl))^2} + \left(\rl \log n \right)^{1/4} )$ & \Cref{thm:pd-lb-intro} \\
  {Static} & {Deterministic} & $\Omega(\frac{\rl}{\log (2n/\rl)} + \sqrt{\rl})$ & \cite{Stoeckl23} \\
      &  & $O(\frac{\rl \log \rl}{\log n} + \sqrt{\rl \log \rl})$ & \cite{Stoeckl23} \\
  \bottomrule
  \end{tabular}
  \caption{\label{table:mif-results} Bounds for the space complexity of
  $\mif(n,\rl)$, from this and prior work. To keep expressions simple, these
  bounds are evaluated at error level $\delta = 1/n^2$, when applicable.
  ($\dagger$) indicates that the precise results are stronger. \label{table:bounds-summary}}
  \end{minipage}
\end{table}

\begin{figure}[!htbp]
  \centering
  \includegraphics[height=9cm]{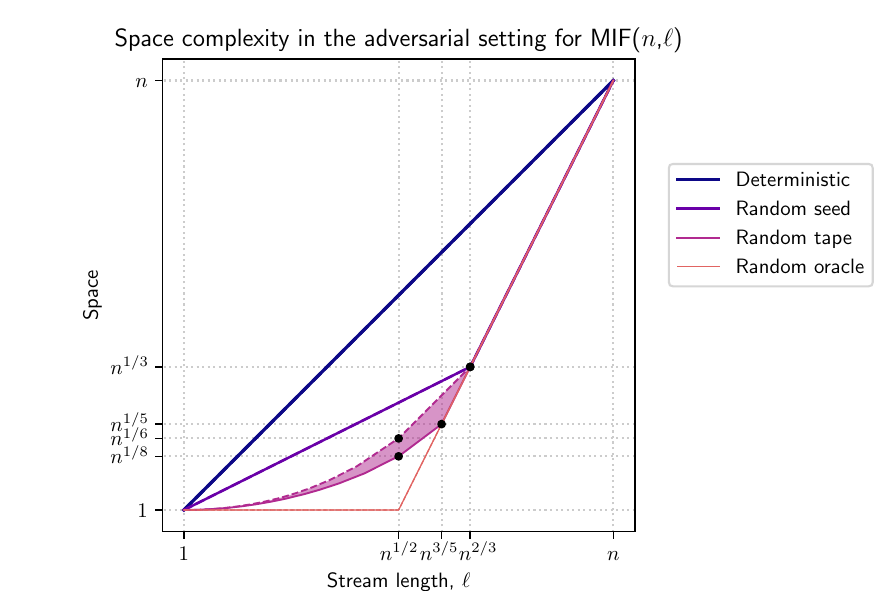}
  \caption{Known bounds for the space complexity of $\mif(n,\rl)$ in different streaming models, at error level $\delta=1/n^2$. This is a log-log plot. Upper and lower bounds are drawn using lines of the same color; the region between them is shaded. The upper and lower bounds shown all match (up to $\polylog(n)$ factors) except for the case of adversarially robust, random tape algorithms. Pseudo-deterministic and deterministic complexities match within $\polylog(n)$ factors. \label{fig:bounds-summary}}
\end{figure}

The main results of this paper consist of three new lower bounds and one new
upper bound on the space complexity of $\mif(n,\rl)$.  Stating the bounds in
their strongest forms leads to complicated expressions; therefore, we first
present some easier-to-read takeaways from these bounds that carry important
conceptual messages. In the lower bounds below, the error level should be
thought of as $\delta = 1/n^2$.

\begin{result} \label{res:rt-lb}
  At $\rl = \sqrt{n}$, adversarially robust random tape algorithms for
  $\mif(n,\rl)$ require $\Omega(\rl^{1/4})$ bits of space. More generally, for
  every constant $\alpha \in (0,1)$, there is a constant $\beta \in (0,1)$
  such that at $\rl = \Omega(n^\alpha)$, the space requirement is
  $\Omega(\rl^\beta)$, in the adversarially robust random tape setting.
\end{result}

This shows that \mif remains hard, even for modest values of $\rl$, if we must
be robust while using only a random {\em tape}, i.e., if there is a cost to
storing random bits we want to reuse---a very reasonable requirement for a
practical algorithm. The above result is an exponential separation between the
random tape and random oracle models.

The random {\em seed} model places an even greater restriction on an
algorithm: besides counting towards storage cost, random bits are available
only at initialization and not on the fly. Many actual randomized algorithms,
including streaming ones, are structured this way, making it a natural model
to study. We obtain the following result.

\begin{result} \label{res:rs-lb}
  Adversarially robust random seed algorithms for $\mif(n,\rl)$ require
  $\tOmega(\sqrt{\rl})$ bits of space.
\end{result}

Consider the two results above as $\rl$ decreases from $\sqrt{n}$ to
$\Theta(1)$. The bound in \Cref{res:rs-lb} stays interesting even when $\rl =
n^{o(1)}$, so long as $\rl \ge (\log n)^C$ for a suitable constant $C$ (in
fact, the full version of the result is good for even smaller $\rl$). In
contrast, the bound in \Cref{res:rt-lb} peters out at much larger values of
$\rl$. There is a very good reason: \mif starts to become ``easy,'' even under
a random-tape restriction, once $\rl$ decreases to sub-polynomial in $n$.
Specifically, we obtain the following {\em upper} bound.

\begin{result} \label{res:rt-ub}
  There is an adversarially robust random tape algorithm for $\mif(n,\rl)$
  that, in the regime $\rl = O(2^{\sqrt{\log n}})$, uses $O(\log\rl \cdot \log
  n)$ bits of space.
\end{result}

Notice that at $\rl = \Theta(2^{(\log n)^{1/C}})$, where $C \ge 2$ is
a constant, the bound in \Cref{res:rt-ub} is polylogarithmic in $\rl$.
Combined with the lower bound in \Cref{res:rs-lb}, we have another exponential
separation, between the random seed and random tape models.

Empirically, the existing literature on streaming algorithms consists of
numerous linear-sketch-based algorithms, which tend to be efficient in the
random seed model, and sampling-based algorithms, which naturally fit the
(stronger) random tape model. In view of this, the combination of
\Cref{res:rs-lb,res:rt-ub} carries the following important message.
\begin{center}\begin{minipage}{.8\linewidth}\begin{center}
  \emph{Sampling is provably more powerful than sketching in an adversarially
  robust setting.}
\end{center}\end{minipage}\end{center}

The proof of \Cref{res:rs-lb} uses a reduction, given in prior
work~\cite{Stoeckl23}, that converts a space lower bound in the
pseudo-deterministic setting to a related bound in the random-seed setting. A
pseudo-deterministic algorithm is allowed to use randomness (which, due to
Newman's theorem, might as well be of the oracle kind) but must, with high
probability, map each input to a {\em fixed} output, just as a deterministic
algorithm would. This strong property makes the algorithm adversarially
robust, because the adversary has nothing to learn from observing its outputs.
Thanks to the \cite{Stoeckl23} reduction, the main action in the proof of
\Cref{res:rs-lb} is the following new lower bound we give.

\begin{result} \label{res:pd-lb}
  Pseudo-deterministic random oracle algorithms for $\mif(n,\rl)$ require
  $\tOmega(\rl)$ bits of space.
\end{result}

These separations rule out the possibility of a way to convert an
adversarially robust random oracle algorithm to use only a random {\em seed}
or even a random {\em tape}, with only minor (e.g., a $\polylog(\rl,n)$
factor) overhead. In contrast, as we noted earlier, such a conversion is
routine in the static setting, due to Newman's theorem~\cite{Newman91}.
The separation between random oracle and random tape settings shows that
$\miflong$ is a problem for which much lower space usage is possible if one's
adversaries are computationally bounded (in which case a pseudo-random
generator can emulate a random oracle.)

\Cref{table:bounds-summary} shows more detailed versions of the above results
as well as salient results from earlier work. Together with
\Cref{fig:bounds-summary}, it summarizes the state of the art for the space
complexity of $\mif(n, \rl)$.  The fully detailed versions of our results,
showing the dependence of the bounds on the error probability, appear in later
sections of the paper, as indicated in the table.

\subsection{Related Work}

We briefly survey related work. An influential early work~\cite{HardtW13}
considered adaptive adversaries for {\em linear} sketches.  The adversarial
setting was formally introduced by \cite{BenEliezerJWY20}, who provided
general methods (like sketch-switching) for designing adversarially robust
algorithms given classic streaming algorithms, especially in cases where the
problem is to approximate a real-valued quantity. For some tasks, like
$F_0$-estimation, they obtained slightly better upper bounds by using a random
oracle, although later work \cite{WoodruffZ22} removed this need.
\cite{BenEliezerY20} observed that in sampling-based streaming algorithms,
increasing the sample size is often all that is needed to make an algorithm
adversarially robust. \cite{HassidimKMMS20} described how to use differential
privacy techniques as a more efficient alternative to sketch-switching, and
\cite{BenEliezerEO22} used this as part of a more efficient adversarially
robust algorithm for turnstile $F_2$-estimation.

Most of these papers focus on providing algorithms and general techniques, but
there has been some work on proving adversarially robust lower bounds.
\cite{KaplanMNS21} described a problem (of approximating a certain real-valued
function) that requires exponentially more space in the adversarial setting
than in the static setting. \cite{ChakrabartiGS22}, in a brief comment,
observed a similar separation for a simple problem along the lines of $\mif$.
They also proved lower bounds for adversarially robust coloring algorithms for
graph edge-insertion streams. \cite{Stoeckl23} considered the \mif problem as
defined here and, among upper and lower bounds in a number of models,
described an adversarially robust algorithm for \mif that requires a random
oracle; they asked whether a random oracle is {\em necessary} for
space-efficient algorithms.

The {\em white-box} adversarial setting~\cite{AjtaiBJSSWZ22} is similar to the
adversarial setting we study, with the adversary having the additional power
of seeing the internal state of the algorithm, including (if used) the random
oracle.  \cite{Stoeckl23} proved an $\Omega(\rl/\polylog(n)$) lower bound for
$\mif(n,\rl)$ for random tape algorithms in this setting, suggesting that any
more efficient algorithm for \mif must conceal some part of its internal
state. Pseudo-deterministic streaming algorithms were introduced by
\cite{GoldwasserGMW20}, who gave lower bounds for a few problems.
\cite{BravermanKKS23, GrossmanGS23} gave lower bounds for pseudo-deterministic
algorithms that approximately count the number of stream elements. The latter
shows they require $\Omega(\log m)$ space, where $m$ is the stream length; in
contrast, in the static setting, Morris's counter algorithm\footnote{Morris's
is a ``random tape'' algorithm; ``random seed'' algorithms for counting aren't
better than deterministic ones.} uses only $O(\log \log m)$ space.

While it is not posed as a streaming task, the {\em mirror game} introduced by
\cite{GargS18} is another problem with conjectured separation between the
space needed for different types of randomness. In the mirror game, two
players (Alice and Bob) alternately state numbers in the set $\{1,\ldots,n\}$,
where $n$ is even, without repeating any number, until one player mistakenly
states a number said before (loss) or the set is completed (tie).
\cite{GargS18} showed that if Alice has $o(n)$ bits of memory and plays a
deterministic strategy, Bob can always win. Later, \cite{Feige19, MenuhinN22}
showed that if Alice has access to a random oracle, she can tie-or-win
w.h.p.~using only $O(\polylog(n))$ space. A major open question here is how
much space Alice needs when she does not have a random oracle. \cite{MagenN22}
did not resolve this, but showed that if Alice is ``open-book'' (equivalently,
that Bob is a white-box adversary and can see her state), then Alice needs
$\Omega(n)$ bits of state to tie-or-win.

Assuming access to a random oracle is a reasonable temporary measure when
designing streaming algorithms in the static setting. As noted at the
beginning of \Cref{sec:intro}, \cite{Indyk06} designed $L_p$-estimation
algorithms using random linear sketch matrices, without regard to the amount
of randomness used, and then described a way to apply Nisan's
PRG~\cite{Nisan90} to partially derandomize these algorithms and obtain
efficient (random seed) streaming algorithms. In general, the use of PRGs for
linear sketches has some space overhead, which later work (see
\cite{JayaramW23} as a recent example) has been working to eliminate.

It is important to distinguish the ``random oracle'' type of streaming
algorithm from the ``random oracle model'' in cryptography~\cite{BellareR93},
in which one assumes that \emph{all} agents have access to the random oracle.
\cite{AjtaiBJSSWZ22}, when defining white-box adversaries, also assumed that
they can see the same random oracle as the algorithm; and, for one task,
obtained a more efficient algorithm against a computationally bounded
white-box adversary, when both have access to a random oracle, than when
neither do.  Tight lower bounds are known in neither case.

The power of different types of access to randomness has been studied in
computational complexity. \cite{Nisan93B} showed that logspace Turing machines
with a multiple-access random tape can (with zero error) decide languages that
logspace Turing machines with a read-once random tape decide only with bounded
two-sided error. This type of separation does not hold for {\em time}
complexity classes.

For a more detailed history and survey of problems related to \miflong, we
direct the reader to \cite{Stoeckl23}.

\section{Technical Overview}\label{sec:technical-overview}

The proofs of \Cref{res:rt-lb,res:rt-ub,res:pd-lb} are all significant
generalizations of existing proofs from \cite{Stoeckl23} which handled
different (and more tractable) models. The proof of \Cref{res:rs-lb} consists
of applying a reduction from \cite{Stoeckl23} to the lower bound given by
\Cref{res:pd-lb}. As we explain our techniques, we will summarize the relevant
``basic'' proofs from \cite{Stoeckl23}, which will clarify the enhancements
needed to obtain our results.  

Space complexity lower bounds in streaming models are often proved via
communication complexity. This meta-technique is unavailable to us, because
the setup of communication complexity blurs the distinctions between random
seed, random tape, and random oracle models and our results are all about
these distinctions. Instead, to prove \Cref{res:rt-lb}, we design a suitable
strategy for the stream-generating adversary that exploits the algorithm's
random-tape limitation by learning enough about its internal state. Our
adversary uses a nontrivially recursive construction.  To properly appreciate
it, it is important to understand what streaming-algorithmic techniques the
adversary must contend with.  Therefore, we shall discuss our {\em upper}
bound result first.

\subsection{Random Tape Upper Bound\texorpdfstring{ (\Cref{res:rt-ub}; \Cref{thm:rt-ub-intro})}{}}\label{subsec:rt-ub-overview}

The adversarially robust random tape algorithm for $\mif(n,\rl)$ can be seen as a
generalization of the random oracle and random seed algorithms.

\mypara{The random oracle algorithm and its adversaries} The random
oracle algorithm for $\mif(n,\rl)$ from \cite{Stoeckl23} has the following structure. It interprets
its oracle random string as a uniformly random sequence $L$ containing $\rl +
1$ distinct elements in $[n]$. As it reads its input, it keeps track of which
elements in $L$ were in the input stream so far (were ``covered''). It reports
as its output the first uncovered element of $L$.  Because $L$ comes from the
oracle, the space cost of the algorithm is just the cost of keeping track of
the set $J$ of covered positions in $L$. We will explain why that can be done
using only $O((\rl^2/n + 1) \log \rl)$ space, in expectation.

An adversary for the algorithm only has two reasonable strategies for choosing the
next input. It can ``echo'' back the current algorithm output to be the next input to
the algorithm. It can also choose the next input to be a value from the set $U$ of
values that are neither an earlier input nor the current output---but because $L$
is chosen uniformly at random, one can show that the adversary can do no better than picking
the next input uniformly at random from $U$. (The third strategy, of choosing an old
input, has no effect on the algorithm.) When the algorithm is run against an
adversary that chooses inputs using a mixture of the echo and random strategies,
the set $J$ will be structured as the union of a contiguous interval starting at 1
(corresponding to the positions in $L$ covered by the echo strategy) and a sparse
random set of expected size $O(\rl^2/n)$ (corresponding to positions in $L$ covered by
the random strategy). Together, these parts of $J$ can be encoded using 
$O((\rl^2/n + 1) \log \rl)$ bits, in expectation.

\mypara{Delaying the echo strategy} If we implemented the above random
oracle algorithm as a random seed algorithm, we would need $\Omega(\rl)$ bits
of space, just to store the random list $L$. But why does $L$ need to have
length $\rl + 1$?  This length is needed for the algorithm to be resilient to
the echo strategy, which covers one new element of $L$ on every step; if $L$
were shorter, the echo strategy could entirely cover it, making the algorithm
run out of possible values to output. The random seed algorithm for
$\mif(n,\rl)$ works by making the echo strategy less effective, ensuring that
multiple inputs are needed for it to cover another element of $L$. It does this
by partitioning $[n]$ into $\Theta(\rl)$ disjoint subsets (``blocks'') of size
$\Theta(n/\rl)$, and then taking $L$ to be a random list of blocks (rather
than a random list of elements of $[n]$).  We will now say that a block is
``covered'' if \emph{any} element of that block was an input. Instead of
outputting the first uncovered element in $L$, the algorithm will run a
deterministic algorithm for $\mif$ \emph{inside} the block corresponding to
the first uncovered block of $L$, and report outputs from that; and will only
move on to the next uncovered block when the nested algorithm stops. See
\Cref{alg:rt-example} for the details of this design. Because the analogue of
the echo strategy now requires many more inputs to cover a block, we can make
the list $L$ shorter. This change will not make the random strategy much more
effective.\footnote{The fact that $[n]$ is split into $\Omega(\rl)$ blocks is
enough to mitigate the random strategy; with $\rl$ guesses, the adversary is
unlikely to guess more than a constant fraction of the elements in $L$.} The
minimum length of $L$ is constrained by the $O(n/\rl)$ block sizes, which limit the number of outputs that the nested algorithm can make; as a result, one must have $L = \Omega(\rl^2/n)$. In the end, after balancing the length of the list with the
cost of the nested algorithm, the optimal list length for the random seed
algorithm will be $O(\rl^2/n + \sqrt{\rl})$.

\begin{algorithm}[htb]
  \caption{Example: recursive construction for a random tape $\mif(n,\rl)$ algorithm, building on algorithm $\cA$}
  \label[listing]{alg:rt-example}

  \begin{algorithmic}[1]
  \Statex Parameter: $t \in [\Omega(\rl^2 / n), \rl]$ is the number of parts into which the input stream is split
  \vspace{2mm}
  \Statex \ul{\textbf{Initialization}}:
    \State Let $k = O(t)$, $s = O(\rl)$, and $B_1,\ldots,B_{s}$ be a partition of $[n]$ into $s$ equal ``blocks'' \Comment{assuming $s \mid n$}
    \State $L \gets$ uniformly randomly chosen sequence of $k$ distinct elements of $[s]$
    \State $J \gets \emptyset$, is a subset of $[k]$ \Comment{a set marking which blocks of $L$ have been covered}
    \State $c \gets 1$ \Comment{the current active block}
    \State $A \gets$ instance of algorithm $\cA$ solving $\mif(n/s, \ceil{\rl/t})$

  \vspace{2mm}
  \Statex \ul{\textbf{Update}($a \in [n]$)}: 
    \State Let $h$ be the block containing $a$, and $x$ the rank of $a$ in $B_h$
    \If{$h \in L$}
        \State Add $j$ to $J$, where $L_j = h$ \Comment{Mark list element containing $h$ as used}
    \EndIf
    \If{$h = L_c$}
        \State $A.\textsc{Update}(x)$
    \EndIf
    \If{$A$ is out of space} \Comment{This requires that $A.\textsc{Update}()$ be called $\ge \ceil{\rl/t}$ times}
      \State $c \gets$ least integer which is $>c$ and not in $J$ \Comment{This line may abort if $J = [k]$}
      \State $A \gets$ new instance of algorithm $\cA$ \Comment{Using new random bits, if $\cA$ is randomized}
    \EndIf
    
  \vspace{2mm}
  \Statex \ul{\textbf{Output} $\rightarrow [n]$}:
    \State Let $x \in [n/s]$ be the output of $A$
    \State \textbf{return} $x$th entry of block $B_c$
  \end{algorithmic}
\end{algorithm}

\begin{figure}[ht]
  \centering
  \includegraphics[width=15cm]{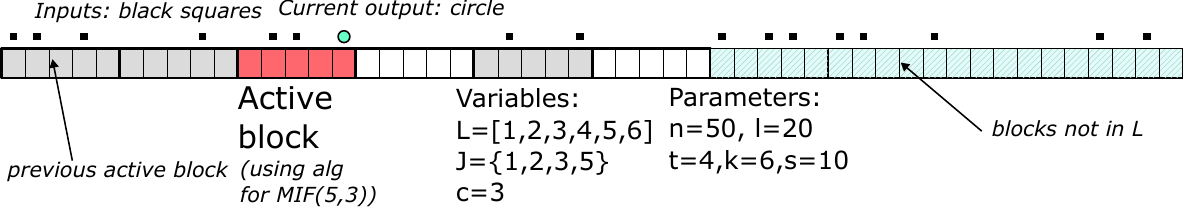}
  \caption{A diagram illustrating the state of an instance of \Cref{alg:rt-example} on an example input. Positions on the horizontal axis correspond to integers in $[n]$; the set of values in the input stream ($\{1,2,4,9,12,13,\ldots\}$) is marked with black squares; the current output value ($15$) with a circle. Outside this example, $L$ need not be contiguous or in sorted order. \label{fig:rt-rec-diagram}}
\end{figure}

\mypara{The recursive random tape algorithm} The random seed algorithm for
$\mif(n,\rl)$ used the construction of \Cref{alg:rt-example} to build on top
of an ``inner'' deterministic algorithm.\footnote{The construction uses
randomness in two places: when initializing the random sequence $L$, and
(possibly) each time the inner algorithm is initialized. For the random seed
model, every ``inner'' initialization would require a corresponding set of
random bits, which are counted toward the space cost of the algorithm. Using a
deterministic inner algorithm avoids this cost.} To get an efficient random
tape algorithm, we can recursively apply the construction of
\Cref{alg:rt-example} $d - 1$ times, for $d = O(\min(\log \rl, \log n / \log
\rl))$; at the end of this recursion, we can use a simple deterministic
algorithm for $\mif$. The optimal lengths of the random lists used at each
level of the recursion are determined by balancing the costs of the different
recursion levels.  We end up choosing list lengths that all bounded by a
quantity which lies between $O(\rl^{1/d})$ and $O(\rl^{1/(d-1)})$.

In the extreme case where $d = \Theta(\log \rl)$ and the required error level
$\delta$ is constant, our recursive algorithm may have a stack of random
lists, each of length $2$, and every time a level of the algorithm completes
(i.e., all blocks of a list have been used), it will make a new instance of
that level. That is, some large uncovered block will be split into many
smaller blocks, and the algorithm will randomly pick two of them for the new
instance's list. Because the lists are all short, the algorithm will not need
to remember many random bits at a given time; in exchange, for this regime it
needs a very large ($n = \rl^{\Omega(d)}$) number of possible outputs and will
frequently need to sample new random lists.

We defer the exact implementation details to the final version of our
algorithm, \Cref{alg:rt}.  It looks somewhat different from the recursive
construction in \Cref{alg:rt-example}, because we have unraveled the recursive
framing to allow for a simpler error analysis that must only bound the
probability of a single ``bad event.''

\subsection{Random Tape Lower Bound\texorpdfstring{ (\Cref{res:rt-lb}; \Cref{thm:rt-lb-intro})}{}} \label{subsec:rt-lb-overview}

\mypara{The \avoid problem} At the core of many of the \mif lower bounds is
the \textsc{SubsetAvoidance} communication problem, introduced
in~\cite{ChakrabartiGS22}. Here we have two players, Alice and Bob, and a
known universe $[m]$: Alice has a set $A \subseteq [m]$ of size $a$, and
should send a message (as short as possible) to Bob, who should use the
message to output a set $B \subseteq [m]$ of size $b$ which is disjoint from
$A$. Henceforth, we'll call this problem $\avoid(m,a,b)$.
\cite{ChakrabartiGS22} showed that both deterministic and constant-error
randomized one-way protocols for this problem require $\Omega(a b /m)$ bits of
communication. An adversarially robust $z$-space algorithm for $\mif(m, a +
b)$ can be used as a subroutine to implement a $z$-bit one-way protocol for
$\avoid(m,a,b)$, thereby proving $z = \Omega(a b /m)$. This immediately gives
us an $\Omega(\rl^2/n)$ space lower bound for $\mif(n,\rl)$, which, as we have
seen, is near-optimal in the robust, random oracle setting.

\mypara{The random tape lower bound} 
To prove stronger lower bounds that exploit the random tape limitation of the
algorithm, we need a more sophisticated use of \avoid. Fix an adversarially
robust, random tape, $z$-space algorithm $\cA$ for $\mif(n,\rl)$.  Roughly speaking,
while the random oracle argument used $\cA$ to produce an \avoid protocol at
the particular scale $a = b = \rl$, for the fixed universe $[n]$, our random
tape argument will ``probe'' $\cA$ in a recursive fashion---reminiscent of the
recursion in our random tape upper bound---to identify a suitable scale and
sub-universe at which an \avoid protocol can be produced. This probing will
itself invoke the \avoid lower bound to say that if an $\avoid(m,a,b)$
protocol is built out of a $z$-space streaming algorithm where $z \ll a$, then
$B$ must be small, with size $b = O((z/a) m)$.

We will focus on the regime where $\delta = O(1/n)$. This error level requires
a measure of structure from the algorithm: it cannot just pick a random output
each step, because that would risk colliding with an earlier input with $\ge
1/n$ probability. Our recursive argument works by writing $z$, the space usage
of $\cA$, as a function of a space lower bound for $\mif(w,t)$, where $w =
\Theta(z n / \rl)$ and $t = \Theta(\rl/z)$. For small enough $z$, $t^2/w \gg
\rl^2/n$, so by repeating this reduction step a few times we can increase the
ratio of the stream length to the input domain size until we can apply the
simple $\Omega(\hat{\rl}^2/\hat{n})$ lower bound for
$\mif(\hat{n},\hat{\rl})$.  With the right number of reduction steps, one
obtains the lower bound formula of \Cref{thm:rt-lb-intro}, of which
\Cref{res:rt-lb} is a special case.

\mypara{The reduction} The reduction step argues that the $\mif(n,\rl)$
algorithm $\cA$ ``contains'' a $z$-space algorithm for $\mif(w,t)$, which, on
being given any $t = O(\rl / z)$ items in a certain sub-universe $W \subseteq
[n]$ of size $w = O(z n / \rl)$, will repeatedly produce missing items {\em
from that sub-universe}.  That such a set $W$ \emph{exists} can be seen as a
consequence of the lower bound for $\avoid$: if $\cA$ receives a random sorted
subset $S$ of $\rl/2$ elements in $[n]$, then because there are
\smash{$\binom{n}{\rl/2}$} possible subsets, most of the $2^z$ states of $\cA$
will need to be ``good'' for $\Omega(2^{-z} \binom{n}{\rl/2})$ different
subsets. In particular, upon reaching a given state $\sigma$, for $\cA$ to
solve \mif with error probability $O(1/n)$, its outputs henceforth---for the
next $\rl/2$ items in the stream---must avoid most of the sets of inputs that
could have led it to $\sigma$.  We will prove by a counting argument
(\Cref{lem:forward-avoid}) that after the random sequence $S$ is sent, each
state $\sigma$ has an associated set $H_\sigma$ of possible ``safe'' outputs
which are unlikely to collide with the inputs from $S$, and that $|H_\sigma|$
is typically $O(z n / \rl)$. Thus, for a typical state $\sigma$, starting
$\cA$ from $\sigma$ causes its next $\rl/2$ outputs to be inside $H_\sigma$,
w.h.p.; in other words, $\cA$ contains a ``sub-algorithm'' solving $\mif(O(z n
/ \rl), \rl/2)$ on the set $W = H_\sigma$. 

However, even though there exists a set $W$ on which $\cA$ will concentrate
its outputs, it may not be possible for an adversary to find it. In
particular, had $\cA$ been a random {\em oracle} algorithm, each setting of
the random string might lead to a different value for $W$, making $W$
practically unguessable. But $\cA$ is in fact a random {\em tape} algorithm,
so we can execute the following strategy.

In our core lemma, \Cref{lem:mif-rt-lb-step}, we design an adversary
(\Cref{alg:rt-adversary-step}) that can with $\Omega(1)$ probability identify
a set $W$ of size $\Theta(z n /\rl)$ for which the next $\Theta(\rl/z)$
outputs of $\cA$ will be contained in $W$, with $\Omega(1)$ probability, no
matter what inputs the adversary sends next. In other words, our adversary
will identify a part of the stream and a sub-universe of $[n]$ where the
algorithm solves $\mif(\Theta(z n /\rl), \Theta(\rl/z))$. The general strategy
is to use an iterative search based on a win-win argument.  First, the
adversary will send a stream comprising a random subset $S$ of size $\rl/2$ to
$\cA$, to ensure that henceforth its outputs are contained in some (unknown)
set $H_\rho$, where $\rho$ is the (unknown) state reached by $\cA$ just after
processing $S$.  Because $\cA$ has $\le 2^z$ states, from the adversary's
perspective there are $\le 2^z$ possible candidates for $H_\rho$. Then, the
adversary conceptually divides the rest of the stream to be fed to $\cA$ into
$O(z)$ phases, each consisting of $t = O(\rl/z)$ stream items. In each phase,
one of the following things happens.
\begin{enumerate}
\item There exists a ``sub-adversary'' (function to choose the $t$ items
constituting the phase, one by one) which will probably make $\cA$ output an
item that rules out a constant fraction of the candidate values for $H_\rho$
(output $i$ rules out set $J$ if $i \notin J$). The adversary then runs this
sub-adversary.
\item No matter how the adversary picks the $t$ inputs for this phase, there
will be a set $W$ (roughly, an ``average'' of the remaining candidate sets)
that probably contains the corresponding $t$ outputs of $\cA$.
\end{enumerate}
As the set of candidate sets can only shrink by a constant fraction $O(z)$
times, the first case can only happen $O(z)$ times, with high probability.
Thus, eventually, the adversary will identify the set $W$ that it seeks. Once
it has done so, it will run the optimal adversary for $\mif(\Theta(z n /\rl),
\Theta(\rl/z))$.  This essentially reduces the lower bound for $\mif(\rl,n)$
to that for $\mif(\Theta(z n /\rl), \Theta(\rl/z))$.

One subtlety is that we will need to carefully account for the probability
that $\cA$, over the next $\Theta(\rl/z)$ stream items, produces outputs
outside $W$. This will require us to distinguish between two types of
``errors'' for the algorithm over those next $\Theta(\rl/z)$ items: an $O(1)$
chance of producing an output outside $W$, and a smaller chance of making a
mistake per the definition of \mif, i.e., outputting an item that was not
missing (cf.~\Cref{def:mif-rt-err-complexity}).

\subsection{Random Seed Lower Bound via Pseudo-Determinism\texorpdfstring{ (\Cref{res:rs-lb};
\Cref{cor:rs-lb-intro})}{}} \label{subsec:rs-lb-overview}

The adversary constructed above for our random tape lower bound can be seen as
a significant generalization of the adversary used by \cite{Stoeckl23} to
prove a random seed lower bound conditioned on a (then conjectured)
pseudo-deterministic lower bound. Indeed, \cite{Stoeckl23}'s adversary against
a $z$-space algorithm $\cA$ also proceeds in a number of phases, each of
length $t = \Theta(\rl / z)$. In each step, either (1)~it can learn some new
information about the initial state of $\cA$ (the ``random seed''), by sending
$\cA$ a specific stream of inputs in $[n]^t$, looking at the resulting output,
and ruling out the seed values that could not have produced the output; or
(2)~it cannot learn much information, because for any possible input stream in
$[n]^t$, $\cA$ has an output that it produces with constant probability.  Each
time the adversary follows the case~(1), a constant fraction of the $\le 2^z$
seed values are ruled out. Therefore, either within $O(z)$ steps the adversary
will exactly learn the seed, at which point it can perfectly predict $\cA$'s
behavior, which lands us in case~(2); or $\cA$ will not reveal much
information about the seed in a given phase, which also puts us in case~(2).
Because case~(2) means that $\cA$ behaves pseudo-deterministically, $\cA$ must
use enough space to pseudo-deterministically solve $\mif(n, t)$.

Thus, \Cref{res:rs-lb} follows as a corollary of \Cref{res:pd-lb}, which we
discuss next.

\subsection{Pseudo-Deterministic Lower Bound\texorpdfstring{ (\Cref{res:pd-lb}; \Cref{thm:pd-lb-intro})}{}} \label{subsec:pd-lb-overview}

This proof generalizes \cite{Stoeckl23}'s space lower bound for {\em
deterministic} $\mif(n,\rl)$ algorithms, which we briefly explain. Fix a
deterministic $\mif(n,\rl)$ algorithm $\cA$ that uses $z$ bits of space. For
each stream $\stream$ with length $|\stream| \le \rl$, define $F_\stream$ to
be the set of {\em all possible outputs} of $\cA$ corresponding to
length-$\rl$ streams that have $\stream$ as a prefix. Let $\rho$ be a stream
such that $|\stream| + |\rho| \le \rl$. Then, by definition, $F_{\stream
\concat \rho} \subseteq F_\stream$ whereas, by the correctness of $\cA$,
$F_{\stream \concat \rho} \cap \rho = \emptyset$. Now consider the $\avoid$
problem over the universe $F_\stream$, for a fixed $\stream$: if Alice gets
$\rho \subseteq F_\stream$ as an input, she could send Bob the state $\sigma$
of $\cA$ upon processing $\stream \concat \rho$, whereupon Bob could determine
$F_{\stream \concat \rho}$ (by repeatedly running $\cA$'s state machine
starting at $\sigma$), which would be a valid output. 

Let us restrict this scenario to suffixes $\rho$ of some fixed length $t$;
we'll soon determine a useful value for $t$. By the above observations, were
it the case that
\begin{align} \label{eq:non-shrinkage}
  \exists \stream \in [n]^{\le \rl - t}~ 
  \forall \rho \in [n]^t \colon
  |F_{\stream \concat \rho}| \ge \tfrac12 |F_\stream| \,,
\end{align}
we would have a $z$-bit protocol for $\avoid(|F_\stream|, t, \frac12
|F_\stream|)$. By the \cite{ChakrabartiGS22} lower bound, we would have $z \ge
Ct$ for a universal constant $C$.  On the other hand, if the opposite were
true, i.e.,
\begin{align} \label{eq:shrinkage}
  \forall \stream \in [n]^{\le \rl - t}~ 
  \exists \rho \in [n]^t \colon
  |F_{\stream \concat \rho}| < \tfrac12 |F_\stream| \,,
\end{align}
then, starting from the empty stream $\emptystream$, we could add a sequence
of length-$t$ suffixes $\rho_1, \ldots, \rho_d$ (where $d \le
\floor{\rl/t}$) such that $|F_{\rho_1 \concat \cdots \concat \rho_d}| < 2^{-d}
|F_\emptystream| \le 2^{-d} n$. Since $\cA$ must produce {\em some} output at
time $\rl$, this would be a contradiction for $d \ge \log n$. Thus, for a
setting of $t = \Theta(\rl/\log n)$, situation \eqref{eq:non-shrinkage} must
occur, implying a lower bound of $z = \Omega(\rl/\log n)$.

\mypara{Relaxing ``all outputs'' to ``common outputs''} Examining the above
argument closely shows where it fails for pseudo-deterministic algorithms. In
constructing an $\avoid$ protocol above, we needed the key property that
$F_\stream$ can be determined from just the {\em state} of $\cA$ upon
processing $\stream$.  For pseudo-deterministic algorithms, if we simply
define $F'_\stream$ to be ``the set of all \emph{canonical} outputs at time
$\rl$ for continuations of $\stream$,'' we cannot carry out the above proof
plan because this $F'_\stream$ cannot be computed reliably from a single
state: given a random state $\sigma$ associated to $\stream$, on average a
$\delta$ fraction of the outputs might be incorrect and have arbitrary values;
even a single bad output could corrupt the union calculation!

To work around this issue, we replace $F_\stream$ with a more elaborate
recursive procedure \textsc{FindCommonOutputs}, (or \fco for short) that
computes the ``most common outputs'' at time $\rl$ for a certain distribution
over continuations of $\stream$. To explain this, let us imagine positions $1$
through $\rl$ in the input stream as being divided into $d$ contiguous ``time
intervals.'' In the deterministic proof, these intervals were of length $t$
each. Given a stream $\stream$ that occupies the first $d-k$ of these
intervals, $F_\stream$ can be thought of as the output of a procedure
\textsc{FindAllOutputs} (or \fao for short) where $\fao(\cA, \stream, k)$
operates as follows: for each setting $\rho$ of the $(d-k+1)$th time interval,
call $\fao(\cA, \stream \concat \rho, k-1)$ and return the union of the sets
so obtained. In the base case, $\fao(\cA, \stream, 0)$ takes a stream $\stream
\in [n]^\rl$ and returns the singleton set $\{\cA(\stream)\}$. The
deterministic argument amounts to showing that, with interval lengths $t =
\Theta(z)$, the set $\fao(\cA, \stream, k)$ has cardinality $\ge 2^k$; since
$\fao(\cA, \emptystream, d)$ has cardinality $\le n$, this bounds $d \le \log
n$, which lower-bounds $z$.

For our pseudo-deterministic setting, we use time intervals as above and we
design an analogous procedure $\fco(B, C, \stream, k)$ that operates on a
function $B \colon [n]^\rl \to [n]$ (roughly corresponding to an \mif
algorithm), a matrix $C$ of random thresholds,\footnote{The use of random
thresholds is a standard trick for robustly computing quantities in the
presence of noise.} and a stream $\stream$ of length $\le \rl$ that occupies
the first $d-k$ time intervals. The recursive structure of $\fco(B, C,
\stream, k)$ is similar to \fao, but crucially, the sets computed by the
recursive calls $\fco(B, C, \stream \concat \rho, k-1)$ are used differently.
Instead of simply returning their union, we use these sets to collect
statistics about the outputs in $[n]$ and return only those that are
sufficiently common. The thresholds in $C$ control the meaning of
``sufficiently common.''

The function $B$ provided to \fco can be either the canonical output function
$\Pi$ of the given pseudo-deterministic algorithm $\cB$ or a deterministic
algorithm $A \sim \cB$ obtained by fixing the random coins of $\cB$. We will
show that:
\begin{itemize}

  \item With high probability over $C$ and the randomness of $\cB$, $\fco$
  will produce the same outputs on $\Pi$ and $\cB$. In other words, $\fco$ is
  robust to noise (i.e., to algorithm errors).

  \item When applied to the canonical algorithm, the cardinalities of the sets
  returned by $\fco$ will grow exponentially with $k$. Equivalently, similar
  to $|F_\stream|$ from the deterministic proof, the cardinality of
  $\fco(\stream, \ldots)$ will shrink exponentially as the length $|\stream|$
  grows.  Ultimately, this is proven by implementing $\avoid$ using $\fco$ on
  the actual algorithm as a subroutine. Critically, this implementation uses
  the fact that the recursive calls to $\fco$ w.h.p. produce the same output
  on $\Pi$ and $\cB$.

  \item The argument can be carried out with all but one of the $d$ time
  intervals being of length $\approx \Theta(z)$. If $z$ were too small, $d$
  would be large enough that for the empty stream prefix we would have
  $|\fco(\emptystream,\ldots)| > n$, which contradicts $\fco(\ldots) \subseteq
  [n]$; this lets us derive a lower bound on $z$.

\end{itemize}

\mypara{Error amplification and the case $n \gg \rl$} One technical issue that
arises is that the correctness of \fco requires $\cB$'s error probability to
be as small as $1/n^{\Omega(\log n)}$.  Fortunately, even if the original
error probability was $1/3$, we can reduce it to the required level since
pseudo-deterministic algorithms allow efficient error reduction by independent
repetition. A second technical point is that a $z$-space pseudo-deterministic
algorithm can be shown to have only $O(2^{z})$ possible outputs; so if $n \gg
\rl$, we can sometimes obtain a stronger lower bound by pretending that $n$ is
actually $O(2^{z})$. This is formalized by a simple encoding argument.

\section{Preliminaries}\label{sec:prelim}

\mypara{Notation} Throughout this paper, $\log x = \log_2 x$, while $\ln x =
\log_e x$. The set $\NN$ consists of all positive integers; $[k] :=
\{1,2,\ldots,k\}$; and $[a,b)$ is a half open interval of real numbers. For a
condition or event $E$, the symbol $\indic_{E}$ takes the value $1$ if $E$
occurs and $0$ otherwise. The sequence (stream) obtained by concatenating
sequences $a$ and $b$, in that order, is denoted $a \concat b$. For a set $S$
of elements in a totally ordered universe, $\sort(S)$ denotes the sequence of
elements of $S$ in increasing order; $\binom{S}{k}$ is the set of $k$-element
subsets of $S$; and $\SSD{S}{k} = \{\sort(Y) : Y \in \binom{S}{k}\}$. We
sometimes extend set-theoretic notation to vectors and sequences; e.g., for $y
\in [n]^t$, write $y \subseteq S$ to mean that $\forall i \in [t]: y_i \in S$.
For a set $X$, $\triangle[X]$ denotes the set of probability distributions
over $X$, while $A \in_R X$ indicates that $A$ is chosen uniformly at random
from $X$. When naming probability distributions, we will use either
calligraphic letters (e.g., $\cA,\cD$) or the letter $\mu$; $A \sim \mu$ means
that $A$ is drawn from the distribution $\mu$.

\subsection{Useful Lemmas}\label{subsec:useful-lemmas}

These will be used in following sections. When no external work is cited, a proof is given for completeness either here or in \Cref{subsec:useful-proofs}.

\begin{lemma}[Multiplicative Azuma's inequality]\label{lem:azumanoff}
  Let $X_1,\ldots,X_t$ be $[0,1]$ random variables, and $\alpha \ge 0$. If, for all $i \in [t]$, $\EE[X_i \mid X_1,\ldots,X_{i-1}] \le p_i$, then
  \begin{align*}
    \Pr\left[ \sum_{i=1}^{t} X_i \ge (1 + \alpha) \sum_{i =1}^{t} p_i \right] 
    \le \exp\left( - ((1+\alpha) \ln (1+\alpha) - \alpha) \sum_{i =1}^{t} p_i\right) 
    \le \exp\left( - \frac{\alpha^2}{2 + \alpha} \sum_{i =1}^{t} p_i\right) \,.
  \end{align*}
  On the other hand, if for all $i$, $\EE[X_i \mid X_1,\ldots,X_{i-1}] \ge p_i$, then
  \begin{align*}
    \Pr\left[ \sum_{i=1}^{t} X_i \le (1 - \alpha) \sum_{i =1}^{t} p_i \right]
    \le \exp\left( - ((1-\alpha) \ln (1-\alpha) + \alpha) \sum_{i =1}^{t} p_i\right) 
    \le \exp\left( - \frac{\alpha^2}{2} \sum_{i =1}^{t} p_i\right) \,.
  \end{align*}
\end{lemma}

\noindent In contrast to the above, the ``usual'' form of Azuma's inequality uses a
martingale presentation and gives an additive-type bound.

\begin{lemma}[Chernoff bound with negative association, from \cite{JoagDevP83}]\label{lem:chernoff-neg-assoc} The standard multiplicative Chernoff bounds work with negatively associated random variables.
\end{lemma}

\begin{lemma}[Error amplification by majority vote]\label{lem:error-reduction-by-vote}
  Let $\epsilon \le \delta \le 1/3$. Say $X$ is a random variable, and $v$ a value with $\Pr[X = v] \ge 1 - \delta$. If $X_1,\ldots,X_p$ are independent copies of $X$, then the most common value in $(X_1,\ldots,X_p)$ will be $v$ with probability $\ge 1 - \epsilon$, for $\epsilon = \left(2\delta\right)^{p / 30}$.
\end{lemma}

\noindent The above is a standard lemma, useful for trading error for space
for algorithms with a single valid output.

As outlined in \Cref{subsec:rt-lb-overview}, $\avoid(m,a,b)$ is a one-way
communication problem, wherein player Alice has a set $A \in
\smash{\binom{[m]}{a}}$, and should send a short message to player Bob, who
should use the message to output a set $B \in \smash{\binom{[m]}{b}}$ that is
disjoint from $A$. In the randomized $\delta$-error setting, this disjointness
should hold with probability $\ge 1-\delta$.

\begin{theorem}[\avoid communication lower bound, from \cite{ChakrabartiGS22}]\label{lem:avoid-lb}
  Suppose there exists a randomized protocol for $\avoid(m,a,b)$, in which
  Alice communicates $\le K$ bits, that is $\delta$-error either on a
  worst-case input or when Alice's input is chosen uniformly at random from
  $\binom{[t]}{a}$. Then
  \begin{align*}
    K \ge \frac{a b}{t \ln 2} + \log(1-\delta) \,.
  \end{align*}
\end{theorem}

Using \Cref{lem:avoid-lb}, it is straightforward to derive the following lower
bound for robust algorithms for $\mif(n,\rl)$, as was done in
\cite{Stoeckl23}. The proof is short, yet instructive, so we outline it here.

\begin{theorem}[Adversarially robust random oracle lower bound, from \cite{Stoeckl23}]\label{lem:ext-robust-lb}
  If there exists a $\delta$-error adversarially robust random oracle
  algorithm for $\mif(n,\rl)$ using $z$ bits of space, then
  \begin{align*}
    z \ge \frac{\rl^2}{4 n \ln 2} + \log(1-\delta) \,.
  \end{align*}
\end{theorem}
\begin{proof}
  Such an algorithm $\cA$ yields a $z$-bit $\delta$-error randomized protocol
  for $\avoid(n, \ceil{\rl/2}, \floor{\rl/2}+1)$ wherein Alice feeds her set
  $A$ into $\cA$, sends the state of $\cA$ to Bob, and Bob extracts set $B$ by
  using the ``echo'' adversarial strategy, i.e., repeatedly asking $\cA$ for
  an output item and feeding that item back as the next input. The result now
  follows by appealing to \Cref{lem:avoid-lb}.
\end{proof}

We also note the following simple lower bound.

\begin{lemma}\label{lem:rt-triv-logr-lb}
  For every $\delta < 1$, if there exists a $\delta$-error random tape
  algorithm for $\mif(n,\rl)$ using $z$ bits of space, then $z \ge
  \log(\rl+1)$.
\end{lemma}
\begin{proof}
  Each state of a random tape streaming algorithm $\cA$ has a unique
  associated output value. If $z < \log(\rl + 1)$, then $\cA$ has at most
  $\rl$ states. Let $H$ be the set of outputs associated with these states; so
  $|H| \le \rl$. When sent a stream containing each element of $H$, $\cA$ will
  fail with probability $1$ because every output it could make is wrong.
\end{proof}

By the remarks in \Cref{subsec:models} following the definitions of the models of
computation, \Cref{lem:ext-robust-lb} also applies to random tape and random
seed algorithms and \Cref{lem:rt-triv-logr-lb} also applies to random seed
algorithms.

\section{The Random Tape Lower Bound}\label{sec:rt-lb}

This section presents our first and perhaps most important lower bound, of
which \Cref{res:rt-lb} is a consequence. We shall carry out the proof plan
outlined in \Cref{subsec:rt-lb-overview}, designing a recursive adversary to
foil a given random-tape \mif algorithm $\cA$ that runs in $z$ bits of space.
Correspondingly, our lower bound proof will be inductive.

\subsection{Setup and Base Case}

Recall that the adversary organizes the $\rl$-length input to be fed into
$\cA$ as a random prefix of length $\rl/2$ followed by another $\rl/2$ inputs
divided into several phases, each consisting of $t = O(\rl/z)$ inputs. The
adversary's eventual goal is to identify a particular phase and a
corresponding sub-universe $W \subseteq [n]$ so that $\cA$, when suitably
conditioned and restricted to that phase, yields a sub-algorithm $\cB$ that is
good for \mif for inputs from $W$. However, we will need to generalize the notion of a
``good'' \mif algorithm, because this sub-algorithm might produce outputs outside of $W$,
even when fed inputs from $W$. To aid our analysis, we will make $\cB$ abort
anytime it would have produced an output outside of $W$. In what follows, it
will be important to maintain a distinction between these aborts and actual
mistakes.

\begin{definition}\label{def:mif-rt-err-complexity}
  An algorithm $\cA$ for $\mif(n,\rl)$ can fail in either of two ways. It may
  make an incorrect output, or {\em mistake}, if outputs an element in $[n]$
  that {\em is} in its input stream (i.e., not missing). It may also {\em
  abort}, by outputting a special value $\bot$ (where $\bot \notin [n]$)
  and stopping its run.
 
  For integers $n,\rl,z$ with $1 \le \rl < n$, and $\gamma \in [0,1]$, let
  $\Algs(n,\rl,\gamma,z)$ be the set of all $z$-bit \emph{random tape}
  algorithms for $\mif(n,\rl)$ which on \emph{any} adversary abort
  with probability $\le \gamma$. Define
  \begin{align*}
    \Mstk(n,\rl,\gamma,z) := \min_{\cA \,\in\, \Algs(n,\rl,\gamma,z)} 
      \delta_{\max}(\cA, n,\rl) \,,
  \end{align*}
  where $\delta_{\max}(\cA, n,\rl)$ is the maximum probability, over all
  possible adversaries, that $\cA$ makes a mistake. As a consequence of the
  definition, $\Mstk(n,\rl,\gamma,z)$ is non-increasing in $\gamma$ and $z$.
\end{definition}

We shall establish \Cref{res:rt-lb} (concretely, \Cref{thm:rt-lb-intro}) using
a proof by induction. The base case is straightforward and handled by the
following lemma.

\begin{lemma}[Base case]\label{lem:mif-rt-lb-base-case}
  If a random tape adversarially robust algorithm $\mif(n,\rl)$ uses at most
  $\rl^2/(16 n \ln 2)$ bits of space and aborts with probability $\le \frac12$,
  then it makes a mistake with probability $\ge \frac14$. Equivalently,
  \begin{align}
    \Mstk(n,\rl,\gamma,z) 
    \ge \frac{1}{4} \indic_{z \le \rl^2 / (16 n \ln 2)} \indic_{\gamma \le 1/2} \,.
    \label{eq:mif-rt-lb-base-eq}
  \end{align}
\end{lemma}
\begin{proof}
  Suppose that $\gamma \le \frac12$. Let $\cA \in \Algs(n,\rl,\gamma,z)$ be an
  algorithm with mistake probability $\delta \le \frac14$. Blurring the
  distinction between aborts and mistakes, and applying the space lower bound
  for random {\em oracle} algorithms (\Cref{lem:ext-robust-lb}), we obtain
  \begin{align*}
    z \ge \frac{\rl^2}{4 n \ln 2} + \log(1 - \gamma - \delta) 
    \ge \frac{\rl^2}{4 n \ln 2} - 2 > \frac{\rl^2}{16 n \ln 2} \,.
  \end{align*}
  The latter inequality holds because $z \ge 1$ (trivially) and we have
  $\max(1, x-2) > x/4$. Taking the contrapositive, if $z \le \rl^2 / (16 n \ln
  2)$, then $\delta > 1/4$. This proves \cref{eq:mif-rt-lb-base-eq}.
\end{proof}

\subsection{The Induction Step} \label{subsec:rt-ind-step}

The induction step consists of a reduction, using an adaptive adversary
described in \Cref{alg:rt-adversary-step}, to prove a lower bound on the
mistake probability. This is formalized in the next lemma and the rest of
\Cref{subsec:rt-ind-step} is devoted to its proof.

\begin{lemma}[Induction lemma]\label{lem:mif-rt-lb-step}
  Let $1 \le \rl < n$ and $z$ be integers. Define, matching definitions in \Cref{alg:rt-adversary-step},
  \begin{align}
    w := 2 \floor{32 \frac{z n}{\rl}} \qquad \text{and} \qquad 
    t := \floor{\frac{\rl}{64 z}} \,.
    \label{eq:w-t-def}
  \end{align}
  If $z \ge 8$ and $t < w$, then:
  \begin{align}
    \Mstk\Big(n,\rl,\frac{1}{2},z\Big) \ge \min\left( \frac{\rl}{2^7 n k}, \,
      \frac{1}{4} \Mstk\Big(w,t,\frac{1}{2},z\Big) \right) \,. 
      \label{eq:mif-rt-lb-step-err}
  \end{align}
\end{lemma}

\mypara{The Initial Random Prefix} The adversary begins by sending $\cA$ a
uniformly random sequence $X \in_R \SSD{[n]}{q}$, i.e., a sorted sequence of
$q$ distinct elements of $[n]$; we'll eventually use $q = \ceil{\rl/2}$.  Let
$F$ be the {\em random function} where, for $x \in \SSD{[n]}{q}$, $F(x)$ is
the random state reached by $\cA$ upon processing $x$, starting at its initial
state; note that $F$ is determined by the transition function $T$ of $\cA$
(see~\Cref{subsec:models}).  Let $\Sigma$ be the set of states of $\cA$, so
$|\Sigma| = 2^z$. For each $\sigma \in \Sigma$, define
\begin{align}
  H_{\sigma} := \left\{ i\in[n]:\, \Pr[i\in X \mid F(X) = \sigma] \le \frac{q}{4n} \right\} \,,
  \label{eq:H-def}
\end{align}
which we can think of as the set of inputs that are ``unlikely'' to have been
seen given that $\cA$ has reached $\sigma$.  A key part of our proof of
\Cref{lem:mif-rt-lb-step} is the following lemma, which says that if $\Sigma$
is small, then $\cA$ doesn't have enough space to mark too many inputs as
unlikely, so $H_{F(X)}$ is likely to be small. The lemma can be seen as a
smoothed variant of the communication lower bound for $\avoid$.

\begin{lemma}\label{lem:forward-avoid}
  Let $\Sigma$ be a set with $|\Sigma| \le 2^z$, let $1 \le q \le n$ be
  integers, let $F$ be a random function that maps each sequence in
  $\SSD{[n]}{q}$ to a random element of $\Sigma$, and let $X \in_R
  \SSD{[n]}{q}$, chosen independently of $F$. For each $\sigma\in\Sigma$,
  define $H_\sigma$ as in \cref{eq:H-def}. Then, for all
  $\alpha\in\left(0,1\right)$,
  \begin{align*}
    \Pr\left[\left|H_{F(X)}\right|\ge\hat{w}\right] \le \alpha
    \qquad\text{where}\qquad
    \hat{w} := \left\lceil \frac{2\ln2}{1-\ln2} \frac{z+1+\log\frac{1}{\alpha}}{q}\, n\right\rceil \,.
  \end{align*}
\end{lemma}
\begin{proof}
   Consider a specific $\sigma\in\Sigma$. By linearity of expectation:
  \begin{align*}
    \EE\left[\sum_{i\in H_{\sigma}}\indic_{i \in X}\Bigm| F\left(X\right)=\sigma\right]=\sum_{i\in H_{\sigma}}\Pr\left[i\in X\mid F\left(X\right)=\sigma\right]\le\frac{q}{4n}\left|H_{\sigma}\right| \,.
  \end{align*}
  Then by Markov's inequality,
  \begin{align*}
    \Pr\left[\sum_{i\in H_{\sigma}}\indic_{i \in X}\ge2\cdot\frac{q}{4n}\left|H_{\sigma}\right|\Bigm| F\left(X\right)=\sigma\right]
      &\le \frac{1}{2} \\
      \text{which implies}\qquad
      \Pr\left[\sum_{i\in H_{\sigma}}\indic_{i \in X}\le\frac{q}{2n}\left|H_{\sigma}\right|\Bigm| F\left(X\right)=\sigma\right]
        &\ge \frac{1}{2} \,.
  \end{align*}
  Since $X$ is drawn uniformly at random from $\SSD{[n]}{q}$, the random
  variables $\{\indic_{i \in X}\}_{i \in [n]}$ are negatively associated, with
  $\EE \indic_{i \in X} = q/n$ for each $i \in [n]$. For any set $A
  \subseteq [n]$, we use the multiplicative Chernoff bound
  (\Cref{lem:chernoff-neg-assoc}) to bound the probability that $X$'s overlap
  with $A$ is much smaller than the expected value:
  \begin{align*}
    \Pr\left[\sum_{i\in A}\indic_{i \in X}\le\left(1-\frac{1}{2}\right)\cdot\frac{q}{n}|A|\right]
      \le\left(\frac{e^{-1/2}}{(1/2)^{1/2}}\right)^{(q/n)|A|}
      =\exp\left(-\frac{1}{2}(1-\ln 2)\frac{q}{n}|A|\right) \,.
  \end{align*}
  We now bound
  \begin{align*}
    \Pr\left[F\left(X\right)=\sigma\right]
    &\le \frac{\Pr\left[\sum_{i\in H_{\sigma}}\indic_{i \in X} \le \frac{q}{2n}|H_{\sigma}|\right]}%
      {\Pr\left[\sum_{i\in H_{\sigma}}\indic_{i \in X} \le \frac{q}{2n}|H_{\sigma}| \bigm| F(X)=\sigma\right]}
    \le 2\exp\left(-\frac{1}{2}(1-\ln 2)\frac{q}{n}|H_{\sigma}|\right) \,.
  \end{align*}
  Finally, let $B = \{\sigma\in\Sigma:\, |H_{\sigma}|\ge\hat{w}\}$. Then
  \begin{align*}
    \Pr\left[\left|H_{F(X)}\right|\ge\hat{w}\right]
    &=\sum_{\sigma\in B}\Pr\left[F\left(X\right)=\sigma\right] \\
    &\le 2^{z}\cdot2\exp\left(-\frac{1}{2}(1-\ln 2)\frac{q}{n}\hat{w}\right) \\
    &\le 2^{z}\cdot2\exp\left(-\left(z+1+\log\frac{1}{\alpha}\right)\ln2\right) 
    \le 2^{z}\cdot2\cdot2^{-\left(z+1+\log\frac{1}{\alpha}\right)}
    = \alpha \,. \qedhere
	\end{align*}
\end{proof}

\mypara{The Recursive Phases and Win-Win Argument} Let $\rho$ denote the
random state in $\Sigma$ reached by $\cA$ upon processing the random sequence
$X$. Having observed $\cA$'s outputs (transcript) in response to $X$, the
adversary can compute $\cD$, the distribution of $\rho$ conditioned on this
transcript. The adversary has $\floor{\rl/2}$ more items to send to $\cA$ and
it organizes these into phases of $t$ items each (see \cref{eq:w-t-def}).
Recall, from the discussion in \Cref{subsec:rt-lb-overview}, that the
adversary's goal is to identify a suitable sub-universe $W \subseteq [n]$ so
that, in some phase, $\cA$ can be seen as solving $\mif(|W|,t)$ in this
sub-universe.  As it chooses a suitable input sequence for each phase, the
adversary maintains the following objects to guide the choice:
\begin{itemize}
  \item the evolving transcript of $\cA$ given the inputs chosen so far;
  \item the corresponding distribution $\cD \in \triangle[\Sigma]$;
  \item a set $Q \subseteq \Sigma$ where, for each $\sigma \in Q$, the set
  $H_\sigma$ will be useful for determining $W$.
\end{itemize}
If the current set $Q$ leads to a suitable $W$, the adversary's strategy for
the next phase is to recursively run an optimal sub-adversary for
$\mif(|W|,t)$. If not, then (we shall show that) over the next phase, the
adversary will be able to significantly shrink the set $Q$.

To make these ideas more concrete, we introduce some terminology below; these
terms show up in \Cref{alg:rt-adversary-step}, which spells out the adversary's
actions precisely.
\begin{definition}[Divisive output sequence, Splitting adversary]\label{def:rt-defs}
  Let $\cA$, $\Sigma$, and $H_\sigma$ be as above and let $Q \subseteq
  \Sigma$.  A sequence (of ``outputs'') $y \in [n]^t$ is said to be
  \textsc{divisive} for $Q$ if $|\{\sigma \in Q : y \subseteq H_\sigma\}| \le
  \frac{1}{2} |Q|$.

  Say $\Upsilon$ is a $t$-length deterministic adversary, i.e., a function
  $\Upsilon \colon [n]^{\le t-1} \to [n]$.  For each $\sigma \in \Sigma$, let
  $\Outs(\sigma, \Upsilon)$ be the random variable in $[n]^t \cup \{\bot\}$
  that gives the output if we run $\cA$, starting at state $\sigma$, against
  the adversary $\Upsilon$,\footnote{This means that if, after processing a
  few inputs, the algorithm has output sequence $v \in [n]^\star$, its next
  input will be $\Upsilon(v)$.} with $\bot$ indicating that $\cA$ aborts. 
  We say that $\Upsilon$ is $\alpha$-\textsc{splitting} for $Q$ with respect
  to a distribution $\cD \in \triangle[\Sigma]$ if
  \begin{align*}
    \Pr_{S \sim \cD}[\text{$\Outs(S, \Upsilon)$ is divisive for $Q$}] \ge \alpha \,,
  \end{align*} 
  with the convention that the value $\bot$ \emph{is} divisive.\footnote{The proof can also be made to work if one assumes $\bot$ is not divisive, but gives a weaker and more complicated result.}
\end{definition}

Notice that, knowing $\cD$ and $Q$, the adversary can determine whether a
given $\Upsilon$ is $\alpha$-splitting. This is implicit in 
\Cref{step:rt-pick-split-adv} of \Cref{alg:rt-adversary-step}.

\begin{algorithm}[!ht]
  \caption{An adversary for a random tape $\mif(n,\rl)$ algorithm.}
  \label[listing]{alg:rt-adversary-step}

  \begin{algorithmic}[1]
    \Statex \underline{\textsc{Adversary} (against algorithm $\cA$)}
    \vspace{.5\baselineskip}
    \State $w \gets 2 \floor{32 z n/\rl};~ h_{\max} \gets 32 z;~ t \gets \floor{\rl/(2 h_{\max})}$
    \State $X \gets$ a uniformly random sequence in $\SSD{[n]}{\ceil{\rl/2}}$.\label{step:rt-pick-v}
    \State \textbf{send} $X$ to $\cA$ and record the transcript of outputs \label{step:rt-send-v}
    \State compute $H_\sigma$ for each state $\sigma$ using \cref{eq:H-def}, with $q = \ceil{\rl/2}$
    \State $Q_0 \gets \{\sigma \in \Sigma : |H_\sigma| \le \frac{1}{2} w\}$\label{step:rt-q0-def}
    \For{$h$ in $1,\ldots,h_{\max}$}
      \State $\cD \gets$ distribution over $\Sigma$ conditioned on the cumulative transcript so far
      \If{$\exists$ a $t$-length deterministic adversary $\Upsilon$ that is $\frac12$-splitting for $Q_{h-1}$ w.r.t. $\cD$}\label{step:rt-pick-split-adv}
          \State \textbf{run} $\Upsilon$ against $\cA$ and gather the transcript of outputs $y \in [n]^t$ \label{step:rt-send-splitting}
          \State $Q_h \gets \{\sigma \in Q_{h-1}:\, y \subseteq H_\sigma\}$\label{step:rt-qh-def} \Comment{have a $\ge \frac12$ chance that $|Q_h| \le \frac{1}{2} |Q_{h-1}|$}
          \If{$Q_h = \emptyset$} \textbf{fail} \label{step:rt-abort-qempty} \EndIf
      \Else
        \State $W \gets \{i \in [n]: |\{\sigma \in Q_{h-1} : i \in H_\sigma\}|
        \ge \frac{1}{2} |Q_{h-1}|\}$ \label{step:rt-wkp1-selection} \Comment{will show that $|W| \le w$}
        \State $W' \gets W$ plus $w - |W|$ padding elements
        \State define algorithm $\cB$ to behave like $\cA$ conditioned on the transcript of inputs and outputs so far
        \State modify $\cB$, changing every output outside $W'$ to $\bot$ \label{step:rt-modify-abort} \Comment{thus $\cB$ will abort in these cases}
        \State let $\Xi$ be an adversary using inputs from $W'$, maximizing the probability that $\cB$ makes a mistake \label{step:rt-pick-subadv}
        \Statex \Comment{can be computed using brute-force search}
        \State \textbf{run} adversary $\Xi$, sending $t$ inputs in $W'$\label{step:rt-run-subadv}
        \State \textbf{return} \Comment{succeeded in causing $\cA$ to have a high enough error probability}
      \EndIf
    \EndFor
    \State \textbf{fail}\label{step:rt-abort-hmax}
  \end{algorithmic}
\end{algorithm}

We proceed to prove the induction lemma.

\begin{proof}[Proof of \Cref{lem:mif-rt-lb-step}]
  To prove the lower bound in \cref{eq:mif-rt-lb-step-err}, we show that when
  the adversary in \Cref{alg:rt-adversary-step} is run against a $z$-bit
  random-tape algorithm $\cA$ for $\mif(n,\rl)$ which has $\le \frac{1}{2}$
  worst-case probability of aborting, the probability that $\cA$ makes a
  mistake is at least the right hand side of \cref{eq:mif-rt-lb-step-err}.
  Note that the adversary feeds at most $\ceil{\rl/2} + t h_{\max} =
  \ceil{\rl/2} + \floor{\rl/ (2 h_{\max})} h_{\max} \le \ceil{\rl/2} +
  \floor{\rl/2} = \rl$ inputs to $\cA$.
  
  Consider a run of the adversary against $\cA$. This is a random process,
  with some of the randomness coming from the adversary's choices
  (\Cref{step:rt-pick-v}) and some coming from $\cA$'s internal randomness.
  Let $\rho$ be the state of $\cA$ after $X$ is sent. We now define a number
  of events, as follows.
  \begin{itemize}
    \item $B_{\unsafesc}$ occurs if $\cA$ produces an output in $[n] \setminus H_\rho$.
    \item $B_{\bigsc}$ occurs if the state $\rho$ has $|H_\rho| > \frac{1}{2} w$.
    \item $B_{\emptysc}$ occurs if the adversary fails at \Cref{step:rt-abort-qempty}.
    \item $B_{\timeoutsc}$ occurs if the adversary fails at \Cref{step:rt-abort-hmax}.
    \item $B_{\abortsc}$ occurs if $\cA$ aborts \emph{before} the adversary reaches \Cref{step:rt-run-subadv}.
    \item $R_{\errorsc}$ occurs if $\cA$ makes a mistake \emph{while} the adversary is executing \Cref{step:rt-run-subadv}.
  \end{itemize}
  
  We will consider each of the events listed above.
  \iflipics\begin{itemize}\else
  \begin{itemize}[wide, font=\bfseries]\fi
  \item[(Event $B_{\unsafesc}$)~] 
    If $B_{\unsafesc}$ occurs, then some $i \in [n] \setminus H_\rho$ is
    output and, in view of \cref{eq:H-def}, the set $X$ from
    \Cref{step:rt-pick-v} has the property that
    \[
      \Pr[i \in X \mid \cA \text{ reaches } \rho] \ge \frac{\ceil{\rl/2}}{4n}
      \ge \frac{\rl}{8n} \,.
    \]
    Consequently, the probability that $\cA$ makes a mistake by producing an
    output from $X$ (which would be a non-missing item) is $\ge
    \Pr[B_{\unsafesc}] \cdot \rl/(8 n)$. 

    If $\Pr[B_{\unsafesc}] > 1/(16)$, we then have
    $\Mstk(n,\rl,\frac{1}{2},z) \ge \rl/(2^7 n)$, which implies
    \cref{eq:mif-rt-lb-step-err} leaving nothing more to prove. Therefore, for
    the rest of this proof we will consider the case in which
    \begin{align}
      \Pr[B_{\unsafesc}] \le \frac{1}{16} \,.
      \label{eq:b-unsafe-bound}
    \end{align}

  \item[(Event $B_{\bigsc}$)~]
    We apply \Cref{lem:forward-avoid} with $q = \ceil{\rl/2}$, $\alpha =
    1/16$, $F$ being the random function that maps each $x \in
    \SSD{[n]}{q}$ to the random state reached by $\cA$ upon processing $x$,
    starting at its initial state, and
    \begin{align*}
      \hat{w} = \ceil{\frac{2\ln2}{1-\ln2} \frac{z+1+\log 16}{\ceil{\rl/2}}\, n} 
        &\le 1 + \floor{\frac{2\ln2}{1-\ln2} \frac{z+1+\log 16}{\ceil{\rl/2}}\, n} \\
        &\le 1 + \floor{\frac{8 \ln 2}{1 - \ln 2} \frac{z + 1 + \log 16}{\rl}\, n} \\
        &\le 1 + \floor{19 \frac{z + 5}{\rl} n} 
          \le 1 + \floor{32 \frac{z n}{\rl}} = \frac{1}{2} w + 1 \,,
    \end{align*}
    since $z \ge 8$. As $\frac{1}{2} w$ is an integer,
    \begin{align}
      \Pr\left[B_{\bigsc}\right]
      = \Pr\left[|H_\rho| > \frac{1}{2} w\right] 
      = \Pr\left[|H_\rho| \ge \frac{1}{2} w + 1\right] 
      \le \Pr\left[|H_\rho| \ge \hat{w}\right] 
      \le \frac{1}{16} \,.
      \label{eq:b-big-bound}
    \end{align}
    
  \item[(Event $B_{\abort}$)~] 
    There are exactly exactly three spots in \Cref{alg:rt-adversary-step}
    where the algorithm can abort:
    \Cref{step:rt-send-v,step:rt-send-splitting,step:rt-run-subadv}, when the
    adversary is feeding it inputs.  By assumption, $\cA$'s worst case
    probability of aborting, against {\em any} adversary, is at most
    $\frac12$. Therefore, in particular,
    \begin{align}
      \Pr[B_{\abort}] \le \frac{1}{2} \,.
      \label{eq:b-abort-bound}
    \end{align}
    When the algorithm does abort, we stop running the adversary, so
    $B_{\abortsc}$, $B_{\emptysc}$, and $B_{\timeoutsc}$ are mutually exclusive.
    
  \item[(Event $B_{\emptysc}$)~]
    For this to happen, at some point we must have $Q_h = \emptyset$. At that
    point, the state $\rho$ must not be in $Q_h$, in which case either $\rho
    \notin Q_0$ or $\rho$ was filtered out of $Q_h$ on \Cref{step:rt-qh-def}.
    By the definition of $Q_0$, $\rho \notin Q_0$ iff $B_{\bigsc}$ holds. On
    the other hand, filtering $\rho$ out of $Q_h$ requires that the algorithm
    produce an output outside $H_{\rho}$, which can only happen if either
    $B_{\unsafesc}$ or $B_{\abortsc}$ occurs. Since $B_{\abortsc}$ is mutually
    exclusive with $B_{\emptysc}$, we conclude that
    \begin{align}
      B_{\emptysc} \Rightarrow B_{\unsafesc} \lor B_{\bigsc} \,. \label{eq:b-empty-bound}
    \end{align}
  
  \item[(Event $B_{\timeoutsc}$)~]
    We bound the probability that the adversary will fail using
    \Cref{step:rt-abort-hmax}. For this to
    happen, the adversary must have picked $h_{\max}$ splitting adversaries,
    but fewer than $z+1$ of them must have
    produced a divisive output. (If there is a divisive output in round $h$,
    then $|Q_h| \le \frac{1}{2} |Q_{h-1}|$; if not, then $|Q_h| \le
    |Q_{h-1}|$. Thus with $z+1$ divisive outputs, $|Q_{h_{\max}}| \le |Q_0| /
    2^{z+1} \le |\Sigma|/2^{z+1} \le \frac{1}{2} < 1$, in which case
    the adversary would have failed at \Cref{step:rt-abort-qempty} instead.)
  
    For each $h \in [h_{\max}]$, let $X_h$ be the $\{0,1\}$ indicator random
    variable for the event that a divisive output is found in the $h$th step.
    (If the $h$th step did not occur or no splitting adversary was found, set
    $X_h = 1$.\footnote{Note that this definition accounts for the cases where the
    algorithm aborts: if it aborts on $h$, by \Cref{def:rt-defs} this is interpreted
    as divisive, and the following steps do not occur, so $X_h=X_{h+1}=\ldots,X_{h_{\max}}=1$.
    Then the event $\{\sum X_h < z+1\}$ slightly overestimates the probability
    of $B_{\timeoutsc}$; it would be more accurate to have $\cA$ aborting produce 
    $X_h = \infty$.}) Since in the $h$th step, a splitting adversary for the
    distribution for the current state of the algorithm, conditioned on the
    transcript so far, is chosen, then $\EE[X_h \mid X_1,\ldots,X_{h-1}] \ge
    1/2$.
    Applying \Cref{lem:azumanoff} gives:
    \begin{align*}
      \Pr\Big[\sum_{h \in [h_{\max}]} X_h < z + 1\Big] &= \Pr\left[\sum_{h \in [h_{\max}]} X_h \le \left(1 - \left(1-\frac{2 z}{h_{\max}} \right) \right)\frac{h_{\max}}{2} \right]\\
        &\le \exp\left( - \frac{1}{2} \left(1-\frac{2 z}{h_{\max}} \right)^2 \frac{h_{\max}}{2}\right) \\
        &\le \exp\left( - \frac{1}{8} \frac{h_{\max}}{2}\right) &&\hspace{-2.5cm}\text{\color{black!50!white}since $h_{\max} = 32 z \ge 4 z$} \\
        &\le \frac{1}{8} \,. &&\hspace{-2.5cm}\text{\color{black!50!white}since $h_{\max} \ge 16 \ln 8$}
    \end{align*}
    Thus $\Pr[B_{\timeoutsc}] \le 1/8$.

    Combining this  with \cref{eq:b-unsafe-bound,eq:b-big-bound,eq:b-abort-bound,eq:b-empty-bound} gives:
    \begin{align}
      \Pr[B_{\emptysc} \lor B_{\timeoutsc} \lor B_{\abortsc}]
        &\le \Pr[B_{\unsafesc} \lor B_{\bigsc} \lor B_{\timeoutsc} \lor B_{\abortsc}] \nonumber\\
        &\le \Pr[B_{\unsafesc}] + \Pr[B_{\bigsc}] + \Pr[B_{\timeoutsc}] + \Pr[B_{\abortsc}] \nonumber\\
        &\le \frac{1}{16} + \frac{1}{16} + \frac{1}{8} + \frac{1}{2} = \frac{3}{4} \,.\label{eq:unsafe-big-fail-abort-bound}
    \end{align}
  
  \item[(Event $R_{\errorsc}$)~]
    Let $E$ be the event that the adversary executes
    \Cref{step:rt-run-subadv}. Notice that $R_{\errorsc}$ can only occur when
    $E$ occurs. Note also that
    \[
      E = \neg B_{\emptysc} \land \neg B_{\timeoutsc} \land \neg B_{\abort} \,,
    \]
    so
    \cref{eq:unsafe-big-fail-abort-bound} implies that 
    \begin{align}
      \Pr[E] \ge \frac14 \,.
      \label{eq:e-bound}
    \end{align}

    Suppose that $E$ does occur. We now make two claims: (a)~that $|W| \le w$,
    and (b)~that $\cB \in \Algs(w,t,\frac12,z)$. For the first claim, note
    that by \Cref{step:rt-q0-def}, the set $Q_0$ only contains states $\sigma
    \in \Sigma$ with $|H_\sigma| \le \frac{1}{2} w$. The same bound on
    $|H_\sigma|$ holds for all $\sigma \in Q_{h-1}$ because $Q_{h-1} \subseteq
    Q_0$, thanks to \Cref{step:rt-qh-def}.  By the definition of $W$
    (\Cref{step:rt-wkp1-selection}),
    \begin{align*}
      |W| 
      &= \left|\left\{i \in [n] :\, 
        \frac{|\{\sigma \in Q_{h-1} :\, i \in H_\sigma \}|}{|Q_{h-1}|} 
        \ge \frac{1}{2} \right\}\right| \\
      &\le \sum_{i \in [n]} \frac{2 |\{\sigma \in Q_{h-1} :\, i \in H_\sigma \}|}{|Q_{h-1}|} \\
      &= \frac{2}{|Q_{h-1}|} \sum_{\sigma \in Q_{h-1}} |H_\sigma|
      \le \frac{2}{|Q_{h-1}|} \cdot |Q_{h-1}| \cdot \frac{1}{2} w 
      = w \,,
    \end{align*}
    which proves claim~(a).
  
    For the second claim, consider the sub-algorithm $\cB'$ defined as $\cB$
    just before the modification at \Cref{step:rt-modify-abort}. Then, for a
    particular adversary $\Upsilon$, running $\Upsilon$ against $\cB$ causes
    an abort exactly when running $\Upsilon$ against $\cB'$ causes either an
    abort or an output outside $W'$. In the iteration of the \textbf{for} loop
    that caused $E$, since we reached the \textbf{else} branch, there was no
    deterministic splitting adversary with respect to $\cD$. Thus for
    \emph{any} deterministic adversary $\Upsilon$, by \Cref{def:rt-defs},
    \begin{align*}
      \Pr_{\hat{\sigma} \sim \cD}[\text{$\Outs(\hat{\sigma}, \Upsilon)$ is divisive for $Q_{h-1}$}] 
      \le \frac{1}{2} \,.
    \end{align*}  
    Let $y$ be a realization of the random variable $\Outs(\hat{\sigma},
    \Upsilon)$; note that $y \in [n]^{t} \cup \{\bot\}$. If a given $y$ is not
    divisive, then $y \in [n]^{t}$ and, for each $i \in y$,
    \begin{align*}
      |\{\sigma \in Q_{h-1}:\, i \in H_\sigma\}| 
      \ge |\{\sigma \in Q_{h-1}:\, y \subseteq H_\sigma\}| 
      \ge \frac{1}{2} |Q_{h-1}| \,,
    \end{align*}
    which implies that $i \in W$. Thus in fact $y \subseteq W$. It follows
    that the probability that running $\Upsilon$ against $\cB'$ causes an
    abort or an output outside $W$ (which is a subset of $W'$) is at most
    $\frac12$.
    Viewing $\cB$ as a random-tape \mif algorithm handling input streams
    of length at most $t$, with items from $W'$, in the terminology of
    \Cref{def:mif-rt-err-complexity}, we have $\cB \in \Algs(w,t,\frac12,z)$.
    This proves claim~(b).
  
    Since $\Xi$ is picked to
    maximize the probability of $\cB$ making a mistake, we have
    $\Pr[R_{\errorsc} \mid E] \ge \Mstk(w,t,\frac12,z)$. Using
    \cref{eq:e-bound}, we obtain
    \begin{align*}
      \Pr[R_{\errorsc}] 
      \ge \Pr[E]\, \Mstk\left(w,t,\frac12,z\right)
      \ge \frac{1}{4} \Mstk\left(w,t,\frac12,z\right) \,.
    \end{align*}
  \end{itemize}

  Combining this lower bound with the lower bound for the case where
  $\Pr[B_{\unsafesc}] > \frac{1}{16}$ (considered just before
  \cref{eq:b-unsafe-bound}), we obtain
  \begin{align*}
    \Mstk\left(n,\rl,\frac12,z\right)
    &\ge \min\left(\frac{\rl}{2^7 n},\, 
      \frac{1}{4} \Mstk\left(w,t,\frac12,z\right) \right) \,.\qedhere
  \end{align*}
\end{proof}

\subsection{Calculating the Lower Bound}

\begin{lemma}\label{lem:mif-rt-lb-induct}
  Let $1 \le \rl < n$. For any integer $k \ge 1$, say that $z$ is an integer satisfying $z \le \frac{1}{256} \rl^{1 / k}$. Then:
  \begin{align}
    \Mstk(n,\rl,0,z) > \min\Big(\frac{\rl}{2^{7} n}, \frac{1}{4^{k}} \indic_{z \le L}\Big) \qquad \text{where} \qquad L = \frac{1}{64} \left(\frac{\rl^{k + 1}}{n}\right)^{\frac{2}{k^2 + 3 k - 2}} \,. \label{eq:mif-rt-lb-multistep-lb}
  \end{align}
  Consequently, algorithms for MIF with $\le \min(\frac{\rl}{2^{7} n}, 4^{-k})$ error require $> L$ bits of space.
\end{lemma}

\begin{proof}[Proof of \Cref{lem:mif-rt-lb-induct}]
  Let $n_1 = n$ and $\rl_1 = \rl$, and for $i=2,\ldots,k$, set $n_i = 2\floor{32 \frac{z n_{i-1}}{\rl_{i-1}}}$ and $\rl_i = \floor{\frac{\rl_{i-1}}{64 z}}$. This matches the definitions used in \Cref{lem:mif-rt-lb-step}. As we have been promised that $z \le \frac{1}{256} \rl^{1 / k}$, we have in particular that:
  \begin{align*}
    (256 z)^{k} \le \rl \qquad \text{which implies} \qquad \rl_1 \ge \cdots \ge \rl_k \ge 256 z \,.
  \end{align*}
  By \Cref{lem:rt-triv-logr-lb}, we only need to consider the case $z \ge \log(\rl+1)$, as otherwise $\Mstk(n,\rl,0,z) = 1$. Thus, if $k \ge 2$, we have:
  \begin{align*}
    z \ge \log(\rl+1) \ge k \log(256 z) \ge k \log(256) = 16 \ge 8 \,.
  \end{align*}
  (If $k=1$, then \cref{eq:mif-rt-lb-multistep-lb} follows immediately from \Cref{lem:mif-rt-lb-base-case}.)

  Since $\Mstk(n,\rl,0,z) \ge \Mstk(n,\rl,1/2,z)$, it suffices to lower bound the case where algorithms are permitted up to $1/2$ abort probability. We will lower bound $\Mstk(n,\rl,1/2,z)$ by recursively applying
  \Cref{lem:mif-rt-lb-step} $k-1$ times, and then applying
  \Cref{lem:mif-rt-lb-base-case}. This yields:
  \begin{align*}
    \Mstk(n,\rl,1/2,z) &\ge \min\Big(\frac{\rl_1}{2^7 n_1}, \frac{1}{4} \Mstk(n_1,\rl_1,1/2,z)) \\
      &\ge \min\Big(\frac{\rl_1}{2^7 n_1}, \frac{1}{4} \min\Big(\frac{\rl_2}{2^7 n_2}, \ldots \frac{1}{4} \min\Big(\frac{\rl_{k-1}}{2^7 n_{k-1}},  \frac{1}{4} \indic_{z \le \rl_k^2 / (16 n_k \ln 2)} \Big)\ldots\Big)\Big) \\
      &\ge \min\Big(\frac{\rl_1}{2^7 n_1}, \frac{1}{4} \frac{\rl_2}{2^7 n_2},
      \ldots, \frac{1}{4^{k-1}} \frac{\rl_{k-1}}{2^7 n_{k-1}}, \frac{1}{4^k} \indic_{z \le \rl_k^2 / (16 n_k \ln 2)} \Big) \,.
  \end{align*}
  Only the first and last terms of the minimum are significant, because the terms for $i = 2,\ldots,{k-1}$ are all dominated by the first term:
  \begin{align*}
    \frac{\rl_i}{n_i} = \frac{\floor{\frac{\rl_{i-1}}{64 z} }}{2 \floor{32 \frac{z n_{i-1}}{\rl_{i-1}} } } \ge \frac{1}{2} \frac{\rl_{i-1}}{64 z} \frac{\rl_{i-1}}{64 z n_{i-1}} \ge \frac{\rl_{i-1}}{2 (64 z)^2} \frac{\rl_{i-1}}{n_{i-1}} \ge 4 \frac{\rl_{i-1}}{n_{i-1}} \,,
  \end{align*}
  where in the last step, we used the fact that $\rl_{i-1} \ge \rl_{k-2} \ge (256 z)^2$. Thus:
  \begin{align*}
    \Mstk(n,\rl,0,z) \ge \min\left(\frac{\rl}{2^{7} n k}, \frac{1}{4^k} \indic_{z \le \rl_k^2 / (16 n_k \ln 2)} \right) \,.
  \end{align*}
  
  We have almost proven \cref{eq:mif-rt-lb-multistep-lb}. It remains to lower bound
  $\indic_{z \le \rl_k^2 / (16 n_k \ln 2)}$ by $\indic_{z \le L}$ for some $L$. We do so by proving $z > \rl_k^2 / (16 n_k \ln 2)$ implies $z > L$. As a consequence of the definitions, we have:
  \begin{align*}
    \rl_i 
    =  \floor{\frac{\rl_{i-1}}{64 z} } 
    = \floor{\frac{\floor{\frac{\rl_{i-2}}{64 z}}}{64 z} } 
    = \cdots 
    =\floor{\frac{\rl}{(64 z)^{i-1}} } 
    \qquad \text{and} \qquad 
    n_i \le 64 \frac{z n_{i-1}}{\rl_{i-1}} \,,
  \end{align*}
  and thus:
  \begin{align*}
     z &> \frac{\rl_{k}^2}{16 \ln 2} \frac{1}{n_{k}} \ge \frac{\rl_{k}^2}{16} \frac{\prod_{k=1}^{k - 1} \rl_{k} }{(64 z)^{k - 1} n} \\
      & \ge \frac{\rl^{k + 1} }{16 \cdot (64 z)^{2 (k -1) + (k - 2) + (k - 3) + \cdots + 1 + 0}  (64 z)^{k - 1} n} \\
      & = \frac{\rl^{k + 1} }{16 \cdot (64 z)^{(k^2 + 3 k - 4) / 2} n} \,.
  \end{align*}
  Rearranging to put all the $z$ terms on the left gives:
  \begin{align*}
    64 z \cdot (64 z)^{\frac{k^2 + 3k - 4}{2}} \ge 64 \frac{\rl^{k+1}}{16 n} \,,
  \end{align*}
  which implies
  \begin{align*}
    z &> \frac{1}{64} \left(\frac{64 \rl^{k + 1}}{16 n}\right)^{\frac{2}{k^2 + 3 k - 2}} \ge \frac{1}{64} \left(\frac{\rl^{k + 1}}{n}\right)^{\frac{2}{k^2 + 3 k - 2}} \,. \qedhere
  \end{align*}
\end{proof}

Now, say a random tape algorithm has error probability $\le \frac{\rl}{2^{7} n}$ against any adaptive adversary, and uses $z$ bits of space. If $\rl \ge n^{2/3}$, it is easy to show that the optimal value of $k$ with which to apply \Cref{lem:mif-rt-lb-induct} is $1$. Otherwise, for all $k$, we have two cases: if $z \le \frac{1}{256} \rl^{1/k}$, then we can apply \Cref{lem:mif-rt-lb-induct}. By \Cref{lem:rt-triv-logr-lb}, $z \log(\rl + 1) \ge 1$, which implies $(256)^k \le \rl$ and thus $k \le \floor{\frac{1}{8} \log \rl}$. As the algorithm has
\begin{align*}
  \text{error} \le \frac{\rl}{2^{7} n} \le \min(\frac{\rl}{2^{7} n}, \frac{1}{2^{7} \rl^{1/2}} ) \le \min(\frac{\rl}{2^{7} n}, \frac{1}{4^{k}}) \,,
\end{align*}
by \Cref{lem:mif-rt-lb-induct} the algorithm's space usage must satisfy:
\begin{align*}
  z \ge \frac{1}{64} \left(\frac{\rl^{k + 1}}{n}\right)^{\frac{2}{k^2 + 3 k - 2}} \,.
\end{align*}
Since we either have this lower bound on $z$, or $z > \frac{1}{256} \rl^{1/k}$, it follows that for any integer $k \ge 1$:
\begin{align*}
  z \ge \frac{1}{256} \min\left(\rl^{1/k}, \left(\frac{\rl^{k + 1}}{n}\right)^{\frac{2}{k^2 + 3 k - 2}}\right)  \,.
\end{align*}
We can take the maximum over all values of $k$ to obtain:
\begin{align*}
  z \ge \frac{1}{256} \max_{k \in \NN} \min\Big(\rl^{1/k}, \left(\frac{\rl^{k + 1}}{n}\right)^{\frac{2}{k^2 + 3 k - 2}}\Big)  \,.
\end{align*}
The right hand side can be simplified with the following lemma, whose proof is mostly calculation and is deferred to \Cref{subsec:rt-lb-mech}:

\begin{lemma}\label{lem:rt-lb-calc-multiple}
  \begin{align}
    \max_{k \in \NN} \min\left(\rl^{1/k},\left(\frac{\rl^{k + 1}}{n}\right)^{\frac{2}{k^2 + 3 k - 2}}  \right) \ge  \max_{k \in \NN} \left(\frac{\rl^{k + 1}}{n}\right)^{\frac{2}{k^2 + 3 k - 2}}  \ge \rl^{\frac{15 \log \rl}{32 \log n}} \label{eq:rt-lb-max-min-lb}
  \end{align}
\end{lemma}

Summarizing, we obtain the following theorem.

\begin{theorem} 
\label{thm:rt-lb-intro}
  Random tape $\delta$-error adversarially robust algorithms for $\mif(n,\rl)$ require
  \begin{align*}
    \Omega\left(\max_{k \in \NN} \left(\frac{\rl^{k + 1}}{n}\right)^{\frac{2}{k^2 + 3 k - 2}}\right) 
    = \Omega\left(\rl^{\frac{15}{32} \log_{n} \rl}\right)
  \end{align*}
  bits of space, for $\delta \le \rl/(2^{7} n)$. \qed
\end{theorem}

\subsection{Remarks on the Lower Bound}

To prove the above lower bound, we required $\delta \le \frac{\rl}{2^{7} n}$.
For larger values of $\delta$, random tape algorithms can be much more
efficient. For example, there is a $(\log t)$-space algorithm for
$\mif(n,\rl)$ with $O(\rl^2/t)$ error probability, which on every input
randomly picks a new state (and output value) from $[t]$.

That being said, if a random tape algorithm $\cA$ with error $\le \delta$
provides the additional guarantee  that it never makes a mistake (i.e, either
produces a correct output or aborts),\footnote{This guarantee essentially
rules out the possibility of algorithms that randomly and blindly guess
outputs. Most of the algorithms for $\mif$ in this paper and in
\cite{Stoeckl23} provide this ``zero-mistake'' guarantee.} one can construct a
new algorithm $\cB$ with error $O(\frac{1}{n})$ by running $\Theta(\frac{\log
n}{\log 1/\delta})$ parallel copies of $\cA$ and reporting outputs from any
copy that has not yet aborted. Proving a space lower bound for $\cB$ then
implies a slightly weaker one for $\cA$.

The lower bound of \Cref{thm:rt-lb-intro} is not particularly tight, and we
suspect it can be improved to match the upper bound within $\polylog(\rl, n)$
factors. There are at least two scenarios that we suspect algorithms must
behave similarly to, in which we could do better than the current adversary's
reduction to $\mif(w, \Theta(\rl/z))$

\begin{itemize}

  \item If a constant fraction of the next $\rl/2$ outputs are contained in
  $W$, we could (essentially) reduce to $\mif(w, \rl/2)$.

  \item If on each search step of length $\Theta(\rl/z)$, the outputs of the
  algorithm are concentrated in a new set of size $\Theta(w/z)$, then we could
  (essentially) reduce to $\mif(\Theta(w/z), \Theta(\rl/z))$.

\end{itemize}

The adversary of \Cref{alg:rt-adversary-step} runs in doubly exponential time,
and requires knowledge of the algorithm. The former condition cannot be
improved by too much: if one-way functions exist, one could implement the
random oracle algorithm for $\mif(n,\rl)$ from \cite{Stoeckl23} using a
pseudo-random generator that fools all polynomial-time adversaries.  One can
also prove by minimax theorem that universal adversaries for (random tape or
otherwise) $\mif(n,\rl)$ algorithms can not be used to prove any stronger
lower bounds than the one for random oracle algorithms.

\section{The Random Tape Upper Bound}\label{sec:rt-ub}

In this section, we describe an adversarially robust random tape algorithm for $\mif(n,\rl)$ which obtains error $\le \delta$. See \Cref{subsec:rt-ub-overview} for a high level overview. The algorithm, shown in \Cref{alg:rt}, can be implemented for almost all pairs $\rl < n$, requiring only $\rl \le n / 64$ and $\rl \ge 4$ for its parameters to be meaningful. It can be seen as a multi-level generalization of the random oracle algorithm from \cite{Stoeckl23}.

\begin{figure}[ht]
  \centering
  \includegraphics[width=15cm]{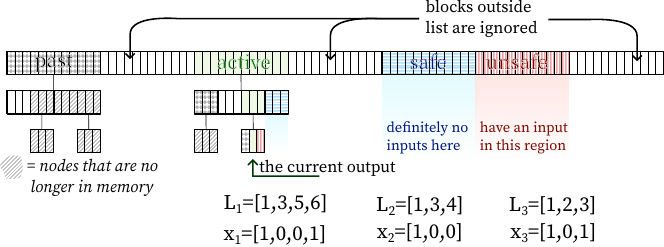}
  \caption{Diagram showing the state of the algorithm in \Cref{alg:rt} and how it relates to the parts of the implicit random tree that the algorithm traverses. Positions on the horizontal axis correspond to different integers in $[n]$. To keep the example legible, we set parameters $d = 3$,  $w_1 = 7, w_2 = 4, w_3 = 3$, and $b_1 = 4, b_2 = 3, b_3 = 3$. \label{fig:rt-ub-intuition}}
\end{figure}

In order to prove that the algorithm in \Cref{alg:rt} is correct, we will need some additional notation. Let $d,\alpha,b_1,\ldots,b_d,w_1,\ldots,w_d$ be
as defined in \Cref{alg:rt}. It is helpful to view this algorithm as traversing over the leaves of a random tree of height $d$, in which:
\begin{itemize}
  \item Every node $v$ in the tree is associated with a subset $S_v$ of $[n]$. We say a node is at level $i$ if it is at depth $i - 1$. All nodes at a given level have disjoint associated subsets.
  \item The ``root'' $\rho$ of the tree has $S_\rho$ of size $\prod_{i=1}^{d} w_i$, and is at level $1$
  \item Each node $v$ at level $i$ (depth $i-1$) has $b_{i+1}$ children; the set $S_v$ is partitioned into $w_{i+1}$ parts of equal size, and each child of $v$ is associated with a random and unique one of these parts
  \item Each leaf node $u$, at depth $d$, is associated with a set of size 1, i.e., a single and unique integer in $[n]$. There are $\prod_{i=1}^{d} b_i$ leaf nodes in total.
\end{itemize}
See for example \Cref{fig:rt-ub-intuition}. The algorithm maintains a view of just the branch of the tree from the root to the current leaf node. Its output will be the number associated to this leaf. For each node $v$ on this branch, at level $i \in [d]$ (depth $i-1$), it keeps a record of the positions $L_i \in [w_i]^{b_i}$ of its children, and a record $x_i \in \{0,1\}^{b_i}$ indicating their status. There are four categories for child nodes:
\begin{itemize}
\item A node is \textsc{past} if the traversal over the tree passed through and leaf the node; past nodes are marked with a $1$ in $x_i$.
\item A node is \textsc{active} if it is on the branch to the current leaf node; this is the node with the lowest index which is marked with a $0$ in $x_i$.
\item A node $v$ is \textsc{safe} if it comes after the active node, and the adversary has never sent an input in $S_v$; safe nodes are marked with a $0$ in $x_i$.
\item A node $v$ is \textsc{unsafe} if it comes after the active node, and the adversary did send an input in $S_v$; unsafe nodes are marked with a $1$ in $x_i$.
\end{itemize}
The algorithm maintains these records as the adversary sends new inputs,
marking safe child nodes $v$ as unsafe if an element in $S_v$ is received. The
current leaf node is found by, from the root, following the chain of active
nodes. If the adversary sends the value of the current leaf node, the
algorithm will mark it by setting the corresponding entry in $x_d$ to $1$,
thereby changing the value of the current active node. If every element of
$x_d$ is a now $1$, this means that the adversary has sent an input for every
child of the level $d$ node on the current branch, so the algorithm marks the
current active child of the level $d - 1$ node with a $1$, thereby moving the
current branch to use a new level $d$ node, $u$, which is
\emph{safe}.\footnote{If the level $d - 1$ node has no children marked with a
$0$ after this, we repeat the process at level $d -2$, and so on.} It the
``loads the positions of the children of $u$''---the tree being randomly
generated, this is implemented by $L_d$ being randomly sampled and $x_d$ being
reset to be all zeros---and proceeds.

While there are $\prod_{i=1}^{d} b_i$ leaf nodes in the ideal random tree, the algorithm's traversal of them may skip a fraction, because they (or one of their ancestors) were marked as unsafe. We say that such leaf nodes, along with those which were once active, have been \textsc{killed}.

\begin{algorithm}
  \caption{Adversarially robust random tape algorithm for $\mif(n,\rl)$ with error $\le \delta$}
  \label[listing]{alg:rt}

  \begin{algorithmic}[1]
  \Statex Requirements: $\rl \le n / 64$ and $\rl \ge 4$.
  \Statex Parameters: $d = \min(\ceil{\log \rl}, \floor{2 \frac{\log(n/4)}{\log 16 \rl}})$. 
  \Statex $\alpha = \begin{cases}2 & \text{if $\ceil{\log \rl} <  \floor{2 \frac{\log(n/4)}{\log 16 \rl}}$} \\ \frac{(4 \rl)^{2/(d-1)}}{(n/4)^{2/(d(d-1))}} & \text{otherwise}\end{cases}$
    \Statex Let $u$ be chosen via \Cref{lem:fract-power-round}, so that $\alpha^{d-2} \le \prod_{i=2}^{d-1} b_i \le 2 \alpha^{d-2}$
    \Statex $b_2 = \ldots = b_u = \ceil{\alpha}$, and $b_{u+1} = \ldots = b_{d-1} = \floor{\alpha}$
    \Statex $b_1 = \min(\rl + 1, \ceil{8 \alpha} + \ceil{3 \log 1/\delta})$; and $b_d = \ceil{ \frac{\rl}{\alpha^{d-1}} }$
    \Statex $w_1 = 16 \rl$; and for each $i \in \{2, \ldots, d\}$, $w_i = \prod_{j=i}^{d} b_j$.
  
  \Statex Let $\iota : [w_1]\times[w_2]\times \cdots \times [w_d] \rightarrow [n]$ be an arbitrary injective function
  
  \Statex
  \Statex \ul{\textbf{Initialization}}:
    \For{$i \in [d]$}
      \State $L_i \gets $ random sequence without repetition in $[w_i]^{b_i}$
      \State $x_i \gets (0,\ldots,0) \in \{0,1\}^{b_i}$
    \EndFor
    
  \Statex
  \Statex \ul{\textbf{Update}($a \in [n]$)}: 
    \If{ $a \notin \iota^{-1}[n]$ } \textbf{return} \Comment{Any integer not in $\iota^{-1}[n]$ can never be an output}\EndIf
    \State $v_1,\ldots,v_d = \iota^{-1}(a)$ \Comment{Map input into $[w_1]\times\ldots\times[w_d]$}
    \State For $i \in [d]$, define $c_i = \min |\{j : x_i[j] = 0\}|$
    
    \If{for all $i\in[d]$, $v_i = L_i[c_i]$}
        \LineComment{Move to the next leaf node, sampling new child node positions as necessary}
        \For{$i = d,\ldots,1$}
          \State $x_i[c_i] \gets 1$
          \If{$x_i$ is the all-1s vector}
            \If{$i = 1$} \textbf{abort}  \Comment{If we reach $i=1$, then even the root node is full}\EndIf
            \State $L_i \gets $ random sequence without repetition in $[w_i]^{b_i}$
            \State $x_i \gets (0,\ldots,0)$
          \Else \ \textbf{break}
          \EndIf
        \EndFor
    \Else
        \LineComment{Mark a branch as unsafe, if there was a hit}
        \State Let $j$ be the smallest integer in $[d]$ for which $v_j \ne L_{j}[c_j]$. 
        \If{$\exists y \in [b_j]$ for which $L_j[y] = v_j$}
            \State $x_j[y] \gets 1$
        \EndIf
    \EndIf
    
  \Statex
  \Statex \ul{\textbf{Output} $\rightarrow [n]$}:
    \State For $i \in [d]$, define $c_i = \min |\{j : x_i[j] = 0\}|$
    \State \textbf{return} $\iota[ (L_1[c_1],L_2[c_2],\ldots,L_d[c_d]) ]$
  \end{algorithmic}
  
\end{algorithm}

\Cref{alg:rt} uses the following lemma to set some of its parameters; the specific rounding scheme for the values $b_2,\ldots,b_{d-1}$ ensures that $b_1$ can
decrease relatively smoothly as $\rl$ decreases. (Setting all $b_2=\ldots=b_{d -1}$ to $\floor{\alpha}$ can lead to having $b_1$ be significantly larger than
necessary (by up to a factor $(3/2)^d = \rl^{O(1)}$); setting all $b_2=\ldots=b_{d -1}$ to $\ceil{\alpha}$ would violate the $\prod_{i=1}^{d} w_i \le n$ constraint.)

\begin{lemma}\label{lem:fract-power-round}
  Let $\alpha \ge 1$. Then for all $k \ge 0$, there exists an integer $u$ depending
  on $\alpha$ and $k$ so that 
  \begin{align*}
    \alpha^k \le \ceil{\alpha}^u \floor{\alpha}^{k - u} \le 2 \alpha^k \,.
  \end{align*}
\end{lemma}

\begin{proof}[Proof of \Cref{lem:fract-power-round}]
  If $\alpha$ is an integer, we are done. Otherwise, with
  \begin{align*}
    u = \ceil{\frac{k \log(\alpha/\floor{\alpha})}{\log(\ceil{\alpha} /\floor{\alpha})}} \,, \quad \text{we have} \quad \ceil{\alpha}^u \floor{\alpha}^{k - u} = \floor{\alpha}^{k} (\ceil{\alpha}/\floor{\alpha})^u \ge \floor{\alpha}^{k} (\ceil{\alpha/\floor{\alpha}})^k = \alpha^k \,.
  \end{align*}
  Similarly, $\ceil{\alpha}^u \floor{\alpha}^{k - u} \le \alpha^k \ceil{\alpha}/\floor{\alpha}$, which is $\le 2\alpha^k$ because $\alpha \ge 1$ implies $\ceil{\alpha}/\floor{\alpha} \le 2$.
\end{proof}

The following lemma is straightforward but tedious, and we defer its proof to  \Cref{subsec:rt-ub-mech}.

\begin{lemma}\label{lem:rt-alg-params}
  The parameters of \Cref{alg:rt} satisfy the following conditions:
  \begin{align}
      \prod_{i=2}^{d} b_i &\ge \frac{\rl}{\alpha} \,, \label{eq:rt-ub-param-ge-l} \\
      \prod_{i=2}^{d} b_i &\le \frac{4 \rl}{\alpha} \,, \label{eq:rt-ub-param-le-2l}\\
      \prod_{i \in [d]} w_i &\le n \,. \label{eq:rt-ub-param-le-n}
  \end{align}
\end{lemma}

We now prove the main lemma:

\begin{lemma}\label{lem:rt-ub-error-limit}
  \Cref{alg:rt} has error $\le \delta$ in the adversarial setting.
\end{lemma}

\begin{proof}[Proof of \Cref{lem:rt-ub-error-limit}]
  We prove, using a charging scheme, that the probability of all leaf nodes
  in the random tree traversed by the algorithm being killed is $\le \delta$.

  The input of the adversary at any step falls into one of $d + 2$ categories.
  For each $i \in [d]$, it could add an input which intersects the list
  of unrevealed child positions of the level $i$ node, possibly killing
  $\prod_{j=i + 1}^{d} b_i$ leaf nodes if it guesses correctly. It could
  also send the value of the current leaf 
  node, thereby killing it (and only it).
  Finally, the adversary's input could be entirely wasted (outside $\iota([w_1]\times\ldots\times[w_d]$, repeating an input it made before, or in
  the region corresponding to one of the past nodes in the random tree); then
  no leaf nodes would be killed.

  As the algorithm proceeds, for each node in the random tree (other than
  the root), we accumulate charge. When the algorithm's current branch changes
  to use new nodes, the charge on the old nodes is kept, and the new nodes
  start at charge $0$.
  
  When the adversary makes an input that is handled by level $i$ ("query" at level $i$),
  for $i \in \{2,\ldots,d\}$, it
  \emph{first} deposits one unit of charge at the active level $i$ node.
  Then, if the query was a hit (i.e, ruled out some future subtree and made
  a child of a node in the current branch change from ``safe'' to ``unsafe''), increase
  the number of killed nodes by the number of leaves for the subtree
  (namely, $\prod_{j={i+1}}^{d} b_j$), and remove up to that amount 
  of charge from the node. The definitions of $(w_j)_{j=2,\ldots,d}$ ensure
  that $\prod_{j={i+1}}^{d} b_j = w_j / b_j$.
  
  For the $t$th query, let $K_t$ be the number of killed leaf nodes on the query,
  minus any accumulated charge on the node. Say the adversary picks a node at level $i$ for $i \in \{2,\ldots,d-1\}$, and that
  node has $\hw$ unexplored subtree regions (i.e, neither revealed because
  the algorithm produced outputs in them, nor because there was there a query at that subtree region
  in the past), and $\hb$ gives the number of subtrees within this unexplored
  region. If $\hb = 0$, $\EE[K_t | K_1,\ldots,K_{t-1}] = 0$. Otherwise, let $u$ be the number of subtrees which were revealed by the algorithm so far; we have $\hw \le w_j - u$ and $\hb \le b_j - u$. Then when we condition on the past increases in charge, the subtree regions within the unexplored region are still uniformly random; hence
  the probability of hitting a subtree is $\hb / \hw$. The number of leaf nodes
  killed by a hit is $w_j / b_j$.  The total charge currently at the node must be
  $\ge (w_j - u - \hw) - (\frac{w_j}{b_j}) (b_j - u - \hb) = \hb \frac{w_j}{b_j} - \hw + (\frac{w_j}{b_j} - 1) u \ge \hb \frac{w_j}{b_j} - \hw$, since each removed node consumes at most $\frac{w_j}{b_j}$ of the existing charge.
  Consequently, the increase in killed leaf nodes if we hit is $\max(0, \frac{w_j}{b_j} - \max(0, \hb \frac{w_j}{b_j} - \hw))$, so the expected payoff is\footnote{When the level is $d$ and $b_j=w_j$, we in fact we have $K_t = 1$ always; but we do not need this stronger fact.}:
  \begin{align*}
    \EE[K_t \mid K_1,\ldots,K_{t-1}] = \frac{\hb}{\hw} \max\left(0,\, \frac{w_j}{b_j} - \max\left(0,\, \hb \frac{w_j}{b_j} - \hw\right)\right) \underset{\text{by \Cref{lem:charge-adjusted-gain}}}{\le} 1 \,.
  \end{align*}
  
  If the level is $1$, then let $J \subseteq [16 \rl]$ give the set of probed subtree positions,
  and $H$ give the set of revealed subtree positions; since there are $\le \rl$ queries,
  $|J|,|H|$ are both $\le \rl$, and the probability of a query in an unexplored
  region to hit is $\le \frac{b_1}{16 \rl - |J \cup H|} \le \frac{b_1}{14 \rl}$. The charging scheme does not apply, so
  \begin{align*}
    \EE[K_t \mid K_1,\ldots,K_{t-1}] \le \frac{b_1}{14 \rl} \cdot \prod_{j=2}^{d} b_j \,.
  \end{align*}
  
  Thus in all cases, $\EE[K_t \mid K_1,\ldots,K_{t-1}] \le \max(1,\, \prod_{j=1}^{d} b_j / 14 \rl)$. 
  
  The total charge deposited on mid-level nodes is $\le \rl$. The algorithm is guaranteed to succeed if the total number of leaves killed is less than the total number of leaves; i.e, if $\sum_{i \in [\rl]} K_t + \rl \le \prod_{j=1}^{d} b_j$. Note that by \Cref{lem:rt-alg-params} and the definition of $b_1$, $\prod_{j=1}^{d} b_j \ge b_1 \rl / \alpha \ge 8 \rl$. Consequently,
  \begin{align*}
    \rl + 7 \EE\left[\sum_{t = 1}^{\rl} K_t\right] 
    \le 8 \max\left(\rl,\, \frac{1}{14} \prod_{j=1}^{d} b_j\right)
    \le 8 \max\left(\frac{1}{8},\, \frac{1}{14}\right) \prod_{j=1}^{d} b_j \le  \prod_{j=1}^{d} b_j \,.
  \end{align*}
  Now let $D_t = K_t / \prod_{j=2}^{d} b_j$, so that each $D_t \in [0,1]$. Writing events in terms of $D_t$ lets us use \Cref{lem:azumanoff} to bound the probability that too many leaves are killed:
  \begin{align*}
  \Pr\left[\rl + \sum_{t \in [\rl]} K_t \ge \prod_{i=1}^{d} b_i\right] &\le 
      \Pr\left[\sum_{t \in [\rl]} K_t \ge 7 \max\Big(\rl, \frac{1}{14} \prod_{j=1}^{d} b_j\Big) \right] \\
      &\le \Pr\left[\sum_{t \in [\rl]} D_t \ge 7 \max\Big(\rl / \prod_{j=2}^{d} b_j, \frac{b_1}{14} \Big) \right] \\
      &\le \exp\left( - \frac{6^2}{2+6} \max\Big(\rl / \prod_{j=2}^{d} b_j, \frac{b_1}{14} \Big) \right) \\
      &\le \exp\left( - \frac{9 b_1}{28} \right) \le 2^{b_1 \cdot \frac{9}{28 \ln 2}} \le 2^{\ceil{3 \log 1/\delta} \frac{9}{28 \ln 2}} \le \delta \,. \qedhere
  \end{align*}  
\end{proof}

In the preceding proof, we used the following:

\begin{lemma}\label{lem:charge-adjusted-gain}
  Let $\hb,\hw,b,w$ be positive, and $\hb \ge 1$. Then:
  \begin{align*}
    \frac{\hb}{\hw} \left( \max\left(0,\, \frac{w}{b} - \max\left(\hb \frac{w}{b} - \hw,\, 0\right) \right)\right) \le 1 \,.
  \end{align*}
\end{lemma}

\begin{proof}[Proof of \Cref{lem:charge-adjusted-gain}]
  If $\frac{\hb}{\hw} \le \frac{b}{w}$, then:
  \begin{align*}
    \frac{\hb}{\hw} \left( \max(0, \frac{w}{b} - \max(\hb \frac{w}{b} - \hw, 0) )\right) \le \frac{\hb}{\hw} \frac{w}{b} \le 1 \,.
  \end{align*}
  Otherwise, $\frac{\hb}{\hw} \ge \frac{b}{w}$, and:
  \begin{align*}
    \frac{\hb}{\hw} &\left( \max\left(0,\, \frac{w}{b} - \max\left(\hb \frac{w}{b} - \hw,\, 0\right) \right)\right) \\
      &\le \frac{\hb}{\hw} \left(\frac{w}{b} - \left(\hb \frac{w}{b} - \hw\right)\right) 
      = \frac{\hb}{\hw} \left(\frac{w}{b} - \hb \left(\frac{w}{b} - \frac{\hw}{\hb}\right)\right)\\
      &\le \frac{\hb}{\hw} \left(\frac{w}{b} - 1 \left(\frac{w}{b} - \frac{\hw}{\hb}\right)\right) = \frac{\hb}{\hw} \frac{\hw}{\hb} = 1 \,, \qquad\quad \text{since $\hb \ge 1$ and $\frac{w}{b} - \frac{\hw}{\hb} \ge 0$ \,.} \qedhere
  \end{align*}
\end{proof}

The proof of the following lemma is a mostly straightforward calculation, which we defer to \Cref{subsec:rt-ub-mech}.

\begin{lemma}\label{lem:rt-ub-space-usage}
  \Cref{alg:rt} uses 
  \begin{align}
    O\left(\ceil{\frac{(4 \rl)^{2/(d-1)}}{(n/4)^{2/(d(d-1))}}} (\log \rl)^2 + \min(\rl, \log 1/\delta) \log \rl \right) \label{eq:rt-ub-semiprec-space}
  \end{align} bits of space, where $d = \min\left(\ceil{\log \rl}, \floor{2 \frac{\log(n/4)}{\log(16\rl)}}\right)$. A weaker upper bound on this is:
  \begin{align*}
    O\left(\rl^{\frac{\log \rl}{\log n}} (\log \rl)^2 + \min(\rl, \log 1/\delta) \log \rl\right) \,.
  \end{align*}
\end{lemma}

If $\rl \ge 4$ and $\rl \le n/64$, then \Cref{lem:rt-ub-space-usage} and \Cref{lem:rt-ub-error-limit} together show that \Cref{alg:rt} has error $\le \delta$ and space
usage as bounded by \cref{eq:rt-ub-semiprec-space}. To handle the cases where
$\rl < 4$ and $\rl > n/64$, one can instead use the simple deterministic algorithm
for $\mif(n,\rl)$ from \cite{Stoeckl23}, using only $\rl$ bits of space. As this is in
fact less than the space upper bound from \cref{eq:rt-ub-semiprec-space}, it follows
that \cref{eq:rt-ub-semiprec-space} gives an upper bound on the space
needed for a random tape, adversarially robust $\mif(n,\rl)$ algorithm for any
setting of parameters. Formally:

\begin{theorem} 
\label{thm:rt-ub-intro}
  There is a family of adversarially robust random tape algorithms, where
  for $\mif(n,\rl)$ the corresponding algorithm has $\le \delta$ error and uses
  \begin{align*}
    O\left(\ceil{ 
    \frac{(4 \rl)^{\frac{2}{d - 1}}}{(n/4)^{\frac{2}{d(d-1)}}}
    } (\log \rl)^2 + \min(\rl,\, \log \tfrac{1}{\delta}) \log \rl \right)
  \end{align*}
  bits of space, where $d = \max\left(2, \min\left(\ceil{\log \rl}, \floor{2 \frac{\log(n/4)}{\log(16\rl)}}\right)\right)$. When $\delta = 1/\poly(n)$ a (weakened) space bound is $O\left(\rl^{\log_n \rl} (\log \rl)^2 + \log \rl\log n\right)$.
\end{theorem}

\begin{remark}
  \Cref{alg:rt} does not use the most optimal assignment of the parameters $b_d,\ldots,b_2$; constant-factor improvements in space usage are possible if one sets $b_d,\ldots,b_2$ to be roughly in an increasing arithmetic sequence, but this would make the analysis more painful.
\end{remark}

\begin{remark}
  In exchange for a constant factor space increase, one can adapt \Cref{alg:rt} to produce an increasing sequence of output values. Similar adjustments can be performed for other $\mif$ algorithms.
\end{remark}

\section{The Pseudo-Deterministic and Random Seed Lower Bounds}\label{sec:pd-lb}

In this section, we prove a space lower bound for pseudo-deterministic
streaming algorithms; in particular, for the most general (random oracle) type
of them. See \Cref{subsec:pd-lb-overview} for a high-level plan of the proof.

Let $\cA$ be a random-oracle pseudo-deterministic algorithm for $\mif(n, \rl)$
using $z$ bits of state, which has worst case failure probability $\delta \le
\frac{1}{3}$. Let $\Pi \colon [n]^\star \to [n]$ be the function giving the
{\em canonical output} of $\cA$ after processing a stream (as was defined in
\Cref{subsec:models}), and let $S = \Pi([n]^\rl)$ be the set of all canonical
outputs at time $\rl$ (for this proof we can ignore outputs at times $<\rl$).
Clearly $|S| \le n$, and we will also prove (see \Cref{lem:pd-s-le-2zp1}) that
$|S| \le 2^{z+1}$. It is easy to see that $|S| \ge \rl + 1$ and that, since
$\mif(n,\rl)$ is nontrivial, correct algorithms must have $z \ge 1$.

The main proof in this section only applies to algorithms with very low error
(potentially as small as $1/n^{\Omega(\log n)}$). To ensure that we are
working with an algorithm with error this small, we will first apply
\Cref{lem:error-reduction-by-vote}, using $p$ independent instances of $\cA$,
where $p \ge 1$ is an integer chosen later. This will produce an $\epsilon$-error
algorithm $\cB$ that uses $z p$ bits of space, where $\epsilon \le
(2\delta)^{p/30}$. Moreover, the canonical outputs of $\cB$ will still be
given by $\Pi$.

We may view $\cB$ as a distribution over deterministic streaming algorithms;
each such algorithm, obtained by fixing $\cB$'s oracle random string, produces
outputs according to a function of the form $A \colon [n]^\star \to [n]$.
Since $\cB$ is pseudo-deterministic, we have the following property:
\begin{align}
  \forall x \in [n]^{\le \rl}: \Pr_{A \sim \cB} [A(x) \ne \Pi(x)] \le \epsilon \,. 
  \label{eq:pd-fn-prop}
\end{align}

It is time to give the details of the key procedure \textsc{FindCommonOutputs}
(abbreviated as $\fco$), outlined earlier in \Cref{subsec:pd-lb-overview},
that underpins our proof.  The input stream positions, from $1$ to $\rl$, are
split into $d = \Theta(\rl / (zp))$ consecutive intervals, of lengths $t_d,
t_{d-1}, \ldots, t_1$, in order, where $t_d = t_{d-1} = \cdots = t_2 =
\Theta(zp)$ and $t_1 = \rl - \sum_{k=2}^{d} t_k$; the constants are chosen
such that $t_1 \ge \rl/2$. A call to \fco takes the form $\fco(B, C, x,
k)$, where $B \colon [n]^\rl \to [n]$ is an output function, $C$ is a random
collection (matrix) of thresholds, $k \in [d]$, and $x \in [n]^{t_d + \cdots +
t_{k+1}}$ is a stream prefix. (Notice that when $k = d$, we must have $x =
\emptystream$.) By design \fco is recursive, bottoming out at $k = 1$, and
always returns a subset of $S$ of cardinality $w_k$, where
\begin{align}
  w_k := 2^{k-1} (t_1 + 1) \,. \label{eq:wk-def}
\end{align}
The precise logic of \fco is given in \Cref{alg:pd-set-finding}. Throughout
this section, we use $\SSD{Z}{t}$ to denote the set of all length-$t$ {\em
sorted sequences} of {\em distinct} elements of $Z$. Note that
\[
  |\SSD{Z}{t}| = \binom{|Z|}{t} \,.
\]

\begin{algorithm}[htb]
  \caption{The procedure to compute a set for \Cref{lem:pd-main-recursion}}\label[listing]{alg:pd-set-finding}

  \begin{algorithmic}[1]
    \Statex Let $t_1,\ldots,t_d$, $w_1,\ldots,w_d$ be integer parameters, and $S$ the set of valid outputs
    \Statex
    \Statex \underline{\textsc{FindCommonOutputs($B$, $C$, $x$, $k$)}} \Comment{abbreviated as $\fco(B,C,x,k)$}
    \LineComment{Inputs: function $B \colon [n]^\rl \to [n]$, 
      matrix $C \in [1,2)^{d \times \NN}$, stream prefix $x \in [n]^{t_k + \cdots + t_1}$}
    \LineComment{Output: a subset of $S$ of size $w_k$}
    
    \If{$k = 1$} 
        \State $e_0 \gets B(x \concat \langle 1,1,\ldots,1\rangle)$\label{step:pd-base-case-start}
        \For{$i$ in $1,\ldots,t_1$} 
          \State $e_i \gets B(x \concat \langle e_0,\ldots,e_{i-1},1,\ldots,1\rangle)$ \label{step:pd-base-case-mid}
        \EndFor
        \If{$e_0, \ldots,e_{t_1}$ are all distinct}
          \State \textbf{return} $\{e_0,e_1,\ldots,e_{t_1}\}$
          \Comment{identify $w_1$ distinct possible outputs}
        \EndIf
        \State \textbf{return} arbitrary subset of $S$ of size $w_1$ (failure)
    \Else
        \For{each $y \in \SSD{[n]}{t_k}$}
            \State $T_{y} \gets \textsc{FindCommonOutputs}(B, C, x \concat y, k-1)$ \label{step:pd-compute-t}
            \Comment{note $|T_{y}| = w_{k-1}$}
        \EndFor
    
        \State $Q_0 \gets T_{\langle 1,2,\ldots,t_k \rangle}$
        \For{$h$ in $1,2,3,4$}\label{step:pd-for-loop-start}
            \LineComment{gather statistics and find common elements among the sets $T_{y}$}
            \For{each $j\in S$}
                \State $f_{j}^{(h)} \gets \left| \left\{ y \in \SSD{Q_{h-1}}{t_k} :
                  j \in T_{y} \right\} \right|$ \label{step:pd-def-f}
                  \Comment{count frequencies}
            \EndFor
            \State $\theta \gets C_{k,h} w_{k-1} / (16 |S|)$ \Comment{set random threshold}
            \State $P_h \gets \left\{j \in S : f_{j}^{(h)} \ge 
              \theta \binom{|Q_{h-1}|}{t_k} \right\}$ \label{step:pd-comp-ph}
              \Comment{identify ``sufficiently common'' elements}
            \State $Q_h \gets Q_{h-1} \cup P_h$
            \If{$|Q_h| \ge w_k$}
              \State \textbf{return} the $w_k$ smallest elements in $Q_h$ \label{step:pd-good-exit}
            \EndIf
        \EndFor \label{step:pd-for-loop-end}
        \State \textbf{return} arbitrary subset of $S$ of size $w_k$ (failure)\label{step:pd-bad-exit}
    \EndIf
  \end{algorithmic}
  
\end{algorithm}

A key property of \fco is that for all $x \in [n]^{t_d + \cdots + t_{k+1}}$,
if $C$ is chosen uniformly at random from $[1,2)^{d \times \NN}$, and $A \sim
\cB$, then w.h.p. $\fco(A, C, x, k)$ produces the same set as $\fco(\Pi, C, x,
k)$. Another important property is that $\fco(\Pi, C, x, k)$ always produces
$w_k$ distinct elements in $S$ without a failure (i.e., returns using
\Cref{step:pd-good-exit}, not \Cref{step:pd-bad-exit}). These properties
are formally proved in \Cref{lem:pd-main-recursion}, to follow.  In
particular, $\fco(\Pi, C, \emptystream, d)$ produces $w_d$ distinct elements
in $S$, showing that $|S| \ge w_d$. Since $w_d$ is a function of $z$,
combining this with upper bounds on $|S|$ lets us solve for a lower bound on
$z$.

We elaborate on how we achieve the two key properties of $\fco$. First, to
ensure $\fco(\Pi,C,x,k)$ equals $\fco(A,C,x,k)$ for random $A \sim \cB$ and $C
\in_R [1,2)^{d \times \NN}$, we use a standard random threshold trick when
computing the set $P_h$ on \Cref{step:pd-comp-ph}. The recursive calls to
$\fco(A, C, x \concat y, k-1)$ on \Cref{step:pd-comp-ph} do not always
match the outputs of $\fco(\Pi, C, x \concat y, k-1)$; as a result, the
element frequency vector $f^{(h)}$ from \Cref{step:pd-def-f} may have
random noise when computed using $A$ instead of using $\Pi$. If
\Cref{step:pd-comp-ph} used a fixed threshold value, then there would exist
pseudo-deterministic $\mif$ algorithms yielding element frequencies close to
this threshold, where even low-magnitude noise could affect which elements are
included in $P_h$.  Using $C$ to choose random threshold values that are
independent of the choice of $A$ ensures that most of the time, the noise has
no influence on $P_h$; ultimately, ensuring that $\fco(\Pi,C,x,k)$ and
$\fco(A,C,x,k)$ most likely produce the same output. More detail is given in
the proof of \Cref{lem:pd-stability}.

Second, the design of \Cref{alg:pd-set-finding} ensures that
$\fco(\Pi,C,\emptystream,d)$
actually produces a set of $w_d$ possible outputs. Recall the $\avoid(m,a,b)$
communication problem described in \Cref{subsec:rt-lb-overview}. The output
sets $T_y$ of size $w_{d-1}$ computed on \Cref{step:pd-compute-t}, for
each $y \in \SSD{S}{t_d}$, will have a similar structure to the input
and output sets for an $\avoid$ protocol, in that the set $T_y$ is typically
disjoint from $y$, and in that the distribution of possible $T_y$ values is
limited (through the requirement that $\Pi$ agrees with a mixture of functions
corresponding to deterministic streaming algorithms). Note that we
cannot just extract the ``most common values'' that occur in the sets $T_y$
for $y \in \SSD{S}{t_d}$, because if $t_d w_{d-1} \ll |S|$, it is
possible that these are very concentrated; for example there could exist a set
$V$ of size $w_{d-1}$ so that if $y \cap V = \emptyset$, then $T_y = V$, in
which case an algorithm satisfying the first property can only reliably
identify $w_{d-1}$ outputs. Instead, we iteratively build up a set of common
outputs, using the following observation. If the current set of common
algorithm outputs $Q_{h-1}$ is smaller than $w_d$, then the set $P_h$
containing the most common elements of $T_y$ for $y \in
\SSD{Q_{h-1}}{t_d}$ can not be entirely contained by $Q_{h-1}$: if we
did have $P_h \subseteq Q_{h-1}$, then one can show it is possible to
construct an impossibly efficient protocol for
$\avoid(|Q_{h-1}|,t_d,w_{d-1})$. Consequently, until for some value of $h$ we
have $|Q_{h}| \ge w_d$, it is possible to find a slightly larger $Q_{h+1}$.
More detail is given in the proof of \Cref{lem:pd-set-growth}.

\medskip Before beginning the formal proof, let us set the various parameters
precisely. Concretely, we set
\begin{align}
  p &= \ceil{\max\left(\sqrt{\frac{10 \rl \log (64|S|)}{3 z \log \frac{1}{2\delta}}}, 
    \frac{30 \log(64|S|)}{\log \frac{1}{2\delta}}\right) } \,.  \label{eq:pd-p-choice} \\
  d &= 1 + \floor{\frac{\rl}{18 zp}} \,, \label{eq:d-def} \\
  t_d = t_{d-1} = \cdots = t_2 &= \ceil{4 \ln 2 (zp+2)} \,. \\
\intertext{Recall also that}
  \epsilon &\le (2\delta)^{p/30} \,. \\
\intertext{We further define}
  \epsilon_k &:= w_k (64 |S|)^{k-1} \epsilon \,, \label{eq:eps-k-def}
\end{align}
which will upper-bound the probability that $\fco(A, C, x, k) \ne \fco(\Pi, C, x, k)$.
We record some useful estimations in the next lemma.

\begin{lemma}\label{lem:pd-p-properties}
  With the above choices, $t_d = \cdots = t_2 \le 9zp$,\, $t_1 \ge \rl / 2$,\, and
  $\epsilon \le 1 / (64 |S|)^d$.
\end{lemma}
\begin{proof}
  The claim about the $t_i$s holds since $z \ge 1$ and $p \ge 1$. Using this
  in \cref{eq:d-def} gives $t_1 \ge \rl - (d - 1) 9zp \ge \rl / 2$. We turn to
  proving the claim about $\epsilon$. Recall that, by design, $\epsilon \le
  (2\delta)^{p/30}$. 

  The two branches of the maximum in \cref{eq:pd-p-choice} ensure that:
  \begin{align}
    p^2 \ge \frac{10 \rl \log (64|S|)}{3 z \log \frac{1}{2\delta}} \qquad \text{and} \qquad p \ge \frac{30 \log(64|S|)}{\log \frac{1}{2\delta}}\,. \label{eq:pd-def-p}
  \end{align}
  Because $d \le \max(1, \frac{\rl}{9 z p})$, it suffices to prove that
  $\epsilon \le 1 / (64 |S|)$ and that $\epsilon \le 1 / (64 |S|)^{\rl / (9 z
  p)}$. Now,
  \begin{align*}
    \log \frac{1}{\epsilon}
    &\underset{\text{\Cref{lem:error-reduction-by-vote}}}{\ge} \frac{p}{30}
      \log \frac{1}{2\delta} \underset{\text{\cref{eq:pd-def-p}}}{\ge}
      \frac{1}{30 p} \frac{10 \rl \log(64 |S|)}{3z} = \frac{\rl}{9zp}
      \log(64|S|) \,; \\ 
    \log \frac{1}{\epsilon} 
    &\underset{\text{\Cref{lem:error-reduction-by-vote}}}{\ge} \frac{p}{30} 
      \log \frac{1}{2\delta} \underset{\text{\cref{eq:pd-def-p}}}{\ge} \log(64|S|) \,. \qedhere
  \end{align*}
\end{proof}

\subsection{Common Outputs Behave Canonically}

We now come to the central lemma in the proof, which asserts that the set of
common outputs is likely the same for the canonical function $\Pi$ as it is
for a random draw $A \sim \cB$. It also asserts two other key properties of
\fco. The lemma can be thought of as a ``proof of correctness'' of \fco.

\begin{lemma}\label{lem:pd-main-recursion}
  Let $k \in [d]$ and $x \in [n]^{t_d + \cdots + t_{k+1}}$. Then \fco
  satisfies the following properties.
  \begin{enumerate}
    \item $\Pr_{A \sim \cB, C \in_R [1,2)^{d \times \NN}}[\fco(A, C, x, k) =
    \fco(\Pi, C, x, k)] \ge 1 - \epsilon_k$.
    \item For all $C \in [1,2)^d$, the set $\fco(\Pi, C, x, k)$ is disjoint from $x$ and a subset of $S$. 
    \item For all $A \colon [n]^\rl \to [n]$ and $C \in [1,2)^d$,
    $\fco(A,C,x,k)$ outputs a set of size $w_k$.
  \end{enumerate}
\end{lemma}
\begin{proof}
  The proof is by induction on $k$, spread over the next few lemmas. The case
  $k = 1$ is handled in \Cref{lem:pd-base-case}. For the induction step, the
  heart of the argument, which invokes a communication lower bound for \avoid, is
  given in \Cref{lem:pd-set-growth}. Following this,
  \Cref{lem:pd-step-case-aux} establishes the latter two claims in the lemma
  and \Cref{lem:pd-stability} establishes the first claim.
\end{proof}

\begin{lemma}\label{lem:pd-base-case} 
  \Cref{lem:pd-main-recursion} holds for $k = 1$.
\end{lemma}
\begin{proof}
  Let $e_0,\ldots,e_{t_1}$ be the values of the variables on Lines
  \ref{step:pd-base-case-start} to \ref{step:pd-base-case-mid} of
  \Cref{alg:pd-set-finding} when $\fco(\Pi, C, x, 1)$ is called; note that
  these do not depend on $C$. For $i \in \{0,\ldots,t_1\}$, define sequences
  $s_i = \langle e_0,\ldots,e_{i-1},1,\ldots,1 \rangle$, so that $s_0 =
  \langle 1,1,\ldots,1 \rangle$, and $s_{t_1} = \langle e_0,\ldots,e_{t_1-1}
  \rangle$. If, for all $i \in \{0,\ldots,t_1\}$, we have $A(x \concat s_i) =
  \Pi(x \concat s_i)$, the value of $\fco(A, C, x, 1)$ will exactly match
  $\fco(\Pi,C,x,1)$. By a union bound,
  \begin{align*}
    \Pr_{A \sim \cB, C}[\fco(A, C, x, k) \ne \fco(\Pi, C, x, k)] 
    \le \sum_{i = 0}^{t_1} \Pr_{A \sim \cB}[A(x \concat s_i) \ne \Pi(x \concat s_i)] 
    \underset{\text{\cref{eq:pd-fn-prop}}}{\le} (t_1 + 1) \epsilon 
    = \epsilon_1 \,.
  \end{align*}
  
  Because $\Pi$ is the canonical output function for a protocol for $\mif$,
  for any $z \in [n]^\rl$, we have $\Pi(z) \notin z$. Consequently, each $e_i
  = \Pi(x \concat \langle e_0, \ldots, e_{i-1}, 1, \ldots, 1\rangle)$ is
  neither contained in $x$ nor by $\{e_0,\ldots,e_{i-1}\}$; thus
  $\{e_0,\ldots,e_{t_1}\}$ has size $t_1 + 1 = w_1$ and is disjoint from $x$.
\end{proof}

\begin{lemma}\label{lem:pd-set-growth}
  Let $x \in [n]^{t_d+\cdots+t_{k+1}}$. 
  When computing $\fco(\Pi,C,x,k)$, in the $h$th loop iteration, if $|Q_{h -
  1}| < w_k$, then $|P_h \setminus Q_{h-1}| \ge \ceil{\frac{1}{4} w_{k-1}}$.
  Consequently, the algorithm will return using \Cref{step:pd-good-exit},
  not \Cref{step:pd-bad-exit}.
\end{lemma}
\begin{proof}
  Assume for sake of contradiction that $|Q_{h - 1}| < w_k$ and $|P_h
  \setminus Q_{h-1}| \le \floor{\frac{1}{4} w_{k-1}}$. Then we can use the
  algorithm $\cA$ to implement a protocol for the one-way communication
  problem $\avoid(|Q_{h-1}|, t_k, \ceil{\frac{1}{2} w_{k-1}})$, with $\le
  \frac{1}{2}$ probability of error.
  
  We assume without loss of generality that $Q_{h-1} = [|Q_{h-1}|]$; if not,
  relabel coordinates so that this holds. In the protocol, after Alice is
  given a subset $W \subseteq Q_{h-1}$ with $|W| = t_k$, they construct a
  sequence $v = x \concat \sort(W)$ in $[n]^{t_d + \cdots + t_k}$.  Then Alice
  uses public randomness to instantiate an instance $E$ of $\cA$;  inputs the
  sequence $v$ to $E$; and sends the new state of $E$ to Bob, using a $z
  p$-bit message. As Bob shares the public randomness,  they can use this
  state to evaluate the output of the algorithm on any continuation of the
  stream. In particular, Bob can evaluate the algorithm for any possible
  suffix, to produce a function $\tilde{A}_{x \concat \sort(W)} :
  [n]^{t_{k-1}+\cdots+t_1} \rightarrow [n]$; Bob then samples a random $C \in
  [1,2)^{d \times \NN}$, and computes $V = \fco(\tilde{A}_{x \concat
  \sort(W)}, C, k-1)$, which is a subset of $S$. If $|V \cap Q_{h-1}| \ge
  \ceil{\frac{1}{2} w_{k-1}}$, Bob outputs the smallest $\ceil{\frac{1}{2}
  w_{k-1}}$ entries of $V \cap Q_{h-1}$. Otherwise, Bob outputs an arbitrary
  set of size $\ceil{\frac{1}{2} w_{k-1}}$.
  
  First, we observe that for any value of $\sort(W)$, the distribution of
  $\tilde{A}_{x \concat \sort(W)}$
  is exactly the same as the distribution of $A_{x \concat \sort(W)}$, when $A$ is drawn from $\cB$;
  this follows because for a fixed setting of the oracle random string of the algorithm,
  it behaves deterministically.
  
  Applying \Cref{lem:pd-main-recursion} at $k-1$, we observe that for any $W \in \binom{Q_{h-1}}{t_k}$,
  \begin{align*}
    \Pr[\fco(\tilde{A}_{x \concat \sort(W)}, C, k-1) = \fco(\Pi, C, x \concat \sort(W), k-1)] 
    \ge 1 - \epsilon_{k-1} \ge \frac{3}{4} \,.
  \end{align*}
  Furthermore, we are guaranteed that $\fco(\Pi, C, x \concat \sort(W), k-1)$
  has size $w_{k-1}$ and is disjoint from $W$.
  
  We now bound the probability, over a uniformly random $y \in
  \SSD{Q_{h-1}}{t_k}$, that $|\fco(\Pi, C, x \concat y, k-1) \cap
  Q_{h-1}| < \ceil{\frac{1}{2} w_{k-1}}$. Define $T_y = \fco(\Pi, C, x \concat y, k-1)$ and, for each $j\in S$,
  $f_j^{(h)}$, as in \Cref{alg:pd-set-finding}. In particular, we have:
  \begin{align}
    \Pr_{y,C}\left[|T_y \cap Q_{h-1}| < \ceil{\frac{1}{2} w_{k-1}} \right]
      &= \Pr_{y,C}\left[|T_y \setminus Q_{h-1}| > \floor{\frac{1}{2} w_{k-1}} \right] \nonumber \\
      &\le \Pr_{y,C}\left[|T_y \setminus P_h \setminus Q_{h-1}| > \floor{\frac{1}{2} w_{k-1}} - \floor{\frac{1}{4} w_{k-1}} \right] \label{eq:pd-lb-minor-set-drop} \\
      &\le \Pr_{y,C}\left[|T_y \setminus P_h| \ge \frac{1}{4} w_{k-1} \right] \,. \nonumber
  \end{align}
  (The inequality on \cref{eq:pd-lb-minor-set-drop} follows since we
  assumed $|P_h \setminus Q_{h-1}| \le \floor{\frac{1}{4} w_{k-1}}$.) Note that:
  \begin{align}
    \sum_{j \notin P_h} f_j^{(h)} = \sum_{y \in \SSD{Q_{h-1}}{t_k}} |T_y \setminus P_h| \ge \frac{1}{4} w_{k-1} \left|\left\{y \in \SSD{Q_{h-1}}{t_k} : |T_y \setminus P_h| \ge \frac{1}{4} w_{k-1}\right\}\right| \label{eq:pd-lb-count-part} \,.
  \end{align}
  Using the fact that $y$ is uniformly distributed over $\SSD{Q_{h-1}}{t_k}$, gives:
  \begin{align*}
    \Pr_{y,C}\left[|T_y \setminus P_h| \ge \frac{1}{4} w_{k-1} \right]
      &= \EE_C \frac{\left|\left\{y \in \SSD{Q_{h-1}}{t_k} : |T_y \setminus P_h| \ge \frac{1}{4} w_{k-1}\right\}\right|}{\binom{|Q_{h-1}|}{t_k}} \\ 
      &\le \EE_C \frac{\sum_{j \notin P_h} f_j^{(h)}}{\frac{1}{4} w_{k-1} \binom{|Q_{h-1}|}{t_k}} && \text{by \cref{eq:pd-lb-count-part}}\\ 
      &\le \EE_C \frac{(|S|-|P_h|) \frac{C_{k,h} w_{k-1}}{16 |S|} \binom{|Q_{h-1}|}{t_k} }{\frac{1}{4} w_{k-1} \binom{|Q_{h-1}|}{t_k}} &&\text{by definition of $P_h$} \\
      &= \EE_C \frac{C_{k,h}}{4} \frac{|S|-|P_h|}{|S|} \le \frac{1}{2} \,. &&\text{since $|P_h|\ge 0$, $C_{k,h} \le 2$} 
  \end{align*}
  Thus the probability that $|T_y \cap Q_{h-1}| < \ceil{\frac{1}{2} w_{k-1}}$
  holds is $\le 1/2$. Since Bob only gives an incorrect output when this
  happens or when $\fco(\tilde{A}_{x \concat \sort(W)}, C, k-1) \ne
  \fco(\Pi, C, x \concat \sort(W), k-1)$, it follows by a union bound that
  the total failure probability is $\le \frac{1}{2} + \frac{1}{4} \le
  \frac{3}{4}$.
  
  Consequently, the protocol implementation has $\le \frac{3}{4}$ error when inputs are drawn from the uniform distribution over $\binom{Q_{h-1}}{t_k}$; by \Cref{lem:avoid-lb}, we 
  obtain a lower bound on the required message length, giving
  \begin{align*}
    z p > \frac{t_k \ceil{\frac{1}{2} w_{k-1}}}{|Q_{h-1}| \ln 2} + \log(1 - 3/4) \ge \frac{t_k w_{k-1}}{|Q_{h-1}| \cdot 2 \ln  2} - 2 \,.
  \end{align*}
  Rearranging this slightly and using integrality of $|Q_{h-1}|$ gives:
  \begin{align*}
    |Q_{h-1}| \ge \ceil{\frac{t_k w_{k-1}}{2 \ln  2 (zp+2) }} = \ceil{\frac{\ceil{4 \ln 2 (zp + 2)}}{2 \ln  2 (zp+2) } w_{k-1}} \ge 2 w_{k-1} = w_k \,,
  \end{align*}
  but as $|Q_{h-1}| < w_k$, this implies $w_k < w_k$, which is a
  contradiction; this proves that the assumption $|P_h \setminus Q_{h-1}| \le
  \frac{1}{4} w_{k-1}$ must have been invalid.
  
  Finally, we observe that since, in each iteration of the loop on Lines \ref{step:pd-for-loop-start} to \ref{step:pd-for-loop-end},
  $|Q_h| = |Q_{h-1} \cup P_{h}| = |Q_{h-1}| + |P_{h} \setminus Q_{h-1}| \ge |Q_{h-1}| + \ceil{\frac{1}{4} w_{k-1}}$,
  and we initially have $|Q_0| = w_{k-1}$, the size of $Q_h$ (assuming we haven't returned yet)
  must be $\ge w_{k-1} (1 + h / 4)$. By the last loop iteration (with $h=4$), we will have $|Q_h| \ge 2 w_{k-1} = w_k$.
\end{proof}

\begin{lemma}\label{lem:pd-step-case-aux}
  For $k > 1$, $x \in [n]^{t_d+\cdots+t_{k+1}}$, $\fco(\Pi, C, x, k)$ is
  disjoint from $x$ and a subset of $S$; and for all $A, C, k$,
  $\fco(A,C,x,k)$ outputs a set of size $w_k$.
\end{lemma}
\begin{proof}
  By \Cref{lem:pd-main-recursion} at $k - 1$, the sets $T_{A,x \concat y}$
  chosen on \Cref{step:pd-compute-t} are always subsets of $S$ and
  disjoint from $x \concat y$, and hence disjoint from $x$. Per
  \Cref{lem:pd-set-growth}, \textsc{FindCommonOutputs} will return a subset of
  $Q_h$ using \Cref{step:pd-good-exit}, where $h$ is the last loop
  iteration number. Each element of $Q_h$ was either in $T_{A,x \circ \langle
  1,2,\ldots,t_k\rangle}$ (and hence also in $S$) or was in $P_{h'}$ for some
  $h' \le h$. Note that $P_{h'}$ only contains integers $j$ for which
  $f_j^{(h')} > 0$; i.e., which were contained in one of the sets $(T_{A,x
  \concat y})_{y \in \SSD{Q_{h'-1}}{t_k}}$, and are thereby also in $S$.
  Consequently, the set returned is contained in $S$, which implies $|S| \ge
  w_k$.

  Calls to $\fco(A,C,x,k)$ will either output through
  \Cref{step:pd-good-exit} (where the size of the set has been checked by
  the pseudocode) or through \Cref{step:pd-bad-exit} (where a subset of
  size $w_k$ must exist, because we know $|S| \ge w_k$).
\end{proof}

\begin{lemma}\label{lem:pd-stability}
  For $k > 1$, and all $x \in [n]^{t_d + \cdots + t_{k+1}}$, 
  \begin{align*}
    \Pr_{A \sim \cD, C}[\fco(A, C, x, k) \ne \fco(\Pi, C, x, k)] \le \epsilon_k \,.
  \end{align*}
\end{lemma}
\begin{proof}
  The proof of the lemma follows from the observation that, when computing
  $\fco(A,C,x,k)$, even if a fraction of the recursive calls to $\fco(A,C,x
  \concat y,k-1)$ produced incorrect outputs, the values for $Q_0$ and $(P_h)_{h \ge 1}$ will likely match those computed when $\fco(\Pi,C,x,k)$ is called.
  
  Henceforth, we indicate variables from the computation of $\fco(\Pi,C,x,k)$ without a tilde,
  and variables from the computation of $\fco(A,C,x,k)$ with a tilde. For example, $f_j^{(h)}$
  is computed using $B = \Pi$, while $\tilde{f}_j^{(h)}$ is computed using $B = A$. We also define
  \begin{align*}
    \hf_j^{(h)} &= \left|\left\{ y \in \SSD{Q_{h-1}}{t_k} : j \in T_{A,y} \right\}\right| \\
    \hP_h &= \left\{j \in S : \hf_j^{(h)} \ge \frac{C_{k,h} w_{k-1}}{16 |S|} \left|\binom{Q_{h-1}}{t_k}\right| \right\} \,;
  \end{align*}
  that is, $\hf_j^{(h)}$ and $\hP_h$ are the values that would be computed by $\fco(A,C,x,k)$ if the set $Q_{h-1}$ was used instead of the set $\tQ_{h-1}$.
  
  Say $\fco(\Pi, C, x, k)$ returns from the loop at iteration $h^\star$. The output of
  $\fco(A, C, x, k)$ will equal $\fco(\Pi, C, x, k)$ if $Q_0 = \tQ_0$ and for all $h \in [h^\star]$,
  we have $P_h = \hP_h$. (If this occurs, then as $Q_0 = \tQ_0$, $\hP_1 = \tP_1$, so
  $Q_1 = Q_0 \cup P_1 = \tQ_0 \cup \tP_1 = \tQ_1$, and as $Q_1 = \tQ_1$, $\hP_2 = \tP_2$, and so on.)  By \Cref{lem:pd-main-recursion} at $k-1$, the probability that
  $Q_0 \ne \tQ_0$ is $\le \epsilon_{k-1}$. Consider a specific $h \in [h^\star]$;
  the only way in which $\hP_h \ne P_h$ is if there is some $j \in S$ for which $f_j^{(h)}$ and $\hf_j^{(h)}$
  are on opposite sides of the threshold $\frac{C_{k,h} w_{k-1}}{16 |S|} |\binom{Q_{h-1}}{t_k}|$.
  
  Let $\lambda_{h}$ be the random variable indicating the fraction of $y \in \SSD{Q_{h-1}}{t_k}$
  for which $T_{A,x \concat y} \ne T_{\Pi,x \concat y}$. Note that the values $T_{A,x \concat y}$ are functions of
  the random variable $A$ and of $C_{k',h}$ for $k' < k, h \in \NN$; in particular $T_{A,x\circ y}$ is independent of $(C_{k,h})_{h \in \NN}$. By \Cref{lem:pd-main-recursion} at $k-1$,
  $\Pr[T_{A,x\circ y} \ne T_{\Pi,x\circ y}] \le \epsilon_{k-1}$, which implies $\EE \lambda_h \le \epsilon_{k-1}$.
  
  Fix a particular setting of $A$ and $(C_{k',h})_{k' < k, h \in \NN}$. Since each set $T_{A,x\circ y}$
  contributes $1$ unit to each of $w_{k-1}$ variables $\hf_j^{(h)}$:
  \begin{align*}
    \sum_{j \in S} \left|f_j^{(h)} - \hf_j^{(h)}\right| \le w_{k-1} \left|\left\{y \in \SSD{Q_{h-1}}{t_k} : T_{A,x \concat y} \ne T_{\Pi,x\circ y}\right\}\right| = w_{k-1} \lambda_h \binom{Q_{h-1}}{t_k} \,.
  \end{align*}
  Let $F$ be the set of possible values in $[1,2)$ for $C_{k,h}$ for which $P_h \ne \hP_h$;
  this is a union of intervals corresponding to each pair $\left(f_j^{(h)}, \hf_j^{(h)}\right)$,
  for $j \in S$. A given value $c$ is bad for $j$ if 
  \begin{align*}
    f_j^{(h)} < \frac{c w_{k-1}}{16 |S|} \binom{|Q_{h-1}|}{t_k} \le \hf_j^{(h)} \,; \qquad \text{equivalently:} \qquad c \in \left(\frac{16 |S| f_j^{(h)} }{w_{k-1} \binom{|Q_{h-1}|}{t_k}}, \frac{16 |S| \hf_j^{(h)}}{w_{k-1} \binom{|Q_{h-1}|}{t_k}}\right] \,,
  \end{align*}
  and similarly in the case where $\hf_j^{(h)} < f_j^{(h)}$. The measure of $F$ is:
  \begin{align*}
    \le \sum_{j \in S} \frac{16 |S|}{w_{k-1} \binom{|Q_{h-1}|}{t_k}} |\hf_j^{(h)} - f_j^{(h)}| \le \frac{16 |S|}{w_{k-1}} w_{k-1} \lambda_h = 16 |S| \lambda_h \,.
  \end{align*}
  This upper bounds the probability that $C_{k,h} \in F$ and $P_h \ne \hP_h$. We then have:
  \begin{align*}
    \Pr[P_h \ne \hP_h] &= \EE_{A, (C_{k',h})_{k' < k}} \Pr[ C_{k,h} \in F ] 
      \le \EE_{A, (C_{k',h})_{k' < k}} (16 |S| \lambda_h) = 16 |S| \epsilon_{k-1} \,.
  \end{align*}
  
  By a union bound, the probability that $Q_0 \ne \tQ_0$ or $P_h \ne \hP_h$ for any $h \le h^\star$ is
  \begin{align*}
    \le \epsilon_{k-1} + h^\star 16 |S| \epsilon_{k-1} \le (1 + 4 \cdot 16 |S|) \epsilon_{k-1} \le \frac{64 |S| w_k}{w_{k-1}} \epsilon_{k-1} \,.
  \end{align*}
  Thus $\Pr[\fco(A, C, x, k) \ne \fco(\Pi, C, x, k)] \le \frac{64 |S| w_k}{w_{k-1}} \epsilon_{k-1} = \epsilon_k$.
\end{proof}

We have thus completed the proof of \Cref{lem:pd-main-recursion}, establishing
the key properties of \fco. It is time to use them to derive our lower bound.

\subsection{Obtaining a Pseudo-Deterministic Lower Bound}

\begin{lemma}\label{lem:pd-z-lb} 
  We have:
  \begin{align*}
    z \ge \frac{\rl}{\log\frac{2|S|}{\rl}} \min\left(\frac{1}{36},
    \frac{\log\frac{1}{2\delta}}{17280 \log(64|S|) \log\frac{2|S|}{\rl}} \right) \,.
  \end{align*}
\end{lemma}
\begin{proof}
  A consequence of \Cref{lem:pd-main-recursion} is that $\fco(\Pi, C, d)$ will
  output a set of size $w_d$ which is a subset of $S$. This shows that $|S|
  \ge w_d$.  Now, from the definition of $w_d$, it follows that
  \begin{align*}
    |S| \ge w_d = 2^{d-1} (t_1 + 1) > 2^{d-1} t_1 \ge 2^{d-1} \frac{\rl}{2} \qquad \implies \qquad \log \frac{2 |S|}{\rl} \ge d - 1 = \floor{\frac{\rl}{18 z p}} \,.
  \end{align*}
  Since $|S| \ge \rl + 1$, the left hand side $\log \frac{2 |S|}{\rl} > 1$, so using the inequality $x/2 \le \max(\floor{x}, 1)$ gives:
  \begin{align*}
    \frac{\rl}{36 z p} \le \log \frac{2 |S|}{\rl} \qquad \implies \qquad z \ge \frac{\rl}{36 \log \frac{2|S|}{\rl}} \cdot \frac{1}{p} \,.
  \end{align*}
  Next, we expand the definition of $p$ (see \cref{eq:pd-p-choice}), eliminating the ceiling using the inequality $\ceil{x} \le \max(1, 2 x)$:
  \begin{align*}
    z &\ge \frac{\rl}{36 \log \frac{2|S|}{\rl}} \min\left(1, \frac{1}{2} \sqrt{ \frac{3 z \log \frac{1}{2\delta}}{10 \rl \log(64 |S|) } }, \frac{1}{2} \cdot \frac{\log \frac{1}{2\delta}}{30 \log(64 |S|) }\right) \,.
  \end{align*}
  We have two cases: if the left or right side of the minimum is smallest, then:
  \begin{align}
    z \ge \min\left( \frac{\rl}{36 \log \frac{2|S|}{\rl}}, \frac{\rl \log \frac{1}{2\delta}}{2160 \log \frac{2|S|}{\rl} \log(64|S|)} \right) \,, \label{eq:pd-min-lr}
  \end{align}
  while otherwise, if the center is smallest, we get:
  \begin{align*}
    z^2 \ge  \frac{1}{4} \frac{\rl^2}{(36 \log \frac{2|S|}{\rl})^2} \cdot \frac{3 z \log \frac{1}{2\delta}}{10 \rl \log(64 |S|) } \qquad \implies \qquad z \ge \frac{\rl \log \frac{1}{2 \delta}}{17280 (\log \frac{2|S|}{\rl})^2 \log(64|S|)} \,.
  \end{align*}
  As $\log\frac{2|S|}{\rl} \ge 1$, this is smaller than the right minimum
  branch of \cref{eq:pd-min-lr}, so the common lower bound for all cases is:
  \begin{align*}
    z &\ge \min\left( \frac{\rl}{36 \log \frac{2|S|}{\rl}}, \frac{\rl \log \frac{1}{2 \delta}}{17280 (\log \frac{2|S|}{\rl})^2 \log(64|S|)} \right) \,. \qedhere
  \end{align*}
\end{proof}

\begin{lemma}\label{lem:pd-s-le-2zp1}
  We have $|S| < 2^{z+1}$.
\end{lemma}
\begin{proof}
  For each $a \in S$, let $x_a \in \Pi^{-1}(a)$. One can use $\cA$ to provide a randomized $\le\delta$-error, $z$-bit encoding of the elements in $S$. Using public randomness, encoder and decoder choose the oracle random string for $\cA$. Each $a \in S$ is encoded by sending $x_a$ to $\cA$ and outputting the algorithm state $\sigma$. To decode, given a state $\sigma$, one evaluates the output of $\cA$ at state $\sigma$. Using the minimax principle, one can prove that the randomized encoding requires $\ge \log( (1-\delta) |S| )$ bits of space, which implies $2^z \ge (1-\delta) |S|$. Since $\delta \le \frac{1}{3}$, it follows $s \le \frac{3}{2} 2^{z} < 2^{z+1}$.
\end{proof}

We now establish the main result.

\begin{theorem} 
\label{thm:pd-lb-intro}
  Pseudo-deterministic $\delta$-error random oracle algorithms for $\mif(n,\rl)$ require
  \begin{align*}
    \Omega\left(\min\left(\frac{\rl}{\log \frac{2n}{\rl}} + \sqrt{\rl}, \frac{\rl \log \frac{1}{2\delta}}{(\log \frac{2n}{\rl})^2 \log n} + \left(\rl \log \frac{1}{2\delta}\right)^{1/4}\right)\right)
  \end{align*}
  bits of space when $\delta \le \frac{1}{3}$. In particular, when $\delta = 1/\poly(n)$ and $\rl = \Omega(\log n)$, this is:
  \begin{align*}
    \Omega\left(\frac{\rl}{(\log\frac{2 n}{\rl})^2} + \left(\rl \log n\right)^{1/4}\right) \,.
  \end{align*}
\end{theorem}
\begin{proof}[Proof sketch]
  Using \Cref{lem:pd-s-le-2zp1} and the fact that $S \subseteq [n]$, we obtain
  $|S| \le \min(n, 2^{z+1})$. The theorem follows by combining this bound with
  the inequality of \Cref{lem:pd-z-lb}, and for each of four cases
  corresponding to different branches of $\min$ and $\max$, solving to find a
  lower bound on $z$. The full proof with calculations is given in
  \Cref{subsec:pd-lb-mech}.
\end{proof}

\begin{remark}
  For $\delta \le 2^{-\rl}$, \Cref{thm:pd-lb-intro} reproduces the
  deterministic algorithm space lower bound for $\mif(n,\rl)$ from
  \cite{Stoeckl23} within a constant factor.
\end{remark}

\subsection{Implications for Adversarially Robust Random Seed Algorithms}

The following result, paraphrased from \cite{Stoeckl23} relates the random
seed adversarially robust space complexity with the pseudo-deterministic space
complexity.
\begin{theorem}[\cite{Stoeckl23}]\label{thm:ext-rs-to-pd}
  Let $S^{\text{PD}}_{1/3}(n, \rl)$ give a space lower bound for a
  pseudo-deterministic algorithm for $\mif(n,\rl)$ with error $\le 1/3$. Then
  an adversarially robust random seed algorithm with error $\delta \le
  \frac{1}{6}$, if it uses $z$ bits of space, must have $z \ge
  S^{\text{PD}}_{1/3}(n, \floor{\frac{\rl}{2 z + 2}} )$.
\end{theorem}

\begin{theorem} 
\label{cor:rs-lb-intro}
  Adversarially robust random seed algorithms for $\mif(n,\rl)$ with error $\le \frac16$ require space:
  \begin{align*}
    \Omega\left(\frac{\rl^2}{n} + \sqrt{{\rl}/{(\log n)^3}} + \rl^{1/5}\right) \,.
  \end{align*}
\end{theorem}

This follows by combining \Cref{thm:ext-rs-to-pd}, \Cref{thm:pd-lb-intro}, \Cref{lem:ext-robust-lb}, and performing some algebra; a proof is given in \Cref{subsec:pd-lb-mech}.

%% file: appendix.tex
\appendix
\section{Appendix}\label{sec:appendix}

\subsection{Proofs of Useful Lemmas}\label{subsec:useful-proofs}

Here we provide proofs of the results in \Cref{subsec:useful-lemmas} for which we haven't found an external source.

\begin{proof}[Proof of \Cref{lem:azumanoff}]
  The proof is modeled off that in \cite{Stoeckl23}, which only addresses
  one direction. It is a straightforward blend of standard proofs of the Chernoff bound and of Azuma's inequality.

  First, the $\ge$ direction. Choose, with foresight, $z = \ln(1 + \alpha)$.
  \begin{align*}
    \Pr&\left[ \sum_{i=1}^{t} X_i \ge (1 + \alpha) \sum_{i =1}^{t} p_i \right] \\
      &= \Pr\left[ \exp( z \sum_{i=1}^{t} X_i) \ge \exp(z (1 + \alpha) \sum_{i =1}^{t} p_i) \right] \\
      &\le \frac{\EE \exp( z \sum_{i=1}^{t} X_i)}{\exp(z (1 + \alpha) \sum_{i =1}^{t} p_i)} \\
      &\le \frac{\EE[e^{z X_1} \EE[e^{z X_2} \ldots \EE[e^{z X_t} | X_1,\ldots,X_{t-1}] \ldots | X_1] ]}{\exp(z (1 + \alpha) \sum_{i =1}^{t} p_i)} \,.
  \end{align*}
  The innermost term $\EE[e^{z X_t} | X_1,\ldots,X_{t-1}]$ is, by convexity of $e^z$, $\le p_t e^z + (1-p_t) \le e^{p_t (e^z - 1)}$; after applying this upper bound, we can factor it out and bound the $X_{t-1}$ term, and so on. Thus we continue the chain of inequalities to get:
  \begin{align*}
    \le \frac{\exp\left((e^z - 1)\sum_{i=1}^{t} p_i\right)}{\exp\left(z (1 + \alpha) \sum_{i =1}^{t} p_i\right)} = \exp\left(- ((1+\alpha) \ln (1+\alpha) - \alpha) \sum_{i =1}^{t} p_i\right) \,.
  \end{align*}
  
  For the other direction, set $z = \ln(1-\alpha)$, which is $< 0$.
  This time, $\EE[e^{z X_t} | X_1,\ldots,X_{t-1}] \le p_t e^z + (1-p_t)$ because $p_t$ is a \emph{lower bound} for $\EE[X_t | X_1,\ldots,X_{t-1}]$, and $z$ is negative. That $p_t e^z + (1-p_t) \le e^{p_t (e^z - 1)}$ still holds for negative $z$, so:
  \begin{align*}
    \Pr\left[ \sum_{i=1}^{t} X_i \le (1 - \alpha) \sum_{i =1}^{t} p_i \right] &= \Pr\left[ \exp(z\sum_{i=1}^{t} X_i) \ge \exp(z (1 - \alpha) \sum_{i =1}^{t} p_i) \right] \\
      &\le \ldots \le \frac{\exp\left((e^z - 1)\sum_{i=1}^{t} p_i\right)}{\exp\left(z (1 - \alpha) \sum_{i =1}^{t} p_i\right)} \\
      &= \exp\left(- ((1-\alpha) \ln (1-\alpha) + \alpha) \sum_{i =1}^{t} p_i\right) \,. \qedhere
  \end{align*}
\end{proof}

\begin{proof}[Proof of \Cref{lem:error-reduction-by-vote}]
  For each $i \in [p]$, let $Y_i$ be the random indicator variable for the event that $X_i \ne v$. Let $\alpha = \frac{1}{2 \delta} - 1$.The probability that $v$ is not the most common element can be bounded by the probability that it is the not the majority element; by a Chernoff bound, this is:
  \begin{align*}
    \Pr[\sum_{i\in[p]} Y_i \ge \frac{1}{2} p] &= \Pr[\sum_{i\in[p]} Y_i \ge (1 + \alpha) \delta p] \le \exp\left( - ((1+\alpha) \ln(1+\alpha) - \alpha) \delta p\right) \\
      &\le \exp\left( - 0.073 ((1+\alpha) \ln(1+\alpha) ) \delta p\right) \qquad \text{since $\alpha \ge 1/6$} \\
      &\le \exp\left( - \frac{0.073 }{2\delta} \ln\frac{1}{2\delta} \delta p\right) = \left(2\delta\right)^{0.036 p} \,. \qedhere
  \end{align*}
\end{proof}

\subsection{Mechanical Proofs for \texorpdfstring{\Cref{sec:rt-lb}}{Random Tape Lower Bound}}\label{subsec:rt-lb-mech}

\begin{proof}[Proof of \Cref{lem:rt-lb-calc-multiple}]
  Let $x$ be the left hand side of Eq. \ref{eq:rt-lb-max-min-lb}.
  The left branch of the $\min(\cdot,\cdot)$ terms in Eq. \ref{eq:rt-lb-max-min-lb} is actually unnecessary. For any integer $\lambda \ge 2$, say that
  \begin{align*}
    \left(\frac{\rl^{\lambda + 1}}{n}\right)^{\frac{2}{\lambda^2 + 3 \lambda - 2}} = \max_{k \in \NN} \left(\frac{\rl^{k + 1}}{n}\right)^{\frac{2}{k^2 + 3 k - 2}} \,.
  \end{align*}
  Then in particular,
  \begin{align*}
    \left(\frac{\rl^{\lambda + 1}}{n}\right)^{\frac{2}{\lambda^2 + 3 \lambda - 2}} \ge \left(\frac{\rl^{(\lambda -1) + 1}}{n}\right)^{\frac{2}{(\lambda-1)^2 + 3 (\lambda -1)- 2}} \ge     \left(\frac{\rl^{\lambda}}{n}\right)^{\frac{2}{\lambda^2 + \lambda - 4}} \,,
  \end{align*}
  which implies
  \begin{align*}
    n^{2 \lambda + 2} = n^{(\lambda^2 + 3 \lambda - 2) - (\lambda^2 + \lambda - 4)} \ge \rl^{\lambda \cdot (\lambda^2 + 3 \lambda - 2) - (\lambda + 1) \cdot (\lambda^2 + \lambda - 4)} = \rl^{\lambda^2 + \lambda + 4} \,,
  \end{align*}
  hence we have $\rl \le n^{\frac{2 \lambda +2}{\lambda^2 + \lambda + 4}}$. 

  On the other hand, we have
  \begin{align*}
    \rl^{1/\lambda} \ge \left(\frac{\rl^{\lambda + 1}}{n}\right)^{\frac{2}{\lambda^2 + 3 \lambda - 2}} \qquad \iff \qquad n^\lambda \le \rl^{(\lambda + 1) \lambda - \frac{\lambda^2+3\lambda - 2}{2}} = \rl^{\frac{\lambda^2 - \lambda + 2}{2}} \,,
  \end{align*}
  so the left branch of the $\min(\cdot,\cdot)$ in Eq. \ref{eq:rt-lb-max-min-lb} is only smaller when $\rl \ge n^{\frac{2 \lambda}{\lambda^2 - \lambda + 2}}$. As
  \begin{align*}
    \frac{2 \lambda +2}{\lambda^2 + \lambda + 4} \le \frac{2 \lambda}{\lambda^2 - \lambda + 2} \,,
  \end{align*}
  for all $\lambda \ge 1$, it follows that the left branch of the $\min(\cdot,\cdot)$ in Eq. \ref{eq:rt-lb-max-min-lb} is only smaller than the right when the entire term is not the maximum. Thus
  \begin{align*}
    x &\ge \max_{k \in \NN} \left(\frac{\rl^{k + 1}}{n}\right)^{\frac{2}{k^2 + 3 k - 2}} \,.
  \end{align*}

  To get a looser but more easily comprehensible lower bound, we note that $\max_{k \in \NN} \log \left(\frac{\rl^{k + 1}}{n}\right)^{\frac{2}{k^2 + 3 k - 2}}$ is piecewise linear and convex in $\log \rl$. Consequently,
  we can lower bound it using the convex function $C \frac{(\log \rl)^2}{\log n}$, where $C$ is the maximum value which satisfies the inequality at all ``corner points'' of $\max_{k \in \NN} \log \left(\frac{\rl^{k + 1}}{n}\right)^{\frac{2}{k^2 + 3 k - 2}}$. These corner points occur precisely at values of $\log \rl$ where, for some $k \ge 2$, we have:
  \begin{align*}
    \left(\frac{\rl^{k + 1}}{n}\right)^{\frac{2}{k^2 + 3 k - 2}} = \left(\frac{\rl^{(k-1) + 1}}{n}\right)^{\frac{2}{(k-1)^2 + 3 (k-1) - 2}} \,.
  \end{align*}
  Rearranging this gives:
  \begin{align*}
    \frac{\log n}{\log \rl} = \frac{(k+1) ((k-1)^2 + 3 (k-1) - 2) - (k-1+1) (k^2 + 3 k - 2)}{((k-1)^2 + 3 (k-1) - 2) - (k^2 + 3 k - 2)} = \frac{k^2 + k + 4}{2k + 2} \,.
  \end{align*}
  so the corners occur at $\rl = n^{\frac{2k + 2}{k^2 + k + 4}}$; and at such $\rl$, we have
  \begin{align*}
    \log \left(\frac{\rl^{k + 1}}{n}\right)^{\frac{2}{k^2 + 3 k - 2}} 
    &= \left(\frac{2}{k^2 + 3 k - 2} \cdot ((k+1) \frac{2 k + 2}{k^2 + k + 4} - 1)\right) \log n = \frac{2}{k^2 + k + 4} \log n \\
      &= \frac{(\log \rl)^2}{\log n} \frac{2}{k^2 + k + 4} \left(\frac{\log n}{\log \rl}\right)^2 = \frac{(\log \rl)^2}{\log n} \frac{2}{k^2 + k + 4} \left(\frac{k^2 + k + 4}{2k + 2}\right)^2 \\
      &= \frac{1}{2} \frac{(\log \rl)^2}{\log n} \frac{k^2 + k + 4}{(k+1)^2} \ge \frac{15}{32} \frac{(\log \rl)^2}{\log n} \,.
  \end{align*}
  The function $\frac{k^2 + k + 4}{(k+1)^2}$ has derivative $\frac{k-7}{(k+1)^3}$ and is minimized exactly at $k = 7$, where it has value $\frac{15}{16}$. Consequently, the value $C = \frac{15}{32}$ is the best possible.
\end{proof}

\subsection{Mechanical Proofs for \texorpdfstring{\Cref{sec:rt-ub}}{Random Tape Upper Bound}}\label{subsec:rt-ub-mech}

\begin{proof}[Proof of \Cref{lem:rt-alg-params}]
  First, we handle the case where $\ceil{\log \rl} <  \floor{2 \frac{\log(n/4)}{\log 16 \rl}}$. Then $d = \ceil{\log \rl} \le \floor{2 \frac{\log(n/4)}{\log 16 \rl}} - 1$, and $\alpha = 2$. Note that $\frac{\rl}{2^{d - 1}} \ge 1$ since $2^{d - 1} \le  2^{\ceil{\rl} - 1} \le 2\rl/2 = \rl$.
  \begin{align*}
      \prod_{i=2}^{d} b_i &= \ceil{\frac{\rl}{2^{d - 1}}} 2^{d - 2} \ge \frac{\rl}{2} = \frac{\rl}{\alpha} \\
      \prod_{i=2}^{d} b_i &=  \ceil{\frac{\rl}{2^{d - 1}}} 2^{d - 2} \le 2 \frac{\rl}{2} \le \frac{4 \rl}{\alpha} \,.
  \end{align*}
  Since $b_d = \ceil{\frac{\rl}{2^{d - 1}}} = \ceil{\frac{2 \rl}{2^{\ceil{\log \rl}}}} \le 2$,
  \begin{align*}
    \prod_{i \in [d]} w_i &= (16 \rl) \prod_{i = 2}^{d} \prod_{j=i}^{d} b_j
      \le (16 \rl) \prod_{i = 2}^{d} 2^{d - i + 1} = (16 \rl) 2^{d (d-1) / 2} \\
      &\le (16 \rl) (2^{\ceil{\log \rl}})^{(d-1) / 2} \le (16 \rl) (2 \rl)^{(d-1) / 2} \\
      &\le (16 \rl) (2 \rl)^{(\floor{2 \frac{\log (n/4)}{\log 16\rl}}-2) / 2} \\
      &\le (16 \rl) (2 \rl)^{\frac{\log (n/4)}{\log 16\rl} - 1} \\
      &\le (16 \rl)^{\frac{\log (n/4)}{\log 16\rl}} = \frac{n}{4} \le n \,.
  \end{align*}

  Second, we consider the case where $d = \floor{2 \frac{\log (n/4)}{\log(16\rl)}}$. Because $n \ge 64 \rl$, $d \ge 2$, and so
  \begin{align*}
    d = \floor{2 \frac{\log (n/4)}{\log 16 \rl}} \ge \frac{2}{3} \cdot 2 \frac{\log (n/4)}{\log 16 \rl} = \frac{\log (n/4)}{\frac{3}{4} \log 16 \rl} \ge \frac{\log (n / 4)}{\log 4 \rl} \,.
  \end{align*}
  The second inequality used that $\frac{3}{4} (4 + \log \rl) \le (2 + \log \rl)$ for $\rl \ge 4$. Consequently,
  \begin{align*}
    \alpha = \left(\frac{(4 \rl)^{d}}{n/4}\right)^{\frac{2}{d(d-1)}} \ge \left(\frac{(4 \rl)^{\frac{\log (n / 4)}{\log 4 \rl}}}{n/4}\right)^{\frac{2}{d(d-1)}} =  \left(\frac{n/4}{n/4}\right)^{\frac{2}{d(d-1)}} = 1 \,.
  \end{align*}
  We now prove Eq. \ref{eq:rt-ub-param-ge-l}. Because $\prod_{i=2}^{d-1} b_i \ge \alpha^{d - 2}$,
  \begin{align*}
    \prod_{i=2}^{d} b_i = \ceil{\frac{\rl}{\alpha^{d - 1}}} \prod_{i=2}^{d-1} b_i \ge \frac{\rl}{\alpha \prod_{i=2}^{d-1} b_i} \cdot \prod_{i=2}^{d-1} b_i = \frac{\rl}{\alpha} \,.
  \end{align*}
  For Eq. \ref{eq:rt-ub-param-le-n}, we observe that
  \begin{align*}
    d \le 2 \frac{\log(n/4)}{\log(16 \rl)}  \qquad \implies \qquad 16\rl \le (n/4)^{2/d} \qquad \implies \qquad \rl \ge \alpha^{d - 1} =\frac{(4 \rl)^2}{(n/4)^{2/d}} \,,
  \end{align*}
  and thus $\rl / \alpha^{d - 1} \ge 1$, so $b_d = \ceil{\rl / \alpha^{d -1}} \le 2 \rl / \alpha^{d - 1}$. Then since  $\prod_{i=2}^{d-1} b_i \le 2 \alpha^{d - 2}$,
  \begin{align*}
    \prod_{i=2}^{d} b_i \le \frac{2 \rl}{\alpha^{d - 1}} \prod_{i=2}^{d} b_i \le \frac{4 \rl}{\alpha \prod_{i=2}^{d-1} b_i} \cdot \prod_{i=2}^{d-1} b_i \le \frac{4 \rl}{\alpha} \,.
  \end{align*}
  Finally, we prove Eq. \ref{eq:rt-ub-param-le-n}. As noted above,
  \begin{align*}
    b_d \le \frac{2 \rl}{\alpha^{d-1}} \le \frac{4 \rl}{\alpha \prod_{i=2}^{d-1} b_j} \,.
  \end{align*}
  Applying this fact to bound the left hand side of Eq. \ref{eq:rt-ub-param-le-n} gives:
  \begin{align*}
    \prod_{i \in [d]} w_i &= (16 \rl) \prod_{i = 2}^d \prod_{j=i}^d b_j = 16 \rl (b_d)^d \prod_{i = 2}^{d-1} \prod_{j=i}^{d-1} b_j \\
      &\le 16 \rl \left(\frac{4 \rl}{\alpha \prod_{i=2}^{d-1} b_j} \right)^{d-1} \prod_{i = 2}^{d-1} \prod_{j=i}^{d-1} b_j \\
      &\le \frac{4 \cdot (4 \rl)^d}{\alpha^{d-1}} \frac{1}{\prod_{i = 2}^{d-1} \prod_{j=2}^{i-1} b_j} \\
      &\le \frac{4 \cdot(4 \rl)^d}{\alpha^{d-1}} \frac{1}{\alpha^{(d-1)(d-2)/2}} \qquad \text{since $\prod_{j=2}^{d-1} b_j \ge \alpha^{d-2}$ and $b_2 \ge b_3 \ge \ldots \ge b_{d-1}$} \\
      &= \frac{4 \cdot(4 \rl)^d}{\alpha^{d(d-1)/2}} = \frac{4 \cdot(4 \rl)^{d}}{\frac{(4 \rl)^d}{n / 4}} = n \,. \qedhere
  \end{align*}
\end{proof}

\begin{proof}[Proof of \Cref{lem:rt-ub-space-usage}]
  \Cref{alg:rt} only stores two types of data: for each $i\in[d]$, the vectors $L_i \in [w_i]^{b_i}$, and the vectors $x_i \in \{0,1\}^{b_i}$. These can be stored using $b_i \log w_i$ and $b_i$, bits respectively, for a total of:
  \begin{align*}
    \sum_{i \in [d]} b_i \log(2 w_i) &\le b_1 \log(2 w_1) + \sum_{i = 2}^{d} b_i \log(2 w_i) \\
      &\le b_1 \log(32 \rl) + \sum_{i = 2}^{d} b_i \log(2 \prod_{j=i}^{d} b_i)
      \le \sum_{i=1}^{d} b_i \log(32 \rl) \,,
  \end{align*}
  since by Eq. \ref{eq:rt-ub-param-le-2l}, $\prod_{j=i}^{d} b_i \le 4 \rl$. 
  
  We now observe that $b_d \le \ceil{\alpha}$. If $d = \ceil{\log \rl}$,
  then $\alpha = b_2 = \ldots = b_{d - 1} = 2$ and $b_d = \ceil{\rl/\alpha^{d - 1}} \le 2$. On the other hand, if $d = \floor{2 \frac{\log(n/4)}{\log(16\rl)}}$, then $\alpha = \frac{(4 \rl)^{2/(d-1)}}{(n/4)^{2/(d(d-1))}}$. We have:
  \begin{align*}
    (4\rl)^{d + 1} = (4\rl)^{\floor{2 \frac{\log(n/4)}{\log(16\rl)}} + 1} \ge (4\rl)^{2 \frac{\log(n/4)}{\log(16\rl)}} = (n/4)^2 \,,
  \end{align*}
  which implies
  \begin{align*}
    \alpha = \frac{(4 \rl)^{2/(d-1)}}{(n/4)^{2/(d(d-1))}} \ge \frac{(4 \rl)^{2/(d-1)}}{(4\rl)^{(d+1)/(d(d-1))}} = (4\rl)^{(2 - \frac{d+1}{d}) \cdot \frac{1}{d - 1}} = (4\rl)^{\frac{1}{d}} \,.
  \end{align*}
  Consequently,
  \begin{align*}
    b_d = \ceil{\frac{\rl}{\alpha^{d - 1}}} = \ceil{\frac{1}{4} \frac{4 \rl}{\alpha^{d - 1}}} \le \ceil{\frac{1}{4} (4 \rl)^{1/d}} \le (4 \rl)^{1/d} \le \alpha \,.
  \end{align*}
  
  With the bound on $b_d$, and the fact that $\alpha \ge 1$ in both cases, and that $d \le \ceil{\log \rl}$ we obtain:
  \begin{align*}
    \sum_{i \in [d]} b_i &\le \min(\rl + 1, \ceil{8 \alpha} + \ceil{3 \log 1/\delta}) + (d - 1) \ceil{\alpha} \\
      &\le \min(\rl, \ceil{3 \log 1/\delta}) + (7 + d) 2 \alpha \\
      &\le \min(\rl, \ceil{3 \log 1/\delta}) + 32 \log \rl \ceil{\frac{(4 \rl)^{2/(d-1)}}{(n/4)^{2/(d(d-1))}}} \,.
  \end{align*}
  Multiplying this last quantity by $\log(32 \rl)$ gives a space bound.

  To obtain a much weaker, but somewhat more comprehensible upper bound on $\alpha$, when $d = \floor{2 \frac{\log(n/4)}{\log(16\rl)}}$, we note that:
  \begin{align*}
    \max_{\lambda \in \NN \cap [2,\infty)} \log \frac{(4 \rl)^{2/(\lambda-1)}}{(n/4)^{2/(\lambda(\lambda-1))}} &\le \log \max_{\lambda \in \RR \cap [2,\infty)} \left(\frac{2}{\lambda - 1} \log(4\rl) - \frac{2}{\lambda(\lambda-1)} \log(n/4) \right) \\
    &\le \log \left(2 \log(4\rl) \max_{\lambda \in \RR \cap [2,\infty)} \left(\frac{1}{\lambda - 1} - \frac{1}{\lambda (\lambda -1)} \frac{\log(n/4)}{\log(4\rl)} \right)\right) \,.
  \end{align*}
  Let $\gamma = \frac{\log(n/4)}{\log(4\rl)}$; this is $\ge 1$. Let $f(x) = \frac{1}{x - 1} (1 - \frac{\gamma}{x})$. We will now prove that $\max_{x \ge 2} f(x) \le \frac{1}{2\gamma}$. We note that when $x = 2$, we have:
  \begin{align*}
    f(2) = 1 - \frac{\gamma}{2} \le \frac{1}{2\gamma} \,.
  \end{align*}
  Checking the other endpoint, we have:
  \begin{align*}
    \lim_{x \rightarrow \infty} \frac{1}{x - 1} (1 - \frac{\gamma}{x}) = 0 \,.
  \end{align*}
  Since $f(x)$ is differentiable on $[2,\infty)$, if it has a maximum other than at the endpoints, then it will occur when $\frac{d}{dx} f(x) = 0$. Solving this equation, we obtain:
  \begin{align*}
    \frac{d}{dx} f(x) = -\frac{1}{\left(x-1\right)^{2}}+\frac{\gamma\left(2x-1\right)}{\left(x\left(x-1\right)\right)^{2}} = -\frac{1}{\left(x-1\right)^{2}}\left[1-\frac{\gamma\left(2x-1\right)}{x^{2}}\right] = 0 \,,
  \end{align*}
  which is true iff $x^2 = \gamma(2x - 1)$. The solutions to the quadratic equation are
  \begin{align*}
    x = \gamma-\sqrt{\gamma\left(\gamma-1\right)} \qquad \text{and} \qquad  x = \gamma+\sqrt{\gamma\left(\gamma-1\right)} \,.
  \end{align*}
  Since $\gamma \ge 1$, the $-$ branch has $x \le 1$, which is not in $[2,\infty)$. The $+$ branch is only in $[2,\infty)$ if $\gamma \ge \frac{4}{3}$. The value of $f(x)$ in this case is:
  \begin{align*}
    f(\gamma+\sqrt{\gamma\left(\gamma-1\right)}) &= \frac{1}{\gamma+\sqrt{\gamma\left(\gamma-1\right)} - 1} \left(1 - \frac{\gamma}{\gamma+\sqrt{\gamma\left(\gamma-1\right)}} \right) \\
      &= \frac{\sqrt{\gamma(\gamma-1)}}{2 \gamma - 1 + 2 \sqrt{\gamma(\gamma-1)}} \\
      &\le \frac{1}{4 \sqrt{\gamma(\gamma-1)}}  &&\hspace{-4cm}\text{(since $\sqrt{\gamma(\gamma-1)} \le 2 \gamma - 1$ for all $\gamma \ge 1$)}
      \\
      &\le \frac{1}{2 \gamma} \,. &&\hspace{-4cm}\text{(since $\gamma \le 2 \sqrt{\gamma(\gamma-1)}$ for all $\gamma \ge \frac{4}{3}$)}
  \end{align*}
  Thus, if $f(x)$ does have a maximum in $[2,\infty)$, it is $\le \frac{1}{2\gamma}$. We conclude that $f(x) \le \frac{1}{2\gamma}$ in all cases. This proves:
  \begin{align*}
    \log \alpha \le \max_{\lambda \in \NN} \log \frac{(4 \rl)^{2/(\lambda-1)}}{(n/4)^{2/(\lambda(\lambda-1))}} 
      &\le \log\left(2 \log(4 \ell) \frac{1}{2 \frac{\log(n/4)}{\log(4\rl)}} \right) \le \log \frac{  \log(4 \ell)^2}{\log(n/4)} \,. \qedhere
  \end{align*}
\end{proof}

\subsection{Mechanical Proofs for \texorpdfstring{\Cref{sec:pd-lb}}{Pseudo-Deterministic and Random Seed Lower Bounds}}\label{subsec:pd-lb-mech}

\begin{proof}[Proof of \Cref{thm:pd-lb-intro}]
  By \Cref{lem:pd-z-lb}, we have:
  \begin{align}
    z \ge \min\left(
      \frac{\rl}{36 \log\frac{2 |S|}{\rl}},
      \frac{\rl}{17280 \log\frac{2 |S|}{\rl}} \cdot \frac{\log(1/2\delta)}{\log(64 |S|) \log\frac{2 |S|}{\rl}}
    \right) \,.\label{eq:pd-lb-restate}
  \end{align}
  By \Cref{lem:pd-s-le-2zp1}, $|S| \le \min(n, 2^{z+1}) \le \min(n, 4^z)$. We will apply this
  inequality to each branch of the minimum in Eq. \ref{eq:pd-lb-restate}. First, say that
  the left part of the minimum is larger than the right. Then $z \ge \rl / (36 \log\frac{2 |S|}{\rl})$.
  Applying $|S| \le n$ and $|S| \le 4^z$, this implies:
  \begin{align*}
    z &\ge \frac{\rl}{36 \log\frac{2 n}{\rl}} \,, \qquad \text{and} \\
    z &\ge \frac{\rl}{36 \log\frac{2 \cdot 4^z}{\rl}} \ge \frac{\rl}{36 \cdot 3 z} \qquad \implies \qquad z \ge \sqrt{\frac{\rl}{108}} \,.
  \end{align*}
  Thus:
  \begin{align}
    z \ge \max\left(\frac{\rl}{36 \log\frac{2 n}{\rl}}, \sqrt{\frac{\rl}{108}} \right) \,.\label{eq:pd-lb-case-l}
  \end{align}

  Next, say that the right side of the minimum in Eq. \ref{eq:pd-lb-restate} is larger.
  Then applying $|S| \le n$ and $|S| \le 4^z$ to that side, we get:
  \begin{align*}
    z &\ge \frac{\rl \log(1/2\delta)}{17280 (\log\frac{2 n}{\rl})^2 \log(64 n)} \,, \qquad \text{and} \\
    z &\ge \frac{\rl \log(1/2\delta)}{17280 (\log\frac{2 \cdot 4^z}{\rl})^2 \log(64 \cdot 4^z)} \ge \frac{\rl \log(1/2\delta)}{17280 \cdot 3^2 \cdot 8 z^3} \qquad \implies \qquad z \ge \left(\frac{\rl \log(1/2\delta)}{1244160}\right)^{1/4} \,.
  \end{align*}
  Thus:
  \begin{align}
    z \ge \max\left(
      \frac{\rl \log(1/2\delta)}{17280 (\log\frac{2 n}{\rl})^2 \log(64 n)},
      \left(\frac{\rl \log(1/2\delta)}{1244160}\right)^{1/4}
    \right) \,.\label{eq:pd-lb-case-r}
  \end{align}

  The minimum of the lower bounds from Eqs. \ref{eq:pd-lb-case-l} and \ref{eq:pd-lb-case-r}
  holds in all cases, so:
  \begin{align*}
    z \ge \min\left(\max\left(\frac{\rl}{36 \log\frac{2 n}{\rl}}, \sqrt{\frac{\rl}{108}}\right), \max\left(\frac{\rl \log(1/2\delta)}{17280 (\log\frac{2 n}{\rl})^2 \log(64n)},\left(\frac{\rl \log(1/2\delta)}{1244160}\right)^{1/4}\right)\right) \,.
  \end{align*}
\end{proof}

\begin{proof}[Proof of \Cref{cor:rs-lb-intro}]
  The lower bound from \Cref{thm:pd-lb-intro} for $\mif(n,t)$ with error $\delta = 1/3$, showing constants, is:
  \begin{align}
    \max\left(\frac{t \log(3/2)}{17280 (\log\frac{2 n}{t})^2 \log(64n)},\left(\frac{t \log(3/2)}{1244160}\right)^{1/4}\right) \,.\label{eq:lb-rs-pd-expr}
  \end{align}
  If the space used by an algorithm, $z$, satisfies $z \ge (\rl-1)/2$, then we tautologically have a lower bound of $(\rl-1)/2$. Otherwise, we have $2z+2 \le \rl$.
  Applying \Cref{thm:ext-rs-to-pd} gives, for the left branch of the $\max$ in Eq. \ref{eq:lb-rs-pd-expr}, with $t = \floor{\frac{\rl}{2 z + 2}} \ge \frac{1}{8 z}$:
  \begin{align*}
    z &\ge \floor{\frac{\rl}{2 z + 2}} \frac{\log(3/2)}{17280 (\log\frac{2 n}{t})^2 \log(64n)} \ge \frac{\rl}{8 z} \frac{1/2}{17280 (\log (2 n))^2 \log(64n)} \,,
  \end{align*}
  which implies
  \begin{align*}
    z \ge \sqrt{\frac{\rl}{276480 (\log (2 n))^2 \log(64n)}} \ge  \sqrt{\frac{\rl}{7741440 (\log n)^3}} \,.
  \end{align*}
  For the right branch of Eq. \ref{eq:lb-rs-pd-expr}, we obtain:
  \begin{align*}
    z \ge \left(\floor{\frac{\rl}{2 z + 2}} \frac{\log(3/2)}{1244160}\right)^{1/4} \ge \left(\frac{\rl}{8 z} \frac{1}{2 \cdot 1244160}\right)^{1/4} \,,
  \end{align*}
  which implies:
  \begin{align*}
    z^{5/4} \ge \left(\frac{\rl}{19906560}\right)^{1/4} \qquad \implies \qquad z \ge \left(\frac{\rl}{19906560}\right)^{1/5} \,.
  \end{align*}
  Combining the two lower bounds, gives:
  \begin{align}
    z \ge \max\left(\sqrt{\frac{\rl}{7741440 (\log n)^3}}, \left(\frac{\rl}{19906560}\right)^{1/5}\right) = \Omega\left(\sqrt{\frac{\rl}{(\log n)^3}} + \rl^{1/5}\right) \,. \label{eq:cor-lb-base}
  \end{align}
  This lower bound is everywhere smaller than $(\rl-1)/2$, so it is compatible with the case in which $z \ge (\rl-1)/2$. 
  
  Taking the maximum of Eq. \ref{eq:cor-lb-base} and the known random oracle lower bound for $\mif(n,\rl)$ algorithms in the static setting, \Cref{lem:ext-robust-lb}, gives:
  \begin{align*}
    z &= \Omega\left(\frac{\rl^2}{n} + \sqrt{\frac{\rl}{(\log n)^3}} + \rl^{1/5}\right)  \,. \qedhere
  \end{align*}
  
\end{proof}